\documentclass[10pt,a4paper]{article}
\usepackage[utf8]{inputenc}
\usepackage[english]{babel}

\usepackage{amsmath}
\usepackage{amsfonts}
\usepackage{amssymb}

\numberwithin{equation}{section}

\usepackage{amsthm}
\newtheorem{theorem}{Theorem}

\newtheorem{definition}{Definition}

\usepackage[colorlinks,citecolor=blue,urlcolor=blue,linkcolor=blue]{hyperref}

\usepackage[left=2cm,right=2cm,top=2cm,bottom=2cm]{geometry}

\usepackage{authblk}

\usepackage{feynmp}
\newcommand{\pd}{\partial}
\newcommand{\tr}{\operatorname{tr}}
\newcommand{\okappa}[1]{\mathcal{O}\left( \kappa^{#1} \right)}
\newcommand{\gs}{\mathit{g}_\text{s}}

\newcommand{\ket}[1]{\left\lvert #1 \right\rangle}
\newcommand{\bra}[1]{\left\langle #1 \right\rvert}

\title{FeynGrav and Recent Progress in Computational Perturbative Quantum Gravity}
\author[1]{\href{https://orcid.org/0000-0001-7099-0861}{B. Latosh} \thanks{ \href{mailto:latosh.boris@ibs.re.kr}{latosh.boris@ibs.re.kr} } }
\affil[1]{Particle Theory  and Cosmology Group, Center for Theoretical Physics of the Universe, Institute for Basic Science (IBS), Daejeon, 34126, Korea}
\date{CTPU-PTC-24-01}

\begin{document}

\maketitle

\begin{abstract}
    The article reviews recent progress in computational quantum gravity caused by the framework that efficiently computes Feynman's rules. The framework is implemented in the FeynGrav package, which extends the functionality of the widely used FeynCalc package. FeynGrav provides all the tools to study quantum gravitational effects within the standard model. We review the framework, provide the theoretical background for the efficient computation of Feynman rules, and present the proof of its completeness. We review the derivation of Feynman rules for general relativity, Horndeski gravity, Dirac fermions, Proca field, electromagnetic field, and $SU(N)$ Yang-Mills model. We conclude with a discussion of the current state of the FeynGrav package and discuss its further development.
\end{abstract}

\tableofcontents

\section{Introduction}

Perturbative quantum gravity is one of several approaches to quantum gravity. It exists in the effective field theory paradigm and describes gravitational phenomena at energies below the Planck scale \cite{Burgess:2003jk, Donoghue:1994dn, Calmet:2013hfa}. To a certain extent, the theory associates gravity with small metric perturbations propagating about the flat spacetime. The Planck scale arises naturally as the characteristic scale of small perturbations. Consequently, the theory can perturbatively treat perturbations with small amplitudes. However, the structure of the theory is more sophisticated and extends beyond this simple setup, as we discuss below.

Perturbation theory operates with propagating degrees of freedom, polarisation operators, and interaction operators describing the interaction between different degrees of freedom. On the practical ground, defining a theory means explicitly providing these components. The graviton propagator and the polarisation operators were derived and studied in classical papers \cite{Fierz:1939ix,Weinberg:1964cn,Weinberg:1964ev}. The derivation of the interaction rules was the most challenging part of the theory. Because of the effective nature of the theory, it generates an infinite set of interaction operators suppressed by different powers of the same gravitational coupling. For a given level of the perturbation theory, only a finite number of operators contribute to a given matrix element. However, the number of operators grows with each perturbation order. Consequently, calculations become incredibly challenging without an algorithm describing how to obtain the interaction rules at any given order of perturbation theory.

The interaction rules for the first two orders of perturbation theory are published in the literature \cite{DeWitt:1967yk,DeWitt:1967ub,DeWitt:1967uc,Sannan:1986tz}. These rules were sufficient for one-loop calculations and allowed for significant advancement in perturbative quantum gravity \cite{DeWitt:1967yk,DeWitt:1967ub,DeWitt:1967uc,Grisaru:1975ei,Goroff:1985th,Donoghue:1994dn,Akhundov:1996jd,Bjerrum-Bohr:2002gqz,Holstein:2006bh,Jakobsen:2020diz,Prinz:2020nru}. Nonetheless, the general algorithm producing the interaction rules at a given order remained unknown. Several research directions were focused on creating such an algorithm. The most well-known ones proposed to use Ward–Takahashi identities \cite{DeWitt:1967yk,DeWitt:1967ub,DeWitt:1967uc} and to bootstrap gravity \cite{Gupta:1954zz,Kraichnan:1955zz,Deser:1969wk,Padmanabhan:2004xk,Deser:2009fq}. However, to our knowledge, they never resulted in a practical tool. An additional problem was the complexity of the interaction rules. The paper \cite{DeWitt:1967uc} derived interaction rules for three and four graviton vertices (see also \cite{Sannan:1986tz} for the mentioned paper contained a misprint). In that particular parameterisation, the three graviton vertex contains $171$ terms while the four graviton vertex contains $2850$. While it is exceptionally challenging to manually operate with interaction terms having hundreds and thousands of terms, the contemporary computational packages can (relatively) easily manipulate such expressions.

The recent publications \cite{Latosh:2022ydd,Latosh:2023zsi} created the desirable algorithm. It was implemented in a computational package ``FeynGrav'' created at the base of the widely used ``FeynCalc'' \cite{Mertig:1990an,Shtabovenko:2016sxi,Shtabovenko:2020gxv}. This article reviews the constructed theoretical framework, touching upon its application and perspective for further development. For the discussion of the theoretical framework, we focus on the original papers \cite{Latosh:2022ydd,Latosh:2023zsi}. We briefly discuss how FeynGrav can significantly simplify the derivation of the well-known results \cite{Donoghue:1994dn,Akhundov:1996jd,Latosh:2020jyq,Arbuzov:2020pgp,Latosh:2021usy} and is used to obtain new results \cite{Latosh:2022hrf,Latosh:2023ueg,Latosh:2023cxm,Latosh:2023xej}. We conclude with a discussion of the future development of the package.

The structure of this paper is as follows. Section \ref{Effecitve_Field_Theory} discusses the theoretical background behind the perturbative quantum gravity. Namely, we discuss the role of the perturbative quantum gravity together with the limits of its applicability. We discuss the theory in the context of renormalisation and argue, supporting the growing consensus, that its effective nature leaves no room for any ``renormalisation issues''. Section \ref{Computational_Tools} discusses the practical tools to develop the theoretical and computational framework. We discuss how to factorise an action describing gravity or a model coupled to gravity in a way most suitable for optimal computations. We separate some fundamental structures defining perturbative expansions and derive recursive relations that significantly improve their computation. Section \ref{Feynman_Rules} devoted to the derivation of the Feynman rules for general relativity, scalar field, Horndeski theory, Dirac fermions, Proca field, electromagnetic field, and $SU(N)$ Yang-Mills model. Section \ref{FeynGrav_Section} discusses the FeynGrav package, its structure, and rules implemented in the existing version. We present a few explicit examples of perturbative calculations made with FeynGrav. Section \ref{Summary_Section} summarises the material presented in this paper and highlights perspective further development.

\section{Perturbative quantum gravity}\label{Effecitve_Field_Theory}

We divide the discussion of perturbative quantum gravity into two parts. The first part explores the theory's mathematical and technical aspects, while the second focuses on its physical content. We start with the technical aspects as they help to explain the premises behind its physical interpretation.

Perturbative quantum gravity is a quantum theory of small perturbation of metric $h_{\mu\nu}$ existing about the flat spacetime with the Minkowski metric $\eta_{\mu\nu}$. To preserve the canonical mass dimension of the field variable, one must introduce the gravitational coupling $\kappa$ with mass dimension $-1$. One defines the spacetime metric as a combination of the background metric and perturbations:
\begin{align}\label{the_perturbative_expansion}
    g_{\mu\nu} = \eta_{\mu\nu} + \kappa\, h_{\mu\nu} .
\end{align}
The gravitational coupling $\kappa$ is associated with the Newton constant $G_\text{N}$:
\begin{align}
    \kappa^2 = 32\,\pi\,G_\text{N} .
\end{align}

It is important to highlight that the perturbative expansion \eqref{the_perturbative_expansion} is finite in $\kappa$ and should not be interpreted as a truncation of an infinite series. Despite \eqref{the_perturbative_expansion} is finite, it spawns an infinite series in $\kappa$ for the inverse metric:
\begin{align}
    g^{\mu\nu} = \eta^{\mu\nu} - \kappa\, h^{\mu\nu} + \kappa^2\, h^{\mu\sigma}\,h_\sigma{}^\nu + \okappa{3} .
\end{align}
Consequently, all quantities involving the inverse metric are infinite series in $\kappa$. 

Let us look at the two most important examples. The first example is the Christoffel symbol. If the symbol has all lower indices, then it is a finite expression linear in $\kappa$:
\begin{align}
    \Gamma_{\alpha\mu\nu} =& \cfrac12\,\left[ \pd_\mu g_{\nu\alpha} + \pd_\nu g_{\mu\alpha} - \pd_\alpha g_{\mu\nu} \right] = \cfrac{\kappa}{2}\,\left[ \pd_\mu h_{\nu\alpha} + \pd_\nu h_{\mu\alpha} - \pd_\alpha h_{\mu\nu} \right].
\end{align}
On the contrary, the standard Christoffel symbol with a single upper index is an infinite series since it involves the inverse metric:
\begin{align}
    \begin{split}
        \Gamma^\alpha_{\mu\nu} =& \cfrac12\,g^{\alpha\beta} \left[ \pd_\mu g_{\nu\beta} + \pd_\nu g_{\mu\beta} - \pd_\beta h_{\mu\nu} \right] \\
        =& \cfrac{\kappa}{2} \left[ \eta^{\mu\nu} - \kappa\, h^{\mu\nu} + \kappa^2\, h^{\mu\sigma}\,h_\sigma{}^\nu + \okappa{3} \right] \left[ \pd_\mu h_{\nu\beta} + \pd_\nu h_{\mu\beta} - \pd_\beta h_{\mu\nu} \right].
    \end{split}
\end{align}

The second example is the volume factor $\sqrt{-g}$. The metric determinant itself is a finite expression in $\kappa$ (see \cite{Prinz:2020nru}), but the volume factor is an infinite series because of the square root:
\begin{align}
    \sqrt{-g} = 1 + \cfrac{\kappa}{2} \, \eta^{\mu\nu} h_{\mu\nu} - \cfrac{\kappa^2}{4} \left( h_{\mu\sigma} h^\sigma{}_\nu - \cfrac12\,\eta^{\mu\nu} \eta^{\alpha\beta} h_{\mu\nu} h_{\alpha\beta} \right) + \okappa{3}.
\end{align}

The action of any gravity model is an infinite series in $\kappa$ since it includes the volume factor $\sqrt{-g}$ for the coordinate invariance. Moreover, a typical action of a gravity model includes invariants constructed from the Riemann tensor. Christoffel symbols constitute the Riemann tensor, introducing additional infinite series to the theory. From the technical point of view, this is the reason why the action of any gravity model contains an infinite number of interaction operators.

The path integral formalism provides a way to construct a perturbative quantum theory. The generating functional $\mathcal{Z}$ provides a way to calculate matrix elements:
\begin{align}
    \mathcal{Z} = \int\mathcal{D}[g] \exp\Big[i\, \mathcal{A}[g]\Big] .
\end{align}
Here $\mathcal{A}$ is the microscopic action of the model. Within the perturbative approach, one uses \eqref{the_perturbative_expansion} in the path integral, expanding it in an infinite series. Metric \eqref{the_perturbative_expansion} performs a simple shift of integration variables by a constant value:
\begin{align}
    \begin{split}
        \mathcal{Z} =& \int\mathcal{{D}}[ \eta + \kappa \,h ] \exp\Big[ i \,\mathcal{A}[\eta + \kappa\, h] \Big]\\
        =& \int\mathcal{D}[h] \exp\Bigg[ i \, \mathcal{A}[\eta] + i\, \cfrac{\delta \mathcal{A} }{\delta g_{\mu\nu}} \Bigg|_{g=\eta} \, \kappa h_{\mu\nu} + i\, \cfrac{\delta^2 \mathcal{A}}{\delta g_{\mu\nu} \delta g_{\alpha\beta}} \Bigg|_{g=\eta} \kappa^2\,h_{\mu\nu} h_{\alpha\beta} \\
        & \hspace{90pt}+ i\, \cfrac{\delta^3 \mathcal{A}}{\delta g_{\mu\nu} \delta g_{\alpha\beta} \delta g_{\rho\sigma}} \Bigg|_{g=\eta} \kappa^3 h_{\mu\nu} h_{\alpha\beta} g_{\rho\sigma} + \okappa{4} \Bigg].
    \end{split}
\end{align}

The first term of this expansion is irrelevant. The multiplier does not depend on the integration variables $h_{\mu\nu}$ and factorises out of the path integral. In turn, while calculating a matrix element, the multiplier will be cancelled entirely because of the normalisation.

It is safe to assume that the second term shall also vanish. The term disappears if the flat spacetime makes the microscopic action stationary, which occurs when the classical field equations derived from the microscopic action $\mathcal{A}$ have flat spacetime as a solution. Consequently, the discussed approach only applies to some gravity models with further modifications. One can always extend it to the case of an arbitrary background spacetime \cite{Christensen:1979iy,Birrell:1982ix}. Further, we will discuss the physical content of the theory and argue that it may not be necessary. Let us postpone the discussion, assume that the term vanishes, and proceed with the technical discussion.

The other terms of the expansion do not vanish but describe the propagation of perturbations and their interactions. It is helpful to present $\mathcal{Z}$ in the following form, making its structure explicit.  
The action naturally contains $\kappa^{-2}$ factor, similar to the Einstein-Hilbert action. Therefore, the generating functional takes the following form: 
\begin{align}
    \mathcal{Z} = \int\mathcal{D}[h] \exp\Bigg[  - \cfrac{i}{2}\, h_{\mu\nu} \mathcal{O}^{\mu\nu\alpha\beta} \square h_{\alpha\beta}  + i \, \kappa \, \widehat{\mathcal{V}}_{(3)}^{\mu_1\nu_1\mu_2\nu_2\mu_3\nu_3} h_{\mu_1\nu_1} h_{\mu_2\nu_2} h_{\mu_3\nu_3} + \okappa{2} \Bigg].
\end{align}
Here $\mathcal{O}^{\mu\nu\alpha\beta}$ is the differential operator defining the structure of the graviton propagator well-established in the literature \cite{Fierz:1939ix,Accioly:2000nm}. Operator $\widehat{\mathcal{V}}_(3)$ is the differential operator defining the structure of interaction between three perturbations. The term $\okappa{2}$ includes all other interaction terms, each suppressed by a different power of the gravitational coupling.

The standard prescription provides a way to calculate a given matrix element. One shall introduce a formal external current $J^{\mu\nu}$ linearly coupled to the metric perturbation:
\begin{align}
    \begin{split}
        \mathcal{Z}[J] =& \int\mathcal{D}[h] \exp\Big[i\,\mathcal{A}[\eta + \kappa\, h] + i \, h_{\mu\nu} J^{\mu\nu} \Big] \\
        =& \int\mathcal{D}[h] \exp\Bigg[  - \cfrac{i}{2}\, h_{\mu\nu} \mathcal{O}^{\mu\nu\alpha\beta} \square h_{\alpha\beta}  + i \, h_{\mu\nu} J^{\mu\nu}  + i \, \kappa \, \widehat{\mathcal{V}}_{(3)}^{\mu_1\nu_1\mu_2\nu_2\mu_3\nu_3} h_{\mu_1\nu_1} h_{\mu_2\nu_2} h_{\mu_3\nu_3} + \okappa{2} \Bigg]
    \end{split}
\end{align}
The path integral reduces to the Gaussian integral, which can be solved explicitly:
\begin{align}
    \begin{split}
        \mathcal{Z}[J] =& \exp\Bigg[ i \, \kappa \, \widehat{\mathcal{V}}_{(3)}^{\mu_1\nu_1\mu_2\nu_2\mu_3\nu_3} \frac{\delta}{\delta J^{\mu_1\nu_1}}\,\frac{\delta}{\delta J^{\mu_2\nu_2}}\,\frac{\delta}{\delta J^{\mu_3\nu_3}}\, + \okappa{2} \Bigg]  \exp\Bigg[   \cfrac{i}{2}\, J^{\mu\nu} \mathcal{O}^{-1}_{\mu\nu\alpha\beta} \square J^{\alpha\beta} \Bigg] .
    \end{split}
\end{align}
Finally, each matrix element links to a variational derivative of the functional:
\begin{align}
    \begin{split}
        &\bra{0} h_{\mu_1\nu_1} \cdots h_{\mu_N \nu_N} \ket{0} = \cfrac{1}{\mathcal{Z}} \,\int\mathcal{D}[h] h_{\mu\nu} \cdots h_{\mu_N\nu_N} \exp\Big[i\,\mathcal{A} \Big]\\
        &= \cfrac{1}{\mathcal{Z}[J] } \,\int\mathcal{D}[h] h_{\mu\nu} \cdots h_{\mu_N\nu_N} \exp\Big[i\,\mathcal{A} + i\, h_{\rho\sigma} J^{\rho\sigma} \Big] \Bigg|_{J=0}\\
        &= \cfrac{ \cfrac{\delta}{\delta J^{\mu_1\nu_1}}\cdots \cfrac{\delta}{\delta J^{\mu_N\nu_N}} \exp\left[ i \, \kappa \widehat{\mathcal{V}}_{(3)}^{\rho_1\sigma_1\rho_2\sigma_2\rho_3\sigma_3} \cfrac{\delta}{\delta J^{\rho_1\sigma_1}} \cfrac{\delta}{\delta J^{\rho_2\sigma_2}} \cfrac{\delta}{\delta J^{\rho_3\sigma_3}} + \okappa{2} \right]  \exp\left[   \frac{i}{2} J^{\mu\nu} \mathcal{O}^{-1}_{\mu\nu\alpha\beta} \square J^{\alpha\beta} \right]  }{ \exp\left[ i \, \kappa \, \widehat{\mathcal{V}}_{(3)}^{\rho_1\sigma_1\rho_2\sigma_2\rho_3\sigma_3} \cfrac{\delta}{\delta J^{\rho_1\sigma_1}}\,\cfrac{\delta}{\delta J^{\rho_2\sigma_2}}\,\cfrac{\delta}{\delta J^{\rho_3\sigma_3}}\, + \okappa{2} \right]  \exp\left[   \frac{i}{2}\, J^{\mu\nu} \mathcal{O}^{-1}_{\mu\nu\alpha\beta} \square J^{\alpha\beta} \right]  }  \Bigg|_{J=0}.
    \end{split}
\end{align}

The following comments are due. First and foremost, gravity is a gauge theory, so operator $\mathcal{O}^{-1}$ does not exist. In analogy with other gauge theories, one uses the Faddeev-Popov ghosts \cite{Faddeev:1967fc,Faddeev:1973zb} or the BRST technique \cite{Becchi:1974xu,Becchi:1974md,Becchi:1975nq,Tyutin:1975qk} to obtain the propagator. We discuss this issue in detail in another section since it is irrelevant to the principal discussion. Secondly, the discussed scheme is the standard scheme used within the conventional quantum field theory. In that sense, perturbative quantum gravity is the most general attempt to study quantum gravitational effects while remaining within the standard quantum field theory without any modifications. Lastly, the discussed calculations pose no principal difficulty but a severe computational challenge. The issue is not specific to quantum gravity itself. Even within renormalisable quantum field theories, for instance, within quantum chromodynamics, with each order of perturbation theory, calculations become increasingly complicated and involve more and more contributions \cite{Kleiss:1988ne}. In that sense, perturbative quantum gravity is not an exception to this rule but rather presents a limiting case with the faster-growing complexity.

Let us turn to a discussion of the physical content of the theory. Three physical premises behind the perturbative approach within the quantum field theory are crucial for the perturbative quantum gravity. Firstly, the perturbation theory only requires propagators, interaction operators, and polarisation operators to generate perturbative expansions. Secondly, propagating and external states are not necessarily free states of a theory. Thirdly, the Poincare group is necessary for the perturbative theory within the conventional quantum field theory. Below, we elaborate on these points.

Firstly, the core of a perturbative theory is propagating degrees of freedom, their interactions, and polarisation operators. Indeed, when given these objects, one can use the standard Feynman graphs technique to generate any perturbative expansion up to any desired order. From the physical point of view, these objects provide all the information about the content of the theory. The polarisation operators describe the states of quantum fields in the initial and final states. Propagators describe how perturbations propagate in the spacetime, while the interaction operators describe their interactions. The perturbative expansion approximates a matrix element as a series of consecutive interactions. These series allow one to define the $S$ matrix, the evolution operator of the theory, describing it in the most general way possible. Therefore, one requires no other objects to create and implement the perturbation theory.

Secondly, one shall distinguish propagating, external and free states. Propagating states are those states described by propagators of the theory. External states are those states described by polarisation operators. Free states are those states which can be separated and captured or measured in a (conceivable) physical experiment. The following considerations show the need to make such a distinguishment.

It is well established that propagating states differ from those present in initial and final states. The Faddeev-Popov ghosts provide the best example \cite{Faddeev:1973zb}. These ghosts are non-physical since they enter a theory due to the mathematical redefinition of the path integral integration volume. Consequently, these ghosts do not enter initial or final states and cannot be separated or detected in any physical experiment. Nevertheless, they must be presented in a perturbation theory to ensure the gauge invariance and unitarity of the theory. This example alone shows the need to distinguish between states that propagate and states that can exist in the initial and final states.

The neutrino mixing mechanism \cite{Pontecorvo:1957qd,Maki:1962mu} explicitly shows the difference between propagating states and states that can be registered. The mechanism explicitly states that there are three neutrino states $\nu_1$, $\nu_2$, and $\nu_3$, which are eigenstates of the mass operator. At the same time, exist three neutrino states $\nu_e$, $\nu_\mu$, and $\nu_\tau$, which are eigenstates of the interaction operators. These states are not the same but are linearly related via the Pontecorvo-Maki-Nakagawa-Sakata matrix. Neutrino propagators correspond to the eigenstates of the mass operator. However, these states cannot be registered by any physical apparatus since the eigenstates of the interaction operator are coupled to the matter states. Therefore, this model provides an explicit example of a difference between propagating states and states that can be registered.

Quantum chromodynamics provides the most vivid justification for the need for such a distinction \cite{Halzen:1984mc,Creutz:1983njd,Smilga:2001ck}. Within the theory, quarks experience mixing \cite{Cabibbo:1963yz,Kobayashi:1973fv}, so the same logic is applicable. Quark propagators correspond to the mass operator eigenstates, which differs from the interaction operator eigenstates. The Cabibbo-Kobayashi-Maskawa matrix linearly connects these eigenstates. Further, the confinement makes it impossible to observe quark states directly. The spectrum of observed states is colourless and does not match the spectrum of propagating states. One uses the distribution functions technique to relate calculations within quantum chromodynamics with observational quantities \cite{Altarelli:1981ax,Soper:1996sn,Belitsky:2005qn}. A distribution function describes an external probe's probability of interacting with a quark or a gluon composing a non-perturbative state such as atomic nuclei. In other words, one operates with matrix elements having interaction operator eigenstates in the initial or finite states. In this way, quantum chromodynamics not only points to the necessity to distinguish between these states but actively requires such a distinguishment for all computational purposes.

Lastly, the Poincare group is essential for quantum field theory. The group admits two Casimir operators. A Casimir operator is an operator that commutates with all other operators of the group algebra. Consequently, the eigenstates of a Casimir operator remain invariant under the Poincare group transformations. Two Casimir operators of the Poincare group define the notions of mass and spin, which is the basis for constructing and classifying all quantum fields \cite{Wigner:1939cj,Bargmann:1948ck,Bilal:2001nv}. Therefore, without the Poincare group, one shall construct a new field theory significantly different from the standard one. This reasoning does not imply that such theories do not exist, shall not be constructed, or are irrelevant to quantum gravity. On the contrary, such theories are widely developed, with the conformal field theory being (perhaps) the most well-known \cite{Rychkov:2016iqz,Schottenloher:2008zz,Maldacena:1997re,Klebanov:2002ja,Bellucci:2002ji}. However, the discussion of their possible implementations for quantum gravity lies far beyond the scope of this paper.

The crucial role of the Poincare group found an implementation within the rapidly developing scattering amplitude techniques \cite{Elvang:2013cua,Henn:2014yza,Arkani-Hamed:2017jhn,Travaglini:2022uwo}. In contrast with the conventional quantum field theory, the scattering amplitude technique does not appeal to the notion of a quantum field. One classifies the initial and final states of the theory as eigenstates of the Poincare group Casimir operators without any references to quantum fields. Transformation features of a matrix element with respect to the Lorentz transformation provide a way to fix the matrix element's general structure uniquely. The optical theorem shows that a given matrix element has multiple contributions. The consecutive application of the theorem recovers the perturbative expansion as the unavoidable consequence of the unitarity. In this way, the scattering amplitude description allows one to recover the perturbative approach to quantum field theory without directly referencing quantum fields but using only the Poincare group.

These premises point to the following features in understanding perturbative quantum gravity. Firstly, choosing a flat background is necessary to construct a perturbative quantum theory of gravity within the standard quantum field theory framework. Secondly, the theory applies to a scattering of separated gravitational perturbations and matter states, but its applicability may extend beyond this setup.

The existence of the flat background is crucial due to the role of the Poincare group. Let us highlight one more time that it is possible to work within a theory with a different background, both from fundamental and technical points of view. However, such theories will bring us outside the standard quantum field theory formalism. The flat background allows one to use the Poincare group and to define eigenstates of mass and spin operators. In turn, one can define the scattering problem in the standard way described in many textbooks \cite{Christensen:1979iy,Birrell:1982ix}. One associates initial and final states with past and future spatial infinity while assuming that interaction occurs in the flat spatial region between these states.

The formalism of perturbative quantum gravity extends beyond the standard scattering problem. Firstly, propagating graviton states may not be associated with the states registered by physical apparatus. Gravitons are propagating states because they are associated with eigenstates of mass and chirality operators. As was pointed out above, this does not imply that these exact states are registered by physical apparatus. The graviton propagator receives highly non-trivial quantum corrections altering its structure \cite{Anber:2011ut,Latosh:2020wbo}, further indicating that propagating graviton states shall not be associated with states registered by physical apparatus. Finally, there is an ongoing discussion concerning the mere possibility of detecting graviton states directly  \cite{Dyson:2013hbl}. Although the main point of this article is related to the subject of our discussion, it will bring us far away from our main aim, so we will not discuss this claim further. Still, even without a discussion about graviton's detectability, there are enough reasons to believe that graviton states shall be viewed only as propagating states.

These arguments allow one to use perturbative quantum gravity beyond the narrow scope of the scattering problem. There are many branches of research where such an extension of perturbative quantum gravity is assumed implicitly. Perhaps the most recognised one is the calculation of the effective action. Although the detailed discussion of effective action in quantum gravity lies far beyond the scope of this paper (and was done in many publications, for instance, \cite{Buchbinder:2017ea}), we will touch upon its two most important features. 

The effective action technique provides a tool to calculate the classical value (the expectation value) of a quantum field. In other words, one expects to recover the classical dynamics of a system driven by its quantum behaviour. From the technical point of view, to calculate the ($n$-loop) effective action, one sums ($n$-loop) one-particle irreducible diagrams. It is well-recognised that even one-loop effective action for a gauge theory (including gravity) is gauge-dependent. In response to this finding, the unique effective action technique was developed \cite{Barvinsky:1985an,Vilkovisky:1984st}. It treats the gauge-fixing parameter of a theory as a generalised coordinate and constructs a description of internal field space with curved geometry. In turn, it is possible to define a generalisation of the variational derivative similar to the standard covariant derivative. The unique effective action is defined in terms of such derivatives, which makes it gauge-independent. At the same time, such generalisation of the variational derivative introduces new terms to the effective action. These terms cannot be obtained from one-particle irreducible diagrams, which seems to contradict the original setup. One shall only consider the on-shell effective action to resolve this apparent contradiction. The situation is similar to the calculation of scattering cross sections. Although it is possible to calculate an off-shell scattering cross section, physical apparatus can produce and register only on-shell states. Similarly, the classical value of a quantum field can only be associated with an on-shell object. In this way, the unique effective action removes the gauge dependence from calculations while the fixation of the effective action on the mass shell recovers the classical dynamic of quantum fields.

Lastly, we shall comment on the renormalisation of perturbative quantum gravity. The term "renormalisation" is commonly used in the literature, although it may not be the most suitable for this discussion. For the sake of consistency, we will continue to use this term. However, it is worth noting that the discussion below goes far beyond renormalisation as it is understood in conventional renormalisable models such as quantum electrodynamics, $SU(N)$ Yang-Mills theory, and the standard model. 

The growing consensus is that perturbative quantum gravity does not experience any problems with renormalisation. This point of view does not imply that the theory is well-renormalisable. On the contrary, it is believed that extending the perturbative quantum gravity to a conventionally renormalisable theory is impossible. Consequently, all ``problems'' of the theory shall be viewed as features, while the complete theory of quantum gravity development shall be performed separately.

The main feature of the perturbative quantum gravity defining its renormalisation behaviour is that it generates a new set of higher dimensional operators at each loop order. This feature was first observed in the classical paper \cite{tHooft:1974toh}, where the one-loop divergences of perturbative quantum gravity were studied. At the one-loop level, the theory generates divergencies proportional to operators $R^2$, $R_{\mu\nu}^2$, and $R_{\mu\nu\alpha\beta}^2$. The divergence of the Riemann tensor squared term is irrelevant since the term can always be brought to the Gauss-Bonnet term, which is the complete derivative in $D=4$. 

These divergences vanish for pure general relativity without matter for on-shell amplitudes. On shall states satisfy the vacuum Einstein equations $R_{\mu\nu}=0$, so all the divergent contributions are cancelled. This feature is accidental since it requires the theory to exist in $D=4$, has no matter degrees of freedom, and works only for on-shell amplitudes. Many other results later confirmed this general theory feature (see, for instance, \cite{Goroff:1985th,Stelle:1976gc,Latosh:2018xai,Latosh:2020jyq}).

The fact that the theory generates new operators at each level of perturbation theory can be negated from the technical point of view, but the theory loses its predictability. Indeed, let us assume that we calculated a gravitational scattering cross section at the tree level and compared it with empirical data. At the tree level, the theory has a single coupling, the gravitational coupling, and we shall recover its value from the experiment. When we go at the one-loop level, the value of the gravitational coupling is not affected by quantum effects \cite{Anber:2011ut,Latosh:2020wbo}. However, the theory develops divergences proportional to $R^2$ and $R_{\mu\nu}^2$ operators. From the technical point of view, these divergences can be subtracted and replaced by two new couplings. In turn, the values of these new couplings shall once again be recovered from the experiment. Consequently, by increasing the level of perturbation theory (the number of loop corrections), we do not increase the precision of calculations. On the contrary, we require more data to match experimental results with each iteration.

These arguments show that perturbative quantum gravity cannot be treated like the conventionally renormalisable theory. From the physical point of view, the root of the problem lies in the fact that the theory is effective. Since it admits a separated energy scale, the Planck scale, it marks the natural limit of its applicability. In full agreement with the standard effective field theory logic \cite{Burgess:2003jk,Donoghue:1994dn,Georgi:1993mps}, the theory shall include the term terms suppressed by the gravitational coupling the closer to the Planck energy it approaches.

We shall summarise this section to conclude the discussion of the perturbative quantum gravity. Perturbative quantum gravity is a perturbatively constructed quantum field theory. The following distinctive features:
\begin{list}{$\bullet$}{}
    \item 
        Gravitons are propagating massless degrees of freedom with chirality $\pm 2$.
    \item 
        There are no reasons to believe that gravitons shall be associated with degrees of freedom measured by physical apparatus.
    \item 
        The theory is equivalent to a quantum theory of small metric perturbations.
    \item 
        The theory has infinite interaction operators, but the same gravitational coupling parametrises them.
    \item 
        The theory is effective. It develops new operators at each new order of perturbation theory.
    \item 
        The theory does not experience problems with renormalisation since it shall not be renormalised due to its effective nature.
\end{list}
Although the discussion can be extended further, especially considering the ongoing study of its quantum features, this section describes the core of the theory relevant to the main aim of the article. With the given background, we can meaningfully address the problem of generating Feynman rules within perturbative quantum gravity.

\section{Computational tools}\label{Computational_Tools}

This section discusses technical tools that make an efficient calculation of the interaction rules for perturbative quantum gravity possible. The main feature of the theory is the factorisation that allows one to separate a quantum gravity action into factors of two types. Factors of the first type contain derivatives and a finite number of terms. The others are free from derivatives but are infinite series. We discuss this factorisation in the following subsection. Factors of the first type are easy to calculate and require no sophisticated techniques, so we will not discuss them in detail. Factors of the second type present a more significant challenge. Nonetheless, it is possible to establish certain recursive relations, which highly improve the efficiency of their calculations. The second subsection discusses such relations.

We will present results in a series of definitions and theorems. Detailed proofs for many theorems are omitted as they are not this paper's primary focus and can be found in other publications.

\subsection{Factorisation}

We begin with the formal definition of the perturbative metric. 

\begin{definition}
    {~}\\
    The perturbative metric $g_{\mu\nu}$ is defined as follows:
    \begin{align}
        g_{\mu\nu} \overset{\text{def}}{=} \eta_{\mu\nu} + \kappa \, h_{\mu\nu}.
    \end{align}
    Here, $h_{\mu\nu}$ is the small metric perturbations, $\eta_{\mu\nu} = \operatorname{diag}(+---)$ is the Minkowski metric used to raise and low indices, and $\kappa$ is the gravity coupling related with the Newton's constant $G_\text{N}$
    \begin{align}
        \kappa^2 \overset{\text{def}}{=} 32\,\pi\,G_\text{N}.
    \end{align}
\end{definition}
\noindent This formula is not a truncation of an infinite series but a finite expression with no omitted terms. The metric still introduces infinite series in the theory.

\begin{theorem}
    \begin{align}
        \begin{split}
            g^{\mu\nu} &= \sum\limits_{n=0}^\infty (-\kappa)^n (h^n)^{\mu\nu}, \\
            \sqrt{-g} &= \sum\limits_{n=0}^\infty (-\kappa)^n \sum\limits_{m=1}^n \cfrac{1}{m!} \left( - \frac{1}{2} \right)^m \left[ \sum\limits_{k_1 + \cdots + k_m = n} \frac{ \tr (h^{k_1}) \cdots \tr (h^{k_m})}{k_1 \cdots k_m} \right] ,\\
            \mathfrak{e}^m{}_\mu &= \sum\limits_{n=0}^\infty \kappa^n \binom{\frac12}{n} \left( h^n \right)^m{}_\mu ,\\
            \mathfrak{e}_m{}^\mu &= \sum\limits_{n=0}^\infty \kappa^n \binom{-\frac12}{n} \left( h^n \right)_m{}^\mu ,\\
        \end{split}
    \end{align}
    Here $\mathfrak{e}$ is the vierbein, $\binom{n}{m}$ is the binomial coefficients defined via the $\Gamma$ function, and the following notations are used:
    \begin{align}
        \begin{split}
            \left( h^n \right)^{\mu\nu} &= h^{\mu}{}_{\sigma_1} \,h^{\sigma_1}{}_{\sigma_2} \cdots h^{\sigma_{n-1}\nu} ,\\
            \left( h^n \right) &= h^{\sigma_1}{}_{\sigma_2} h^{\sigma_2}{}_{\sigma_3} \cdots h^{\sigma_n}{}_{\sigma_1} .\\
        \end{split}
    \end{align}
\end{theorem}

\begin{proof}
    The proof of the first expression is trivial. One should use the given expression to calculate $g_{\mu\sigma} g^{\sigma\nu}$ and verify its correctness. The papers \cite{Prinz:2020nru, Latosh:2022ydd, Latosh:2023zsi} provide a more detailed discussion of the proof.

    Proof of the second expression is much more complicated. Firstly, one shall use the relation between the determinant and trace of a matrix:
    \begin{align}
        \sqrt{-g} = \left( - \det[g_{\mu\nu}]\right)^{1/2} = \Bigg( -\exp\Big[ \tr\big\{ \ln \left( \eta_{\mu\nu} + \kappa h_{\mu\nu} \right) \big\} \Big] \Bigg)^{1/2}.
    \end{align}
    Secondly, one shall factorise one flat spacetime metric and reduce the expression to a single exponent:
    \begin{align}
        \sqrt{-g} = \Bigg( -\exp\Big[ \tr\big\{ \ln \left( \eta_{\mu\sigma} \left[\delta^\sigma_{\nu} + \kappa h^\sigma{}_\nu \right] \right) \big\} \Big] \Bigg)^{1/2} = \exp\left[ \frac12 \tr\Big\{ \ln\left( \delta^\sigma_\nu + \kappa \, h^\sigma{}_\nu\right)  \Big\}\right] .
    \end{align}
    Lastly, one expands both $\exp$ and $\ln$ functions as Taylor series. The resulting expression is an infinite power series, with each term also being a power series. Nonetheless, it is possible to rearrange the terms of this series to obtain the desired expression. The publication \cite{Latosh:2022ydd} provides a detailed explanation of the derivation.

    Lastly, the expression for vierbein was first obtained in \cite{Prinz:2020nru}. One assumes the vierbein admits a power series expansion with unknown coefficients. The values of coefficients are fixed uniquely, so the expansion fits well-known relations for the vierbein.
\end{proof}

Only these objects and their combinations generate all infinite expansions that may enter a quantum gravity action within Riemann geometry. However, these are not all the objects present in the theory. We still have to discuss the Christoffel symbols and the Riemann tensor. The following theorem gives their structure.

\begin{theorem}
    \begin{align}
        \begin{split}
            \Gamma_{\alpha\mu\nu} =& \cfrac{\kappa}{2}\,\left[ \pd_\mu h_{\nu\alpha} + \pd_\nu h_{\mu\alpha} - \pd_\alpha h_{\mu\nu} \right], \\
            \Gamma^\alpha_{\mu\nu} =& \cfrac{\kappa}{2} \left[ \eta^{\mu\nu} - \kappa\, h^{\mu\nu} + \kappa^2\, h^{\mu\sigma}\,h_\sigma{}^\nu + \okappa{3} \right] \left[ \pd_\mu h_{\nu\beta} + \pd_\nu h_{\mu\beta} - \pd_\beta h_{\mu\nu} \right], \\
            \left( \Gamma_\mu \right)_{ab} =& \mathfrak{e}_a{}^\alpha \mathfrak{e}_b{}^\beta \,\kappa\,\left[ \pd_\beta h_{\alpha\mu} - \pd_\alpha h_{\beta\mu} \right] .
        \end{split}
    \end{align}
\end{theorem}

\begin{proof}
    Proof of the first two expressions is trivial and relies only on the definition of the Christoffel symbols
    \begin{align}
        \Gamma^\alpha_{\mu\nu} \overset{\text{def}}{=} \cfrac{1}{2}\, g^{\alpha\beta} \left[ \pd_\mu g_{\nu\beta} + \pd_\nu g_{\mu\beta} - \pd_\beta g_{\mu\nu} \right].
    \end{align}
    The last expression is proved as follows. Firstly, one shall use the definition of the spin connection:
    \begin{align}
        \left(\Gamma_\mu \right)_{ab} \overset{\text{def}}{=} \mathfrak{e}_a{}^\alpha \mathfrak{e}_b{}^\beta \, \Gamma_{\alpha\mu\beta} + g_{\rho\sigma} \, \mathfrak{e}_a{}^\rho  \, \pd_\mu \mathfrak{e}_b{}^\sigma.
    \end{align}
    The spin connection is antisymmetric for indices $a$ and $b$. For the perturbative metric, the last term is constituted only by symmetric matrices, so it is symmetric with respect to $a$ and $b$. Therefore, the last term does not contribute to the spin connection. For the first term, one shall use the explicit expression for the Christoffel symbol and make it antisymmetric. Similarly, paper \cite{Latosh:2022ydd} discusses the derivation in more detail.
\end{proof}

The obtained expressions allow us to describe the structure of the Riemann tensor with the following theorem.

\begin{theorem}
    \begin{align}
        \begin{split}
            R_{\mu\nu}{}^\alpha{}_\beta &= g^{\alpha\lambda} \left[ \pd_\mu \Gamma_{\lambda\nu\beta} - \pd_\nu \Gamma_{\lambda\mu\beta} + g^{\rho\sigma} \left\{ \Gamma_{\rho\nu\lambda} \Gamma_{\sigma\mu\beta} - \Gamma_{\rho\mu\lambda} \Gamma_{\sigma\nu\beta} \right\} \right] , \\
            R_{\mu\nu\alpha\beta} &= \pd_\mu \Gamma_{\alpha\nu\beta} - \pd_\nu \Gamma_{\alpha\mu\beta} + g^{\rho\sigma} \left\{ \Gamma_{\rho\nu\alpha} \Gamma_{\sigma\mu\beta} - \Gamma_{\rho\mu\alpha} \Gamma_{\sigma\nu\beta} \right\} ,\\
            R_{\mu\nu} &= g^{\rho\sigma}\left[ \pd_\rho \Gamma_{\sigma\mu\nu} - \pd_\mu \Gamma_{\sigma\rho\nu} \right] + g^{\rho\sigma} g^{\lambda\tau} \left[ \Gamma_{\tau\mu\nu} \Gamma_{\lambda\rho\sigma} - \Gamma_{\lambda\mu\sigma} \Gamma_{\tau\nu\rho} \right] ,\\
            R &= g^{\mu\nu} g^{\alpha\beta} \pd_\mu \left[ \Gamma_{\nu\alpha\beta} - \Gamma_{\alpha\nu\beta} \right] +g^{\mu\nu} g^{\alpha\beta} g^{\rho\sigma} \left[ \Gamma_{\alpha\mu\nu} \Gamma_{\beta\rho\sigma} - \Gamma_{\alpha\mu\rho} \Gamma_{\beta\nu\sigma} \right] .
        \end{split}
    \end{align}
\end{theorem}

\begin{proof}
    Expressions for the Ricci tensor and the scalar curvature can be derived directly from the expressions for the Riemann tensor, making their derivation purely technical. The derivation of the first two formulas is technically straightforward but requires lengthy computations. To obtain the desired result, one shall use a single relation for the inverse metric derivative:
    \begin{align}
        \pd_\mu g^{\alpha\beta} = - g^{\alpha\rho} g^{\beta\sigma} \pd_\mu g_{\rho\sigma}.
    \end{align}
    The remaining part of the proof relies on the manipulation of indices.
\end{proof}

These theorems constitute the factorisation theorem crucial for the perturbative quantum gravity.
\begin{theorem}
    {~}\\
    In Riemann geometry, a gravity action evaluated with the perturbative metric can always take a form where all terms begin infinite series are free from derivatives and expressed via $\sqrt{-g}$, $g^{\mu\nu}$, and $\mathfrak{e}_m{}^\mu$.
\end{theorem}
\noindent In other words, as long as we operate within Riemann geometry, we can always split a gravity action into two different parts. The part that involves derivatives is finite and can be calculated explicitly. The other part is an infinite series, but it is always expressed in terms of elementary series for the volume factors $\sqrt{-g}$, the inverse metric $g^{\mu\nu}$, and the vierbein $\mathfrak{e}_m{}^\mu$.

\subsection{Recursive relations}

The previous section provides a way to factorise a gravity action, but we must still address calculational challenges. Namely, we shall elaborate on the method to describe the structure of perturbative expansions for the inverse metric, volume factor, and vierbein. We address this challenge in this section.

\begin{definition}
    {~}\\
    We define the plain $I$-tensor of the $n$-th order as follows:
    \begin{align}
        I^{\rho_1\sigma_1\cdots\rho_n\sigma_n}_{(n)} \overset{\text{def}}{=} \eta^{\sigma_1\rho_2} \eta^{\sigma_2\rho_3}\cdots \eta^{\sigma_n\rho_1}.
    \end{align}
\end{definition}
\noindent The tensor has no additional symmetries. We introduce an additional tensor to account for the essential symmetries of the interaction rules.
\begin{definition}
    {~}\\
    We define the $\mathcal{I}$-tensor of the $n$-th order as follows:
    \begin{align}
        \mathcal{I}_{(n)}^{\rho_1\sigma_1\cdots\rho_n\sigma_n} \overset{\text{def}}{=} \cfrac{1}{2^n} \, \cfrac{1}{n!} \left[ I_{(n)}^{\rho_1\sigma_1\cdots\rho_n\sigma_n} + \text{permutations} \right].
    \end{align}
    Permutations account for all terms that make the $\mathcal{I}$-tensor symmetric with respect to permutations of indices within each index pair $\mu_i\leftrightarrow\nu_i$, and with respect to permutations of any two index pairs $\{\mu_i,\nu_i\} \leftrightarrow \{\mu_j,\nu_j\}$.
\end{definition}

The $\mathcal{I}$ tensor of the $n$-th order has $n! \, 2^n$ terms. In other words, the number of terms grows faster than the factorial. This feature holds for tensors discussed below. Although this feature presents a computational challenge, it is an essential part of the standard perturbative approach. The final expression for a gravity interaction vertex shall be symmetric with respect to graviton permutations, so it is necessary to operate with tensor structures respecting this symmetry.

The plain $I$-tensor provides a way to operate with powers of $h_{\mu\nu}$ and its traces given by the following theorem. The theorem is trivial, so we omit the proof.

\begin{theorem}
    {~}\\
    The plain $I$-tensor defines the structure of powers of $h_{\mu\nu}$ and its traces.
    \begin{align}
        \begin{split}
            &(h^n)^{\mu\nu}  = I_{(1+n)}^{\mu\nu\rho_1\sigma_1\cdots\rho_n\sigma_n}\,h_{\rho_1\sigma_1}\cdots h_{\rho_n\sigma_n}\,,\\
            & \operatorname{tr}(h^n) =I_{(n)}^{\rho_1\sigma_1\cdots\rho_n\sigma_n}\,h_{\rho_1\sigma_1}\cdots h_{\rho_n\sigma_n}\,.
        \end{split}
    \end{align}
\end{theorem}

Using this theorem, one can demonstrate how the plain $I$-tensor describes the perturbative structure of the inverse metric. The following theorem, which describes this relation, is a direct corollary of previous theorems, so we omit its proof.

\begin{theorem}
    {~}\\
    The inverse metric perturbative expansion in terms of the plain $I$ tensor reads:
    \begin{align}
        g^{\mu\nu} &= \sum\limits_{n=0}^\infty (-1)^n\,\kappa^n\,I^{\mu\nu\rho_1\sigma_1\cdots\rho_n\sigma_n}_{(1+n)}\,h_{\rho_1\sigma_1}\cdots h_{\rho_n\sigma_n}\,.
    \end{align}
\end{theorem}

This theorem provides an easy way to study the features of the I-tensor, such as the following easily proven contraction feature.

\begin{theorem}
    {~}\\
    The $I$-tensor admits the following contraction feature:
    \begin{align}
        \eta_{\mu\nu} \, I^{\mu\nu\rho_1\sigma_1\cdots\rho_n\sigma_n} = I^{\rho_1\sigma_1\cdots\rho_n\sigma_n}.
    \end{align}
\end{theorem}

\begin{proof}
    We shall start with the following general relation:
    \begin{align}
        d = g^{\mu\nu} \, g_{\mu\nu}.
    \end{align}
    This relation can be expanded:
    \begin{align}
        \begin{split}
            d =& \left( \eta_{\mu\nu} + \kappa \,h_{\mu\nu} \right) \sum\limits_{n=0}^\infty (-1)^n \kappa^n I_{(1+n)}^{\mu\nu\rho_1\sigma_1\cdots\rho_n\sigma_n} h_{\rho_1\sigma_1} \cdots h_{\rho_n\sigma_n} \\
            =& \eta_{\mu\nu} \eta^{\mu\nu} + \sum\limits_{n=1}^\infty (-1)^n \kappa^n \, \eta_{\mu\nu}\, I_{(1+n)}^{\mu\nu\rho_1\sigma_1\cdots\rho_n\sigma_n} h_{\rho_1\sigma_1} \cdots h_{\rho_n\sigma_n} \\
            & \hspace{30pt} + \sum\limits_{n=0}^\infty (-1)^n \kappa^{n+1} I_{(1+n)}^{\mu\nu\rho_1\sigma_1\cdots\rho_n\sigma_n} h_{\mu\nu} h_{\rho_1\sigma_1} \cdots h_{\rho_n\sigma_n}.
        \end{split}
    \end{align}
    The desired relation holds as each term in the infinite sums is cancelled.
\end{proof}

We uncover the perturbative structure of the volume factor in a similar way. Initially, we give a non-constructive definition of the plain $C$-tensor. Further, we demonstrate a specific recursive relation between $C$-tensors of different orders. This relation allows one to calculate a given order's plain $C$-tensor.

\begin{definition}
    {~}\\
    We define the plain $C$-tensor of $n$-th order in such a way to describe the perturbative structure of the volume factor:
    \begin{align}
        \sqrt{-g} \overset{\text{def}}{=} \sum\limits_{n=0}^\infty \, \kappa^n \,C^{\rho_1\sigma_1\cdots\rho_n\sigma_n}_{(n)}\,h_{\rho_1\sigma_1}\cdots h_{\rho_n\sigma_n} \,.
    \end{align}
\end{definition}
\noindent Similar to the previous case, the definition does not require specific symmetries among the indices. The next definition of the $\mathcal{C}$-tensor introduces a tensor with suitable symmetries.

\begin{definition}
    {~}\\
    We defined the $\mathcal{C}$-tensor of the $n$-th order as follows:
    \begin{align}
        \mathcal{C}_{(n)}^{\rho_1\sigma_1\cdots\rho_n\sigma_n} \overset{\text{def}}{=} \cfrac{1}{2^n} \, \cfrac{1}{n!} \, \left[ C_{(n)}^{\rho_1\sigma_1\cdots\rho_n\sigma_n} + \text{permutations} \right].
    \end{align}
    Permutations ensure $\mathcal{C}$-tensor is symmetric with respect to index pair and pair permutations.
\end{definition}

The definition is not constructive since it does not show a way to compute the tensor explicitly; instead, it only describes its function. However, obtaining a recursive relation that facilitates an uncomplicated method to construct $C$-tensors is possible.

\begin{theorem}\label{Theorem_1}
    {~}\\
    The following recursive relation holds.
    \begin{align}\label{the_C_recursion}
        C_{(n)}^{\rho_1\sigma_1\cdots\rho_n\sigma_n} = \cfrac{1}{2\,n} \sum\limits_{k=1}^n \,(-1)^{k-1} \, I_{(k)}^{\rho_1\sigma_1\cdots\rho_k\sigma_k} C_{(n-k)}^{\rho_{k+1}\sigma_{k+1}\cdots\rho_n\sigma_n}.
    \end{align}
\end{theorem}

\begin{proof}
    Firstly, we introduce an auxiliary metric:
    \begin{align}
        \mathbf{g}_{\mu\nu} = \eta_{\mu\nu} + z\, \kappa\,h_{\mu\nu}\,.
    \end{align}
    Here, $z$ is an arbitrary positive real number. This metric has no physical or mathematical meaning except to obtain the discussed relation. It matches the original metric when $z=1$.

    Secondly, we compute the following derivative:
    \begin{align}
        \left. \cfrac{d}{dz} \sqrt{-\mathbf{g}} ~ \right|_{z=1} = \cfrac{\kappa}{2} \, \sqrt{-g} \,g^{\mu\nu} \, h_{\mu\nu}.
    \end{align}
    We applied Jacobi's formula for a non-degenerate matrix $X$:
    \begin{align}
        d \,\det X = \det X \, \operatorname{tr}\left[ X^{-1} \, d X \right].
    \end{align}
    Let us calculate the same derivative using the perturbative expansion:
    \begin{align}
        \left. \cfrac{d}{dz} \sqrt{-\mathbf{g}} ~ \right|_{z=1} = \sum\limits_{n=1}^\infty \,n\,\kappa^n\,C^{\rho_1\sigma_1\cdots\rho_n\sigma_n}_{(n)}\,h_{\rho_1\sigma_1}\cdots h_{\rho_n\sigma_n}.
    \end{align}
    Comparing these expressions, we obtain a relationship that does involve $\mathbf{g}$:
    \begin{align}
        \cfrac{\kappa}{2} \, \sqrt{-g} \,g^{\mu\nu} \, h_{\mu\nu} = \sum\limits_{n=1}^\infty \,n\,\kappa^n\,C^{\rho_1\sigma_1\cdots\rho_n\sigma_n}_{(n)}\,h_{\rho_1\sigma_1}\cdots h_{\rho_n\sigma_n}.
    \end{align}

    The left-hand side of this expression can be calculated perturbatively:
    \begin{align}
        \begin{split}
            \cfrac{\kappa}{2} \, \sqrt{-g} \,g^{\mu\nu} \, h_{\mu\nu}  &=\!\cfrac{\kappa}{2}\!\sum\limits_{p_1=0}^\infty \kappa^{p_1} C^{\rho_1\sigma_1\cdots\rho_{p_1}\sigma_{p_1}}_{(p_1)}  h_{\rho_1\sigma_1}\cdots  h_{\rho_{p_1}\sigma_{p_1}} \!\! \sum\limits_{p_2=0}^\infty \! (-1)^{p_2} \kappa^{p_2} I^{\mu\nu\lambda_1\tau_1\cdots\lambda_{p_2}\tau_{p_2}}_{(p_2+1)} h_{\lambda_1\tau_1}\cdots h_{\lambda_{p_2}\tau_{p_2}}h_{\mu\nu}\\
            &=\sum\limits_{n=1}^\infty \sum\limits_{k=1}^n \, \cfrac{(-1)^{k-1}}{2} ~ \kappa^n\,I_{(k)}^{\rho_1\sigma_1\cdots\rho_k\sigma_k}\,C_{(n-k)}^{\rho_{k+1}\sigma_{k+1}\cdots\rho_n\sigma_n}\,h_{\rho_1\sigma_1}\cdots h_{\rho_n\sigma_n} \\
            &= \sum\limits_{n=1}^\infty \kappa^n \,\sum\limits_{k=1}^n \cfrac{(-1)^{k-1}}{2} ~ I_{(k)}^{\rho_1\sigma_1\cdots\rho_k\sigma_k}\,C_{(n-k)}^{\rho_{k+1}\sigma_{k+1}\cdots\rho_n\sigma_n}\,h_{\rho_1\sigma_1}\cdots h_{\rho_n\sigma_n} .
        \end{split}
    \end{align}
    After comparing these expressions, we can derive the required recursive relation.
    \begin{align}
        C_{(n)}^{\rho_1\sigma_1\cdots\rho_n\sigma_n} = \cfrac{1}{2\,n} \sum\limits_{k=1}^n \,(-1)^{k-1} \, I_{(k)}^{\rho_1\sigma_1\cdots\rho_k\sigma_k} C_{(n-k)}^{\rho_{k+1}\sigma_{k+1}\cdots\rho_n\sigma_n}.
    \end{align}
\end{proof}

The equation \eqref{the_C_recursion} defines a plain $C$-tensor of $n$-th order recursively using plain $C$-tensors of lower orders. One would desire to express $C_{(n)}$ using $C_{(n-1)}$ alone, but unfortunately, it appears to be impossible. The reason is that formula \eqref{the_C_recursion} involves all low orders that must contribute to the plain $C$-tensor. Since the volume factor $\sqrt{-g}$ is the simple object, we conjecture that no simpler recursive formula exists.

The same logic shall be used to construct more complicated structures generated by the volume factor and the inverse metric. Firstly, we give a non-constructive definition of a tensor. Secondly, we find a recursive relation that defines the required tensor iteratively.

\begin{definition}
    {~}\\
    We define the family fo plain $C_{(n)}$-tensors as follows:
    \begin{align}
        \begin{split}
            \sqrt{-g} \, g^{\mu\nu} \overset{\text{def}}{=}& \sum\limits_{n=0}^\infty \, \kappa^n \,C_{(1; n)}^{\mu\nu,\rho_1\sigma_1\cdots\rho_n\sigma_n} h_{\rho_1\sigma_1} \cdots h_{\rho_n\sigma_n} \, , \\
            \sqrt{-g} \, g^{\mu\nu}\,g^{\alpha\beta} \overset{\text{def}}{=}& \sum\limits_{n=0}^\infty \, \kappa^n \,C_{(2; n)}^{\mu\nu\alpha\beta,\rho_1\sigma_1\cdots\rho_n\sigma_n} h_{\rho_1\sigma_1} \cdots h_{\rho_n\sigma_n} \, , \\
            \sqrt{-g} \, g^{\mu\nu}\,g^{\alpha\beta}\,g^{\rho\sigma} \overset{\text{def}}{=}& \sum\limits_{n=0}^\infty \, \kappa^n \,C_{(3; n)}^{\mu\nu\alpha\beta\rho\sigma,\rho_1\sigma_1\cdots\rho_n\sigma_n} h_{\rho_1\sigma_1} \cdots h_{\rho_n\sigma_n} \, , \\
            \sqrt{-g} \, g^{\mu_1\nu_1}\cdots g^{\mu_l\nu_l} \overset{\text{def}}{=}& \sum\limits_{n=0}^\infty \, \kappa^n \,C_{ (l; n) }^{\mu_1\nu_1\cdots\mu_l\nu_l,\rho_1\sigma_1\cdots\rho_n\sigma_n} h_{\rho_1\sigma_1} \cdots h_{\rho_n\sigma_n} \, .
        \end{split}
    \end{align}
\end{definition}
\noindent These tensors do not have any additional symmetry. The corresponding symmetric tensors are defined as follows.

\begin{definition}
    \begin{align}
        \mathcal{C}_{(l; n)}^{\mu_1\nu_1\cdots\mu_l\nu_l,\rho_1\sigma_1\cdots\rho_n\sigma_n} \overset{\text{def}}{=} \cfrac{1}{2^n} \, \cfrac{1}{n!} \, \left[ C_{(l; n)}^{\mu_1\nu_1\cdots\mu_l\nu_l,\rho_1\sigma_1\cdots\rho_n\sigma_n} + \text{permutations} \right] .
    \end{align}
    Permutations ensure symmetric $\mathcal{C}$-tensor with respect to index pair and pair permutation. It is important to note that symmetrisation accounts only for indices contracted with the perturbation indices.
\end{definition}

\begin{theorem}\label{Theorem_2}
    {~}\\
    The following recursive relation holds for plain $C_{(l)}$-tensors.
    \begin{align}
        C_{ (l; n) }^{\mu_1\nu_1\mu_2\nu_2 \cdots \mu_l\nu_l, \rho_1\sigma_1\cdots\rho_n\sigma_n} = \sum\limits_{p=0}^n (-1)^p \,I_{(1+p)}^{\mu_1\nu_1\rho_1\sigma_1\cdots\rho_p\sigma_p} C_{(l-1; n-p) }^{\mu_2\nu_2 \cdots \mu_l\nu_l, \rho_{p+1}\sigma_{p+1}\cdots\rho_n\sigma_n} .
    \end{align}
\end{theorem}

\begin{proof}
    The proof relies on the possibility of factorising a single metric from the perturbative expression. We can use this fact to define the $C_{(l)}$ tensors.  
    \begin{align}
        \sqrt{-g} \, g^{\mu_1\nu_1}\cdots g^{\mu_l\nu_l} \overset{\text{def}}{=}& \sum\limits_{n=0}^\infty \, \kappa^n \,C_{ (l; n) }^{\mu_1\nu_1\cdots\mu_l\nu_l,\rho_1\sigma_1\cdots\rho_n\sigma_n} h_{\rho_1\sigma_1} \cdots h_{\rho_n\sigma_n} \, .
    \end{align}
    At the same time, we can perform the perturbative expansion of $g^{\mu_1\nu_1}$ separately.
    \begin{align}
        \begin{split}
            & \sqrt{-g} \, g^{\mu_1\nu_1}\, g^{\mu_2\nu_2}\cdots g^{\mu_l\nu_l} = g^{\mu_1\nu_1} \sqrt{-g} \, g^{\mu_2\nu_2}\cdots g^{\mu_l\nu_l}\\
            &= \sum\limits_{n_1=0}^\infty (-1)^{n_1} \kappa^{n_1} I_{(n_1+1)}^{\mu_1\nu_1\rho_1\sigma_1\cdots\rho_{n_1}\sigma_{n_1}} h_{\rho_1\sigma_1}\cdots h_{\rho_{n_1}\sigma_{n_1}}  \sum\limits_{n_2=0}^\infty \kappa^{n_2} C_{ (l-1; n_2)}^{\mu_2\nu_2\cdots\mu_l\nu_l,\lambda_1\tau_1\cdots\lambda_{n_2}\tau_{n_2} } h_{\lambda_1\tau_1} \cdots h_{\lambda_{n_2}\tau_{n_2}}\\
            &=\sum\limits_{n_1=0}^\infty \sum\limits_{n_2=0}^\infty (-1)^{n_1} \kappa^{n_1 + n_2} I_{(n_1+1)}^{\mu_1\nu_1\rho_1\sigma_1\cdots\rho_{n_1}\sigma_{n_1}} C_{ (l-1; n_2)}^{\mu_2\nu_2\cdots\mu_l\nu_l,\rho_{n_1+1}\sigma_{n_1+1}\cdots\rho_{n_1+n_2}\sigma_{n_1+n_2}} h_{\rho_1\sigma_1}\cdots h_{\rho_{n_1+n_2}\sigma_{n_1+n_2}}.
        \end{split}
    \end{align}
    One can decouple the summation indices by redefining them, resulting in the following formula:
    \begin{align}
        \sqrt{-g} \, g^{\mu_1\nu_1}\, g^{\mu_2\nu_2}\cdots g^{\mu_l\nu_l} =\sum\limits_{n=0}^\infty \sum\limits_{p=0}^n \kappa^n (-1)^p  I_{(1 + p)}^{\mu_1\nu_1\rho_1\sigma_1\cdots\rho_p \sigma_p } C_{ (l-1; n_2)}^{\mu_2\nu_2\cdots\mu_l\nu_l,\rho_{p+1}\sigma_{p+1}\cdots\rho_n \sigma_n } h_{\rho_1\sigma_1}\cdots h_{\rho_n \sigma_n }.
    \end{align}
  The proof is concluded by making a direct comparison with the definition of the original tensor $C_{(l)}$.
\end{proof}

Finally, we treat vierbein similarly. We define the plain $E$-tensor and plain $C_E$-tensor to encapsulate the structure of the following factors.
\begin{theorem}
    \begin{align}
        \begin{split}
            \mathfrak{e}_\mu {}^\nu \overset{\text{def}}{=}& \sum\limits_{n=0}^\infty \kappa^n E_{\mu~(n)}^{~~\nu\rho_1\sigma_1\cdots\rho_n\sigma_n} h_{\rho_1\sigma_1}\cdots h_{\rho_n\sigma_n} \\
            =& \sum\limits_{n=0}^\infty \, \kappa^n \binom{-\frac12 }{~n} I_{\mu~(n)}^{~~\nu\rho_1\sigma_1\cdots\rho_n\sigma_n} h_{\rho_1\sigma_1} \cdots h_{\rho_n\sigma_n} \, ,
        \end{split} \\
        \begin{split}
            \sqrt{-g} \, \mathfrak{e}_\mu {}^\nu \overset{\text{def}}{=}& \sum\limits_{n=0}^\infty \kappa^n {C_E}_{\mu ~(n)}^{~\nu,\rho_1\sigma_1\cdots\rho_n\sigma_n} h_{\rho_1\sigma_1}\cdots h_{\rho_n\sigma_n} \\
            =&\sum\limits_{n=0}^\infty \, \kappa^n ~\sum\limits_{p=0}^n  \binom{-\frac12}{~n} C_{(p)}^{\rho_1\sigma_1\cdots\rho_p\sigma_p} I_{\mu~(1+n-p)}^{~~\nu\rho_{p+1}\sigma_{p+1}\cdots\rho_n\sigma_n}  h_{\rho_1\sigma_1} \cdots h_{\rho_n\sigma_n} \, .
        \end{split}
    \end{align}
\end{theorem}
\noindent These tensors do not possess any additional symmetry. In practice, we never use the plain $E$-tensor, so we do not discuss it further. Moreover, the plain $E$-tensor match the plain $I$-tensor up to a constant. On the contrary, the plain $C_E$-tensor is relevant for Dirac fermions, so we define its symmetric generalisation.
\begin{definition}
    \begin{align}
        { \mathcal{C}_E }_{\mu~(n)}^{~~\nu,\rho_1\sigma_1\cdots\rho_n\sigma_n} \overset{\text{def}}{=} \cfrac{1}{2^n} \, \cfrac{1}{n!} \, \left[ {C_E}_{\mu~(n)}^{~\nu,\rho_1\sigma_1\cdots\rho_n\sigma_n} + \text{permutations} \right].
    \end{align}
    Permutations account for all terms that make the $\mathcal{C}_E$-tensor symmetric with respect to permutations of indices within each index pair $\mu_i\leftrightarrow\nu_i$, and with respect to permutations of any two index pairs $\{\mu_i,\nu_i\} \leftrightarrow \{\mu_j,\nu_j\}$.
\end{definition}

The recursive relation for the $C_E$ tensor holds similarly to the $C_{(l)}$ tensor. The proof for $C_E$-tensor is constructed the same way as the proof for $C_{(l)}$-tensor and is omitted.
\begin{theorem}\label{Theorem_3}
    \begin{align}
        {C_E}_{\mu~(n)}^{~\nu, \rho_1\sigma_1\cdots\rho_n\sigma_n} = \sum\limits_{p=0}^n E_{\mu~(1+p)}^{~~\nu\rho_1\sigma_1\cdots\rho_p\sigma_p} C_{(n-p)}^{\rho_{p+1}\sigma_{p+1}\cdots\rho_n\sigma_n} .
    \end{align}
\end{theorem}

The main result of this section are theorems \ref{Theorem_1}, \ref{Theorem_2}, and \ref{Theorem_3}. They are essential from the computational point of view since their implementation significantly improves the code performance. We believe no simpler formula can be found based on the $\sqrt{-g}$ perturbative expansion structure. All the factors discussed above contain $\sqrt{-g}$, which makes it impossible to construct recursive relations for them without a recursive relation for $\sqrt{-g}$. Theorem \ref{Theorem_1} shows that the $n$-th perturbative order incorporates all lower orders, making it impossible to establish a relation between the $n$-th and $(n-1)$-th orders alone. Therefore, we believe it is impossible to construct truly recursive relations for all other factors.

\section{Feynman rules}\label{Feynman_Rules}

The computational tools discussed in the previous section provide a comprehensive framework for deriving Feynman's rules for perturbative quantum gravity. This section covers the derivation of interaction rules for simple scalars, Horndeski gravity, Dirac fermions, massive and massless vector fields, $SU(N)$ Yang-Mills model, and general relativity. This collection of models is sufficient to account for quantum effects within the standard model.

We handle all interaction rules as follows. We use the factorisation theorem to separate terms that are infinite series from terms that are finite but contain derivatives. The finite part is calculated explicitly, while the infinite part is expanded in the infinite perturbative series. The discussed recursive relations provide a way to evaluate each term in such a perturbative expansion. Therefore, pointing to the structure contributing to a given perturbative order is sufficient.

We will not discuss the derivation of the Feynman rules from the action via the path integral formalism. The procedure is discussed in multiple sources, including classical textbooks \cite{Peskin:1995ev,Weinberg:1995mt,Weinberg:1996kr,Weinberg:2000cr}. The procedure also uses Fourier transformation to generate rules in the momentum representation. We also do not discuss the Fourier transformation for the same reason.

Lastly, we shall introduce new notations for the practical sake. In perturbative expansions considered below, using $\mathcal{I}$, $\mathcal{C}$ and other previously defined tensors is not optimal. These tensors play a crucial role in the technical side of the calculation, but they do not provide the most practical description.

If the quantity $X$ admits a perturbative expansion in small metric perturbations, then we shall use the following notations:
\begin{align}
    X \overset{\text{note}}{=} \sum\limits_{n=0}^\infty \kappa^n (X)^{\rho_1\sigma_1\cdots\rho_n\sigma_n} \, h_{\rho_1\sigma_1} \cdots h_{\rho_n\sigma_n}.
\end{align}
Here $(X)^{\rho_1\sigma_1\cdots\rho_n\sigma_n}$ is the contribution to the $n$-th order of the expansion. One can define via derivatives:
\begin{align}
    (X)^{\rho_1\sigma_1\cdots\rho_n\sigma_n} = \cfrac{\delta}{\delta h_{\rho_1\sigma_1}} \cdots \cfrac{\delta}{\delta h_{\rho_n\sigma_n}} \, X \Bigg|_{h=0}.
\end{align}

Let us present the following expressions to illustrate these notations. The components of the inverse metric are noted as follows:
\begin{align}
    g^{\mu\nu} = \sum\limits_{n=0}^\infty \kappa^n \left(g^{\mu\nu} \right)^{\rho_1\sigma_1\cdots\rho_n\sigma_n} h_{\rho_1\sigma_1} \cdots h_{\rho_n\sigma_n}.
\end{align}
The explicit values of the first few of these components are:
\begin{align}
    (g^{\mu\nu} ) &= \eta^{\mu\nu} = \mathcal{I}^{\mu\nu} \,, & (g^{\mu\nu})^{\alpha\beta} &= - \cfrac12 \left( \eta^{\mu\alpha} \eta^{\nu\beta} + \eta^{\mu\beta} \eta^{\nu\alpha} \right)  = - \mathcal{I}^{\mu\nu\alpha\beta} \,.
\end{align}
Implementation of these notations for the volume factor produces the following:
\begin{align}
    \sqrt{-g} = \sum\limits_{n=0}^\infty \kappa^n \left(\sqrt{-g}\right)^{\rho_1\sigma_1\cdots \rho_n\sigma_n} h_{\rho_1\sigma_1}\cdots h_{\rho_n\sigma_n}\,,
\end{align}
\begin{align}
    \left(\sqrt{-g}\right) &= 1\,, & \left(\sqrt{-g}\right)^{\mu\nu} &= \cfrac12\,\eta^{\mu\nu} = \mathcal{C}^{\mu\nu} \,, & \left(\sqrt{-g}\right)^{\mu\nu\alpha\beta} &= \cfrac18\,\left( -\eta^{\mu\alpha} \eta^{\nu\beta} - \eta^{\mu\beta} \eta^{\nu\alpha} + \eta^{\mu\nu} \eta^{\alpha\beta} \right) = \mathcal{C}^{\mu\nu\alpha\beta} \,.
\end{align}

\subsection{Scalar field with a potential}

The action for a single free scalar field with mass $m$ with the minimal coupling to gravity reads:
\begin{align}
    \mathcal{A}_{s=0} = \int d^4 x \sqrt{-g} \left[ \cfrac12\left(\nabla\phi\right)^2 - \cfrac{m^2}{2} \, \phi^2 \right] .
\end{align}
Following the factorisation theorem, the action takes the following form:
\begin{align}
    \mathcal{A}_{s=0} = \int d^4 x \left[ \cfrac12\, \sqrt{-g}\, g^{\mu\nu} \, \pd_\mu\phi \, \pd_\nu\phi - \cfrac{m^2}{2}\,\sqrt{-g} \, \phi^2 \right].
\end{align}
The following theorem describes its structure in the Fourier representation.

\begin{theorem}
    \begin{align}
        \begin{split}
            \mathcal{A}_{s=0} =& \sum\limits_{n=0}^\infty \int \cfrac{d^4 p_1}{(2\pi)^4} \cfrac{d^4 p_2}{(2\pi)^4} \prod\limits_{i=1}^n \cfrac{d^4 k_i}{(2\pi)^4}\,(2\pi)^4 \delta\left(p_1+p_2+\sum k_i\right) h_{\rho_1\sigma_1}(k_1) \cdots h_{\rho_n\sigma_n}(k_n)\\
            &\times \kappa^n \left[ -\cfrac12\, \left(\sqrt{-g}\, g^{\mu\nu}\right)^{\rho_1\sigma_1\cdots\rho_n\sigma_n}  \mathcal{I}_{\mu\nu\alpha\beta} (p_1)^\alpha (p_2)^\beta - \cfrac{m^2}{2} \left(\sqrt{-g}\right)^{\rho_1\sigma_1\cdots\rho_n\sigma_n} \right]\,\phi(p_1) \phi(p_2)\,.
        \end{split}
    \end{align}
    Here, $k_i$ are momenta of gravitons, $p_1$ and $p_2$ are momenta of scalars, and the same definition of $\mathcal{I}$ tensor is used:
    \begin{align}
    \mathcal{I}^{\mu\nu\alpha\beta} = \cfrac12\, \left( \eta^{\mu\alpha}\eta^{\nu\beta} + \eta^{\mu\beta}\eta^{\nu\alpha}\right) .
    \end{align}
\end{theorem}

The contribution of order $\okappa{0}$ descirbes the standard scalar propagator:
\begin{align}\label{scalar_propagator}
  \begin{gathered}
    \begin{fmffile}{Diag01}
      \begin{fmfgraph}(30,30)
        \fmfleft{L}
        \fmfright{R}
        \fmf{dashes}{L,R}
      \end{fmfgraph}
    \end{fmffile}
  \end{gathered}
  = i \, \cfrac{1}{p^2 - m^2} ~.
\end{align}
The other terms define the coupling of scalar field kinetic energy to gravity:
\begin{align}
   \nonumber \\
  \begin{gathered}
    \begin{fmffile}{FR_S_1}
      \begin{fmfgraph*}(30,30)
        \fmfleft{L1,L2}
        \fmfright{R1,R2}
        \fmf{dbl_wiggly}{L1,V}
        \fmf{dbl_wiggly}{L2,V}
        \fmfdot{V}
        \fmf{dashes}{V,R1}
        \fmf{dashes}{V,R2}
        \fmffreeze
        \fmf{dots}{L1,L2}
        \fmflabel{$p_1$}{R1}
        \fmflabel{$p_2$}{R2}
        \fmflabel{$\rho_1\sigma_1$}{L1}
        \fmflabel{$\rho_n\sigma_n$}{L2}
      \end{fmfgraph*}
    \end{fmffile}
  \end{gathered}
  \hspace{10pt}= -i\, \kappa^n \,\left[ \left(\sqrt{-g}\, g^{\mu\nu} \right)^{\rho_1\sigma_1\cdots\rho_n\sigma_n} \, \mathcal{I}_{\mu\nu\alpha\beta} (p_1)^\alpha (p_2)^\beta + m^2\, \left( \sqrt{-g}\right)^{\rho_1\sigma_1\cdots\rho_n\sigma_n} \right]. \\ \nonumber
\end{align}
The dotted line on the left part of the diagram notes the presence of $n\geq 1$ graviton lines.

We can derive the equation for the scalar field potential similarly. The potential $V(\phi)$ is a Taylor series in the scalar field $\phi$. Each term in this expansion represents a scalar field self-interaction coupled with gravity. Thus, deriving the interaction rule for a single power-law potential is sufficient. 

We consider a power-law potential with $p \geq 3$ being a whole number and $\lambda_p$ being a coupling with the mass dimension $4-p$:
\begin{align}
  \mathcal{A} = \int d^4 x \sqrt{-g} \left[ \cfrac{\lambda_p}{p!} ~ \phi^p \right] = \int d^4 x \left[ \sqrt{-g} ~ \cfrac{\lambda_p}{p!} ~ \phi^p \right].
\end{align}
The following theorem describes the action in Fourier representation.

\begin{theorem}
    \begin{align}
        \begin{split}
            \mathcal{A} =& \sum\limits_{n=0}^\infty\int \prod\limits_{j=1}^p \cfrac{d^4 q_j}{(2\pi)^4} \prod\limits_{i=0}^n \cfrac{d^4 k_i}{(2\pi)^4} \, (2\pi)^4 \,\delta\Big( \sum q_j + \sum k_i \Big)\,h_{\rho_1\sigma_1}(k_1)\cdots h_{\rho_n\sigma_n}(k_n) \\
            &\times \kappa^n \,\cfrac{\lambda_p}{p!} \, \left(\sqrt{-g}\right)^{\rho_1\sigma_1\cdots\rho_n\sigma_n} \phi(q_1) \cdots \phi(q_p)\,.
        \end{split}
    \end{align}
\end{theorem}

This results in the following interaction rule:
\begin{align}
    \nonumber \\
    \begin{gathered}
        \begin{fmffile}{FR_S_Power}
            \begin{fmfgraph*}(30,30)
                \fmfleft{L1,L2}
                \fmfright{R1,R2}
                \fmf{dbl_wiggly}{L1,V}
                \fmf{dbl_wiggly}{L2,V}
                \fmf{dashes}{R1,V}
                \fmf{dashes}{R2,V}
                \fmfdot{V}
                \fmffreeze
                \fmf{dots}{L1,L2}
                \fmf{dots}{R1,R2}
                \fmflabel{$\rho_1\sigma_1$}{L1}
                \fmflabel{$\rho_n\sigma_n$}{L2}
                \fmflabel{$q_1$}{R1}
                \fmflabel{$q_p$}{R2}
            \end{fmfgraph*}
        \end{fmffile}
    \end{gathered}
    \hspace{10pt}= i\,\kappa^n \, \lambda_p \, \left(\sqrt{-g}\right)^{\rho_1\sigma_1\cdots\rho_n\sigma_n} . \\ \nonumber
\end{align}
This expression fully describes the gravitational coupling of scalar field potentials for any whole $p \geq 3$.

\subsection{Horndeski's gravity}

The Horndeski theory is the most general scalar-tensor theory of gravity that admits second-order field equations and minimal coupling to matter. It was discovered in \cite{Horndeski:1974wa} and independently rediscovered \cite{Kobayashi:2011nu}. Because of the second-order field equations, the theory is free from the Ostrogradsky instability \cite{Ostrogradsky:1850fid,Woodard:2015zca,Woodard:2006nt}.

A combination of the following Lagrangians gives the theory:
\begin{align}
  \mathcal{A} = \int d^4 x \sqrt{-g} \left[ \mathcal{L}_2 + \mathcal{L}_3 + \mathcal{L}_4 + \mathcal{L}_5 \right] ,
\end{align}
\begin{align}
  \begin{split}
    \mathcal{L}_2 &= G_2 (\phi,X) \,,\\
    \mathcal{L}_3 &= G_3 (\phi,X) \, \square \phi \,,\\
    \mathcal{L}_4 &= G_4 (\phi,X) \, R + G_{4,X} \left[ \left(\square\phi\right)^2  - \left(\nabla_\mu\nabla_\nu\phi\right)^2 \right] \,,\\
    \mathcal{L}_5 &= G_5 (\phi,X) \, G_{\mu\nu} \nabla^\mu\nabla^\nu\phi - \cfrac16\, G_{5,X} \left[ \left( \square\phi\right)^3 - 3 \, \square\phi \, \left(\nabla_\mu\nabla_\nu\phi\right)^2 + 2 \left(\nabla_\mu\nabla_\nu \phi\right)^3 \right] .
  \end{split}
\end{align}
Here $G_i = G_i (\phi,X)$ are functions of the scalar field $\phi$ and its kinetic term $X = g^{\mu\nu} \,\nabla_\mu \phi \, \nabla_\nu\phi$, while $G_{4,X}$ and $G_{5,X}$ note derivatives with respect to the kinetic term. Lagrangians $\mathcal{L}_2$ and $\mathcal{L}_3$ describe scalar field self-interaction with the minimal coupling to gravity. In turn, $\mathcal{L}_4$ and $\mathcal{L}_5$ describe healthy non-minimal coupling to gravity. Terms containing $G_{4, X}$ and $G_{5, X}$ are relevant only when $G_4$ and $G_5$ depend on the scalar field kinetic term. In that case, they cancel out higher derivative contributions to the field equations.

\subsubsection{Horndeski $G_2$ class}

The $G_2$ class of Horndeski theory describes a minimal coupling between the scalar field and gravity. This coupling is minimal because it does not involve graviton momenta.

The function $G_2$ shall be smooth enough to admit a Taylor expansion:
\begin{align}
    G_2(\phi,X) & = \sum\limits_{a=0}^\infty \sum\limits_{b=0}^\infty \lambda_{(a,b)} \, \phi^a \, X^b = \sum\limits_{a=0}^\infty \sum\limits_{b=0}^\infty \lambda_{(a,b)} \, \phi^a \, g^{\mu_1\nu_1} \cdots g^{\mu_b\nu_b} \pd_{\mu_1}\phi \, \pd_{\nu_1} \phi \cdots \pd_{\mu_b} \phi \, \pd_{\nu_b} \phi .
\end{align}
Here $\lambda_{(a,b)}$ is the dimensional coupling with the mass dimension $ 4 (1-b) - a $. Without the loss of generality, one can consider a single term of this expansion and consider the following action:
\begin{align}
    \int d^4 x \sqrt{-g} \, G_2(\phi, X) \to \int d^4 x \sqrt{-g} \, \lambda_{(a,b)} \, \phi^a \, X^b \,.
\end{align}
\noindent The perturbative expansion is given by the following theorem.
\begin{theorem}
    \begin{align}
        \begin{split}
        & \int d^4 x \sqrt{-g} \, \lambda_{(a,b)} \, \phi^a \, X^ b = \int d^4 x \sqrt{-g} \, g^{\mu_1\nu_1} \cdots g^{\mu_b \nu_b}  \,\phi^a \, \pd_{\mu_1} \phi \, \pd_{\nu_1} \phi \cdots \pd_{\mu_b} \phi \pd_{\nu_b} \phi  \\
        &=\sum\limits_{n=0}^\infty \int \prod\limits_{i=0}^n \cfrac{d^4 k_i}{(2\pi)^4} \, h_{\rho_i\sigma_i}(k_i) \, \prod\limits_{j=1}^{a + 2b} \cfrac{d^4 p_j}{(2\pi)^4} \, \phi(p_j) \, (2 \pi)^4 \delta \left( \sum k_i + \sum p_i \right)\\
        & \hspace{20pt} \times \kappa^n \,(-1)^b\, \lambda_{(a,b)} ~ \left(\sqrt{-g}\, g^{\mu_1\nu_1}\cdots g^{\mu_b\nu_b}\right)^{\rho_1\sigma_1\cdots\rho_n\sigma_n} (p_1)_{\mu_1} (p_2)_{\nu_1} \cdots (p_{2b-1})_{\mu_b} (p_{2b})_{\nu_b} \,.
        \end{split}
    \end{align}
\end{theorem}

\noindent Lastly, the Feynman rule corresponding to that action reads
\begin{align}
    \nonumber \\
    \begin{split}
    & \hspace{25pt}
        \begin{gathered}
            \begin{fmffile}{G2_vertex}
                \begin{fmfgraph*}(30,30)
                    \fmfleft{L1,L2}
                    \fmfright{R1,R2}
                    \fmf{dbl_wiggly}{L1,V}
                    \fmf{dbl_wiggly}{L2,V}
                    \fmf{dashes}{V,R1}
                    \fmf{dashes}{V,R2}
                    \fmfdot{V}
                    \fmffreeze
                    \fmf{dots}{L1,L2}
                    \fmf{dots}{R1,R2}
                    \fmflabel{$\rho_1\sigma_1,k_1$}{L1}
                    \fmflabel{$\rho_n\sigma_n,k_n$}{L2}
                    \fmflabel{$p_1$}{R1}
                    \fmflabel{$p_{2b+a}$}{R2}
                    \fmflabel{$\lambda_{a,b}$}{V}
                \end{fmfgraph*}
            \end{fmffile}
        \end{gathered}
        \hspace{25pt} = i\, \kappa^n (-1)^b\,\lambda_{a,b} \, \left(\sqrt{-g}\, g^{\mu_1\nu_1}\cdots g^{\mu_b\nu_b}\right)^{\rho_1\sigma_1\cdots\rho_n\sigma_n} (p_1)_{\mu_1} (p_2)_{\nu_1} \cdots (p_{2b-1})_{\mu_b} (p_{2b})_{\nu_b} \\
        & \hspace{90pt} + \text{permutations} .
    \end{split} 
    \\ \nonumber
\end{align}
\noindent The permutation terms account for all possible permutations of the scalar field momenta. Let us also note that for the $b=0$ case, the interaction reduces the scalar field potential.

\subsubsection{Horndeski $G_3$ class}

The interactions of the $G_3$ class of Horndeski theory represent the simplest non-minimal coupling between the scalar field and gravity. The final expression of this interaction contains a single Christoffel symbol, which allows for the explicit inclusion of a single graviton momentum in the corresponding Feynman rule.

We shall first simplify the expression and make the Christoffel symbol explicit.
\begin{theorem}
    \begin{align}
        \begin{split}
        \int d^4 x \sqrt{-g} \, G_3 \, \square \phi = \int d^4 x \sqrt{-g}\, g^{\mu\nu} G_3 \,\pd_\mu \pd_\nu \phi -\int d^4 x \sqrt{-g} \, g^{\mu\nu} g^{\alpha\beta}\,\Gamma_{\alpha\mu\nu} G_3 \, \pd_\beta \phi  .
        \end{split}
    \end{align}
\end{theorem}

Secondly, $G_3$ should be expanded in a Taylor series:
\begin{align}
    G_3(\phi,X) = \sum\limits_{a=0}^\infty \sum\limits_{b=0}^\infty \, \Theta_{(a,b)} \, \phi^a \, X^b .
\end{align}
We can study a single term of the expansion without the loss of generality:
\begin{align}
    \int d^4 x \sqrt{-g}\, G_3 (\phi,X) \, \square \phi \to \int d^4 x \sqrt{-g} \, \Theta_{(a,b)} \, \phi^a \, X^b \, \square \phi .
\end{align}

The following theorem gives the perturbative structure of this interaction.
\begin{theorem}
    \begin{align}
        \begin{split}
            &\int d^4 x \sqrt{-g} \, \Theta_{(a,b)}\, \phi^a\, X^b\, \square \phi \\
            &= \int d^4 x \sqrt{-g}\, g^{\mu\nu} \Theta_{(a,b)} \,\phi^a \, X^b \, \pd_\mu \pd_\nu \phi - \int d^4 x \sqrt{-g}\, g^{\mu\nu} g^{\alpha\beta} \Gamma_{\alpha\mu\nu} \Theta_{(a,b)} \,\phi^a \, X^b \, \pd_\beta \phi \\
            & = \sum\limits_{n=0}^\infty \int \prod\limits_{i=0}^n \cfrac{d^4 k_i}{(2\pi)^4}\,h_{\rho_i\sigma_i}(k_i) \!\!\!\! \prod\limits_{j=0}^{a+2 b+1} \!\! \cfrac{d^4 p_j}{(2\pi)^4}\, \phi(p_j) ~(2\pi)^4 \delta\left( \sum k_i + \sum p_i \right)\\
            &\hspace{10pt} \times\kappa^n (-1)^{b+1} \Theta_{(a,b)} \left(\sqrt{-g}\,g^{\mu\nu} \!g^{\alpha_1\beta_1} \!\cdots\! g^{\alpha_b\beta_b}\right)^{\rho_1\sigma_1\cdots\rho_n\sigma_n} \!\!(p_1)_{\alpha_1}(p_2)_{\beta_1}\!\cdots (p_{2b-1})_{\alpha_b}(p_{2b})_{\beta_b} (p_{2b+1})_\mu (p_{2b+1})_\nu\\
            &+\sum\limits_{n=1}^\infty \int \prod\limits_{i=1}^n \cfrac{d^4 k_i}{(2\pi)^4}\,h_{\rho_i\sigma_i}(k_i) \!\!\!\! \prod\limits_{j=0}^{a+2 b+1} \!\! \cfrac{d^4 p_j}{(2\pi)^4} \, \phi(p_j) ~(2\pi)^4 \delta\left( \sum k_i + \sum p_i \right)\\
            & \hspace{10pt} \times\kappa^n (-1)^{b+2} \Theta_{(a,b)} \left(\sqrt{-g}\, g^{\mu\nu} g^{\rho\sigma}g^{\alpha_1\beta_1}\cdots g^{\alpha_b\beta_b}\right)^{\rho_1\sigma_1\cdots\rho_{n \!-\! 1} \sigma_{n \!-\! 1}} (\Gamma_{\rho\mu\nu})^{\lambda\rho_n\sigma_n} (k_n)_\lambda \\
            & \hspace{20pt} \times (p_1)_{\alpha_1}(p_2)_{\beta_1}\cdots (p_{2b-1})_{\alpha_b}(p_{2b})_{\beta_b} (p_{2b+1})_\sigma .
        \end{split}
    \end{align}
\end{theorem}
\noindent It should be noted that the second part of the expression does not contribute at the $\kappa^0$ level since the Christoffel symbol vanishes on the background.

The corresponding interaction rule is given by the following.
\begin{align}
    \nonumber \\
    \begin{split}
        & \hspace{20pt}
        \begin{gathered}
            \begin{fmffile}{G3_vertex}
                \begin{fmfgraph*}(30,30)
                    \fmfleft{L1,L2}
                    \fmfright{R1,R2}
                    \fmf{dbl_wiggly}{L1,V}
                    \fmf{dbl_wiggly}{L2,V}
                    \fmf{dashes}{V,R1}
                    \fmf{dashes}{V,R2}
                    \fmfdot{V}
                    \fmffreeze
                    \fmf{dots}{L1,L2}
                    \fmf{dots}{R1,R2}
                    \fmflabel{$\rho_1\sigma_1,k_1$}{L1}
                    \fmflabel{$\rho_n\sigma_n,k_n$}{L2}
                    \fmflabel{$p_1$}{R1}
                    \fmflabel{$p_{2b+a+1}$}{R2}
                    \fmflabel{$\Theta_{a,b}$}{V}
                \end{fmfgraph*}
            \end{fmffile}
        \end{gathered}
        \hspace{20pt} = i \,\kappa^n (-1)^{b+1} \Theta_{(a,b)} \Bigg[ \left(\sqrt{-g}\,g^{\mu\nu}\,g^{\alpha_1\beta_1} \!\cdots\! g^{\alpha_b\beta_b}\right)^{\rho_1\sigma_1\cdots\rho_n\sigma_n} (p_{2b+1})_\mu (p_{2b+1})_\nu \\ \\
        & - \left(\sqrt{-g}\,g^{\mu\nu}\!g^{\rho\sigma} \!\cdots\! g^{\alpha_b\beta_b}\right)^{\rho_1\sigma_1\cdots\sigma_{n \!-\! 1}\rho_{n \!-\! 1}} (\Gamma_{\rho\mu\nu})^{\lambda\rho_n\sigma_n} (k_n)_\lambda (p_{2b+1})_\sigma \Bigg] (p_1)_{\alpha_1}(p_2)_{\beta_1}\cdots (p_{2b-1})_{\alpha_b}(p_{2b})_{\beta_b} \\
        & + \text{permutations} .
    \end{split}
\end{align}
\noindent The permutation terms account for all possible permutations of the scalar field and graviton momenta.
  
\subsubsection{Horndeski $G_4$ class}

Interactions of $G_4$ class describe more sophisticated coupling of the scalar field to gravity. Similarly to the previous case, it involves Christoffel symbols. In contrast with the previous case, Christoffel symbols are present implicitly in the scalar curvature. Moreover, the interaction is described by two related terms, with the second term cancelling higher derivative contributions to the classical equations of motion.

We assume $G_4$ is smooth enough to admit a Taylor expansion. To ensure that the theory admits general relativity in the corresponding limit, one shall assume that the expansion takes the following form:
\begin{align}
    G_4 = -\cfrac{2}{\kappa^2} + \sum\limits_{a,b} \Upsilon_{(a,b)} \phi^a \, X^b .
\end{align}
The first term corresponds to pure general relativity, while the other terms describe the interaction between gravity and the scalar field. We discuss the part corresponding to general relativity in another section, so we do not discuss it here. Without loss of generality, we will use a single term from this expansion:
\begin{align}
    \int d^4 x \sqrt{-g} \left[  \Upsilon_{(a,b)}\,\phi^a\, X^b \,R + b\,\Upsilon_{(a,b)}\,\phi^a\, X^{b-1} \left[ \left(\square\phi\right)^2 - \left(\nabla_\mu\nabla_\nu\phi\right)^2\right] \right].
\end{align}
The following theorems give the structure of this action.

\begin{theorem}
    \begin{align}
        \begin{split}
            &\int d^4 x \sqrt{-g} \, \Upsilon_{(a,b)} \,\phi^a X^b R = \!\!\int \!\! d^4 x \sqrt{-g} \, g^{\mu\nu}g^{\alpha\beta} \pd_\alpha \left[ \Gamma_{\beta\mu\nu} \!-\! \Gamma_{\mu\beta\nu}\right] \Upsilon_{(a,b)} \phi^a X^b \\
            &\hspace{120pt}+ \int\!\! d^4 x \sqrt{-g} \, g^{\mu\nu} g^{\alpha\beta} g^{\rho\sigma}\left[ \Gamma_{\alpha\mu\rho} \Gamma_{\beta\nu\sigma} \!-\! \Gamma_{\alpha\mu\nu} \Gamma_{\beta\rho\sigma} \right] \Upsilon_{(a,b)} \phi^a X^b \\
            & =\sum\limits_{n=1}^\infty\int\prod\limits_{i=1}^n \cfrac{d^4 k_i}{(2\pi)^4} h_{\rho_i\sigma_i}(k_i) \prod\limits_{j=1}^{a+2b}\cfrac{d^4 p_j}{(2\pi)^4}\,\phi(p_j) \, (2\pi)^4\delta\left(\sum k_i + \sum p_j \right) \\
            &\hspace{10pt}\times \kappa^n (-1)^{b+1}\Upsilon_{(a,b)}\!\left(\sqrt{-g}\,g^{\mu\nu}\!g^{\alpha\beta}\!g^{\alpha_1\beta_1} \!\cdots\! g^{\alpha_b\beta_b}\right)^{\rho_1\sigma_1\cdots\rho_{n \!-\! 1}\sigma_{n \!-\! 1}} (p_1)_{\alpha_1} (p_2)_{\beta_1} \cdots (p_{2b-1})_{\alpha_b} (p_{2n})_{\beta_b} \\
            &\hspace{20pt} \times \left[ \left(\Gamma_{\beta\mu\nu}\right)^{\lambda\rho_n\sigma_n} - \left(\Gamma_{\mu\beta\nu}\right)^{\lambda\rho_n\sigma_n}\right] (k_n)_\alpha(k_n)_\lambda \\
            &+\sum\limits_{n=2}^\infty\int\prod\limits_{i=1}^n \cfrac{d^4 k_i}{(2\pi)^4} h_{\rho_i\sigma_i}(k_i) \prod\limits_{j=1}^{a+2b}\cfrac{d^4 p_j}{(2\pi)^4}\,\phi(p_j) \, (2\pi)^4\delta\left(\sum k_i + \sum p_j \right)\\
            &\hspace{10pt}\times \kappa^n (-1)^{b+2}\Upsilon_{(a,b)} \!\left(\sqrt{-g}\,g^{\mu\nu}\! g^{\alpha\beta} \! g^{\rho\sigma} \! g^{\alpha_1\beta_1} \!\cdots\! g^{\alpha_b\beta_b} \right)^{\rho_1\sigma_1\cdots\rho_{n \!-\! 2}\sigma_{n \!-\! 2}} (p_1)_\alpha (p_2)_\beta\cdots (p_{2b-1})_{\alpha_b} (p_{2b})_{2b}  \\
            & \hspace{20pt}\times \left[ \left(\Gamma_{\alpha\mu\rho}\right)^{\lambda_1\rho_{n-1}\sigma_{n-1}} \left(\Gamma_{\beta\nu\sigma}\right)^{\lambda_2\rho_n\sigma_n} - \left(\Gamma_{\alpha\mu\nu}\right)^{\lambda_1\rho_{n-1}\sigma_{n-1}} \left(\Gamma_{\beta\rho\sigma}\right)^{\lambda_2\rho_n\sigma_n} \right] (k_{n-1})_{\lambda_1} (k_{n})_{\lambda_2} .
        \end{split}
    \end{align}
\end{theorem}

\noindent The expression consists of two parts. The first part contributes at the $\okappa{1}$ level and describes the interaction between a few scalars and a single graviton. The second part does not contribute to $\okappa{1}$ and describes the interaction between two or more gravitons with a few scalars.

The following theorem describes the structure of the term containing $G_{4,X}$.
\begin{theorem}
    \begin{align}
        \begin{split}
            & \int d^4 x \sqrt{-g} \, \phi^a X^{b-1} \left[ \left( \square \phi \right)^2 - \left(\nabla_\mu\nabla_\nu\phi\right)^2 \right]\\
            & = \int d^4 x \sqrt{-g} \left[g^{\mu\nu} g^{\alpha\beta} - \cfrac12\,\left( g^{\mu\alpha}g^{\nu\beta} + g^{\mu\beta}g^{\nu\alpha} \right)\right] \phi^a X^b \pd_\mu\pd_\nu\phi \pd_\alpha\pd_\beta \phi \\
            & \hspace{10pt} + \int d^4 x\sqrt{-g} \left[g^{\mu\nu} g^{\alpha\beta} - \cfrac12\,\left( g^{\mu\alpha}g^{\nu\beta} + g^{\mu\beta}g^{\nu\alpha} \right)\right] g^{\rho\sigma} \Gamma_{\rho\mu\nu}\, \phi^a X^b \pd_\sigma\phi \pd_\alpha\pd_\beta\phi \\
            & \hspace{10pt} + \int d^4 x\sqrt{-g} \left[g^{\mu\nu} g^{\alpha\beta} - \cfrac12\,\left( g^{\mu\alpha}g^{\nu\beta} + g^{\mu\beta}g^{\nu\alpha} \right)\right] g^{\rho\sigma} g^{\lambda\tau} \Gamma_{\rho\mu\nu} \Gamma_{\lambda\alpha\beta} \phi^a X^b \pd_\sigma\phi \pd_\tau \phi .
        \end{split}
    \end{align}
\end{theorem}
\noindent In turn, the Fourier structure of this contribution is given by the following.
\begin{theorem}
    \begin{align}
        \begin{split}
            & \int d^4 x \sqrt{-g} \, b \,\Upsilon_{(a,b)} \phi^a X^{b-1} \left[ \left( \square \phi \right)^2 - \left(\nabla_\mu\nabla_\nu\phi\right)^2 \right]\\
            & = \sum\limits_{n=0}^\infty \int\prod\limits_{i=0}^n \cfrac{d^4 k_i}{(2\pi)^4} h_{\rho_i\sigma_i}(k_i) \prod\limits_{j=1}^{a+2b} \cfrac{d^4 p_j}{(2\pi)^4} \, (2\pi)^4 \delta\left(\sum k_i + \sum p_j \right) \\
            & \hspace{10pt} \times \kappa^n (-1)^{b-1} \,b \Upsilon_{(a,b)} \left( \sqrt{-g} \left[g^{\mu\nu} g^{\alpha\beta} - \cfrac12\,\left( g^{\mu\alpha}g^{\nu\beta} + g^{\mu\beta}g^{\nu\alpha} \right)\right] g^{\alpha_1\beta_1} \cdots g^{\alpha_b \beta_b} \right)^{\rho_1\sigma_1\cdots\rho_n\sigma_n} \\
            & \hspace{20pt} \times (p_1)_{\alpha_1} (p_2)_{\beta_1} \cdots (p_{2b-1})_{\alpha_b} (p_{2b})_{\beta_b} (p_{2b+1})_\mu (p_{2b+1})_\nu (p_{2b+2})_\alpha (p_{2b+2})_\beta\\
            & + \sum\limits_{n=1}^\infty \int\prod\limits_{i=0}^n \cfrac{d^4 k_i}{(2\pi)^4} h_{\rho_i\sigma_i}(k_i) \prod\limits_{j=1}^{a+2b} \cfrac{d^4 p_j}{(2\pi)^4} \, (2\pi)^4 \delta\left(\sum k_i + \sum p_j \right) \\
            & \hspace{10pt} \times \kappa^n (-1)^{b} \,2\,b \Upsilon_{(a,b)} \left( \sqrt{-g} \left[g^{\mu\nu} g^{\alpha\beta} - \cfrac12\,\left( g^{\mu\alpha}g^{\nu\beta} + g^{\mu\beta}g^{\nu\alpha} \right)\right] g^{\rho\sigma} g^{\alpha_1\beta_1} \cdots g^{\alpha_b \beta_b} \right)^{\rho_1\sigma_1\cdots\rho_{n \!-\! 1} \sigma_{n \!-\! 1} } \\
            & \hspace{20pt} \times (\Gamma_{\rho\mu\nu} )^{\lambda\rho_n\sigma_n}(k_n)_\lambda (p_1)_{\alpha_1} (p_2)_{\beta_1} \cdots (p_{2b-1})_{\alpha_b} (p_{2b})_{\beta_b} (p_{2b+1})_\sigma (p_{2b+2})_\alpha (p_{2b+2})_\beta \\
            & + \sum\limits_{n=2}^\infty \int\prod\limits_{i=0}^n \cfrac{d^4 k_i}{(2\pi)^4} h_{\rho_i\sigma_i}(k_i) \prod\limits_{j=1}^{a+2b} \cfrac{d^4 p_j}{(2\pi)^4} \, (2\pi)^4 \delta\left(\sum k_i + \sum p_j \right) \\
            & \hspace{10pt} \times \kappa^n (-1)^{b} \,b \Upsilon_{(a,b)} \left( \sqrt{-g} \left[g^{\mu\nu} g^{\alpha\beta} - \cfrac12\,\left( g^{\mu\alpha}g^{\nu\beta} + g^{\mu\beta}g^{\nu\alpha} \right)\right] g^{\rho\sigma} g^{\lambda\tau} g^{\alpha_1\beta_1} \cdots g^{\alpha_b \beta_b} \right)^{\rho_1\sigma_1\cdots\rho_{n \!-\! 2} \sigma_{n \!-\! 2} } \\
            & \hspace{20pt} \times (\Gamma_{\rho\mu\nu} )^{\lambda_1 \rho_{n \!-\! 1} \sigma_{n \!-\! 1} } (\Gamma_{\lambda\alpha\beta})^{\lambda_2 \rho_n\sigma_n} (k_{n \!-\! 1})_{\lambda_1} (k_n)_{\lambda_2} (p_1)_{\alpha_1} (p_2)_{\beta_1} \cdots (p_{2b-1})_{\alpha_b} (p_{2b})_{\beta_b} (p_{2b+1})_\sigma (p_{2b+2})_\rho .
        \end{split}
    \end{align}
\end{theorem}
\noindent The first part of the expression contributes to the background at order $\okappa{0}$. The second part contributes to $\okappa{1}$ and describes the interaction of a single graviton with a few scalars. The last part contributes to order $\okappa{2}$, and its leading contribution represents the interaction of two gravitons with a few scalars.

The following sophisticated expression gives the resulting expression for the interaction rule.
\begin{align}
    \nonumber \\
    \begin{split}
    &\hspace{40pt}
        \begin{gathered}
            \begin{fmffile}{G4_vertex}
                \begin{fmfgraph*}(30,30)
                    \fmfleft{L1,L2}
                    \fmfright{R1,R2}
                    \fmf{dbl_wiggly}{L1,V}
                    \fmf{dbl_wiggly}{L2,V}
                    \fmf{dashes}{V,R1}
                    \fmf{dashes}{V,R2}
                    \fmfdot{V}
                    \fmffreeze
                    \fmf{dots}{L1,L2}
                    \fmf{dots}{R1,R2}
                    \fmflabel{$\rho_1\sigma_1,k_1$}{L1}
                    \fmflabel{$\rho_n\sigma_n,k_n$}{L2}
                    \fmflabel{$p_1$}{R1}
                    \fmflabel{$p_{2b+a}$}{R2}
                    \fmflabel{$\Upsilon_{a,b}$}{V}
                \end{fmfgraph*}
            \end{fmffile}
        \end{gathered}
        \hspace{15pt}= i\, \kappa^n\, (-1)^{b+1} \Upsilon_{(a,b)} (p_1)_{\alpha_1} (p_2)_{\beta_1}\cdots (p_{2b-1})_{\alpha_b} (p_{2b})_{\beta_b} \\ \\
        & \times \Bigg[ \left( \sqrt{-g}\, g^{\mu\nu}\! g^{\alpha\beta}\! g^{\alpha_1\beta_1} \!\cdots\! g^{\alpha_b\beta_b} \right)^{\rho_1\sigma_1\cdots\rho_{n \!-\! 1} \sigma_{n \!-\! 1} } (k_n)_\alpha (k_n)_\lambda \left[ (\Gamma_{\beta\mu\nu})^{\lambda\rho_n\sigma_n} - (\Gamma_{\mu\beta\nu})^{\lambda\rho_n\sigma_n} \right] \\
        & \hspace{5pt} -\left(\sqrt{-g} \, g^{\mu\nu}\! g^{\alpha\beta}\! g^{\rho\sigma}\! g^{\alpha_1\beta_1}\cdots g^{\rho_b \sigma_b}\right)^{\rho_1\sigma_1 \cdots \rho_{n \!-\! 2} \sigma_{n \!-\! 2}} (k_{n \!-\! 1})_{\lambda_1} (k_n)_{\lambda_2}\\
        &\hspace{10pt}\times\left[ (\Gamma_{\alpha\mu\rho})^{\lambda_1\rho_{n \!-\! 1} \sigma_{n \!-\! 1} } (\Gamma_{\beta\nu\sigma})^{\lambda_2\rho_n\sigma_n} - (\Gamma_{\alpha\mu\nu})^{\lambda_1\rho_{n \!-\! 1} \sigma_{n \!-\! 1} } (\Gamma_{\beta\rho\sigma})^{\lambda_2\rho_n\sigma_n} \right] \\
        & \hspace{5pt} + b \left( \sqrt{-g}\,g^{\mu\nu} \!g^{\alpha\beta} - \frac12\, \sqrt{-g}\, g^{\mu\alpha}\! g^{\nu\beta} - \frac12 \, \sqrt{-g} \, g^{\mu\beta}\! g^{\nu\alpha}\right)^{\rho_1\sigma_1\cdots \rho_n\sigma_n} (p_{2b+1})_\mu (p_{2b+1})_\nu (p_{2b+2})_\alpha (p_{2b+2})_\beta \\
        & \hspace{5pt} - 2 \, b \left( \sqrt{-g}\,g^{\mu\nu} \!g^{\alpha\beta}\! g^{\rho\sigma} - \frac12\, \sqrt{-g}\, g^{\mu\alpha}\! g^{\nu\beta}\! g^{\rho\sigma} - \frac12 \, \sqrt{-g} \, g^{\mu\beta}\! g^{\nu\alpha}\! g^{\rho\sigma} \right)^{\rho_1\sigma_1\cdots \rho_{n \!-\! 1}\sigma_{n \!-\! 1} }\\
        & \hspace{10pt} \times(k_n)_\lambda (\Gamma_{\rho\mu\nu})^{\lambda\rho_n\sigma_n} (p_{2b+1})_\sigma  (p_{2b+2})_\alpha (p_{2b+2})_\beta \\
        & + b\left( \sqrt{-g}\,g^{\mu\nu} \!g^{\alpha\beta}\! g^{\rho\sigma}\, g^{\lambda\tau} - \frac12\, \sqrt{-g}\, g^{\mu\alpha}\! g^{\nu\beta}\! g^{\rho\sigma}\, g^{\lambda\tau} - \frac12 \, \sqrt{-g} \, g^{\mu\beta}\! g^{\nu\alpha}\! g^{\rho\sigma}\, g^{\lambda\tau} \right)^{\rho_1\sigma_1\cdots \rho_{n \!-\! 2}\sigma_{n \!-\! 2} } \\
        &\hspace{10pt} \times (k_{n \!-\! 1})_{\lambda_1} (k_{n \!-\! 2})_{\lambda_2} (\Gamma_{\rho\mu\nu})^{\lambda_1\rho_{n \!-\! 1} \sigma_{n \!-\! 1} } (\Gamma_{\lambda\alpha\beta} )^{\lambda_2\rho_n\sigma_n} (p_{2b+1})_\sigma  (p_{2b+2})_\tau \Bigg] \\
        & +\text{permutations} .
    \end{split}     
\end{align}
Let us note that the second and last terms in the square brackets contribute only to vertices with two or more gravitons.

\subsubsection{Horndeski $G_5$ class}

The final class of Horndeski interaction also describes a non-minimal coupling between the scalar field and gravity. This interaction involves the most derivatives, which still do not result in the higher-order classical field equations. We study this term similar to the others and make its structure explicit.

The following theorem describes the structure of the non-minimal coupling. It does not involve any integration by parts and relies purely on calculating derivatives.
\begin{theorem}
    \begin{align}
        \begin{split}
            G^{\mu\nu} \nabla_\mu \nabla_\nu \phi  =& -\cfrac12 \left[ g^{\mu\nu} g^{\alpha\beta} - g^{\mu\alpha} g^{\nu\beta} - g^{\mu\beta} g^{\nu\alpha} \right] g^{\rho\sigma} \left\{ \pd_\rho \Gamma_{\sigma\alpha\beta} - \pd_\alpha \Gamma_{\rho\beta\sigma} \right\} \pd_\mu \pd_\nu \phi \\
            &+\cfrac12 \left[ g^{\mu\nu} g^{\alpha\beta} - g^{\mu\alpha} g^{\nu\beta} - g^{\mu\beta} g^{\nu\alpha} \right] g^{\rho\sigma} g^{\lambda\tau} \left\{ \pd_\rho \Gamma_{\sigma\alpha\beta} - \pd_\alpha \Gamma_{\rho\beta\sigma} \right\} \Gamma_{\lambda\mu\nu} \pd_\tau \phi \\
            &-\cfrac12 \left[ g^{\mu\nu} g^{\alpha\beta} - g^{\mu\alpha} g^{\nu\beta} - g^{\mu\beta} g^{\nu\alpha} \right] g^{\rho\sigma} g^{\lambda\tau} \left\{ \Gamma_{\rho\alpha\lambda} \Gamma_{\sigma\beta\tau} - \Gamma_{\rho\alpha\beta} \Gamma_{\sigma\lambda\tau} \right\} \pd_\mu \pd_\nu \phi \\
            &+\cfrac12 \left[ g^{\mu\nu} g^{\alpha\beta} - g^{\mu\alpha} g^{\nu\beta} - g^{\mu\beta} g^{\nu\alpha} \right] g^{\rho\sigma} g^{\lambda\tau} g^{\epsilon\omega} \left\{ \Gamma_{\rho\alpha\lambda} \Gamma_{\sigma\beta\tau} - \Gamma_{\rho\alpha\beta} \Gamma_{\sigma\lambda\tau} \right\} \Gamma_{\epsilon\mu\nu} \pd_\omega \phi  .
        \end{split}
    \end{align}
\end{theorem}

The coupling function $G_5 = G_5(\phi,X)$  shall be smooth enough to admit a power series expansion
\begin{align}
    G_5 = \sum\limits_{a,b} \Psi_{(a,b)} \phi^a X^b.
\end{align}
Similar to the previous cases, we can study the following coupling function without the loss of generality:
\begin{align}
    G_5 = \Psi_{(a,b)} \phi^a X^b.
\end{align}

This results in the following theorem describing the Fourier structure of the first part of the interaction.
\begin{theorem}
    \begin{align}
        \begin{split}
            &\int d^4 x \sqrt{-g} \, G_5 \, G^{\mu\nu} \nabla_\mu \nabla_\nu \phi \\
            & = \sum\limits_{n=1}^\infty \int \prod\limits_{i=1}^n \cfrac{d^4 k_i}{(2\pi)^4} \, h_{\rho_i\sigma_i}(k_i) \!\!\!\! \prod\limits_{j=1}^{a+2b+1} \cfrac{d^4 p_j}{(2\pi)^4} \, \phi(p_j) \, (2\pi)^4 \delta\left( \sum k_i + \sum p_j \right) \\
            & \hspace{5pt} \times \cfrac12\, \kappa^n (-1)^{b+1} \Psi_{(a,b)} \left\{\left(g^{\mu\nu} g^{\alpha\beta} - g^{\mu\alpha} g^{\nu\beta} - g^{\mu\beta} g^{\nu\alpha}\right) g^{\rho\sigma}\right\}^{\rho_1\sigma_1\cdots\rho_{n \!-\! 1} \sigma_{n \!-\! 1}} \\
            & \hspace{10pt} \times(k_n)_\lambda \left\{ (k_n)_\rho (\Gamma_{\sigma\alpha\beta})^{\lambda\rho_n\sigma_n} - (k_n)_\alpha (\Gamma_{\rho\beta\sigma})^{\lambda\rho_n\sigma_n} \right\} (p_1)_{\alpha_1} (p_2)_{\beta_1} \cdots (p_{2b-1})_{\alpha_b} (p_{2b})_{\beta_b} (p_{2b+1})_\mu (p_{2b+1})_{\nu}  \\
            & + \sum\limits_{n=2}^\infty \int \prod\limits_{i=2}^n \cfrac{d^4 k_i}{(2\pi)^4} \, h_{\rho_i\sigma_i}(k_i) \!\!\!\! \prod\limits_{j=1}^{a+2b+1} \cfrac{d^4 p_j}{(2\pi)^4} \, \phi(p_j) \, (2\pi)^4 \delta\left( \sum k_i + \sum p_j \right) \\
            & \hspace{5pt} \times \cfrac12\, \kappa^n (-1)^{b} \Psi_{(a,b)} \left\{\left(g^{\mu\nu} g^{\alpha\beta} - g^{\mu\alpha} g^{\nu\beta} - g^{\mu\beta} g^{\nu\alpha}\right) g^{\rho\sigma} g^{\lambda\tau} \right\}^{\rho_1\sigma_1\cdots\rho_{n \!-\! 2} \sigma_{n \!-\! 2}} (k_{n \!-\! 1})_{\lambda_1} (\Gamma_{\lambda\mu\nu})^{\lambda_1 \rho_{n \!-\! 1} \sigma_{n \!-\! 1}} \\
            & \hspace{10pt} \times(k_n)_\lambda \left\{ (k_n)_\rho (\Gamma_{\sigma\alpha\beta})^{\lambda\rho_n\sigma_n} - (k_n)_\alpha (\Gamma_{\rho\beta\sigma})^{\lambda\rho_n\sigma_n} \right\} (p_1)_{\alpha_1} (p_2)_{\beta_1} \cdots (p_{2b-1})_{\alpha_b} (p_{2b})_{\beta_b} (p_{2b+1})_\tau \\
            & + \sum\limits_{n=2}^\infty \int \prod\limits_{i=2}^n \cfrac{d^4 k_i}{(2\pi)^4} \, h_{\rho_i\sigma_i}(k_i) \!\!\!\! \prod\limits_{j=1}^{a+2b+1} \cfrac{d^4 p_j}{(2\pi)^4} \, \phi(p_j) \, (2\pi)^4 \delta\left( \sum k_i + \sum p_j \right) \\
            & \hspace{5pt} \times \cfrac12\, \kappa^n (-1)^{b+1} \Psi_{(a,b)} \left\{\left(g^{\mu\nu} g^{\alpha\beta} - g^{\mu\alpha} g^{\nu\beta} - g^{\mu\beta} g^{\nu\alpha}\right) g^{\rho\sigma} g^{\lambda\tau} \right\}^{\rho_1\sigma_1\cdots\rho_{n \!-\! 2} \sigma_{n \!-\! 2}} (k_{n \!-\! 1})_{\lambda_1} (k_n)_{\lambda_2}\\
            & \hspace{10pt} \times \left\{ (\Gamma_{\rho\alpha\lambda})^{\lambda_1\rho_{n \!-\! 1}\sigma_{n\!-\!1}}(\Gamma_{\sigma\beta\tau})^{\lambda_2\rho_n\sigma_n} \!- (\Gamma_{\rho\alpha\beta})^{\lambda_1\rho_{n \!-\! 1}\sigma_{n\!-\!1}}(\Gamma_{\sigma\lambda\tau})^{\lambda_2\rho_n\sigma_n} \right\} (p_1)_{\alpha_1}  \cdots (p_{2b})_{\beta_b} (p_{2b+1})_\mu (p_{2b+1})_{\nu} \\
            & + \sum\limits_{n=3}^\infty \int \prod\limits_{i=3}^n \cfrac{d^4 k_i}{(2\pi)^4} \, h_{\rho_i\sigma_i}(k_i) \!\!\!\! \prod\limits_{j=1}^{a+2b+1} \cfrac{d^4 p_j}{(2\pi)^4} \, \phi(p_j) \, (2\pi)^4 \delta\left( \sum k_i + \sum p_j \right) \\
            & \hspace{5pt} \times \cfrac12\, \kappa^n (-1)^{b} \Psi_{(a,b)} \left\{\left(g^{\mu\nu} g^{\alpha\beta} - g^{\mu\alpha} g^{\nu\beta} - g^{\mu\beta} g^{\nu\alpha}\right) g^{\rho\sigma} g^{\lambda\tau} g^{\epsilon\omega} \right\}^{\rho_1\sigma_1\cdots\rho_{n \!-\! 3} \sigma_{n \!-\! 3}} (k_{n \!-\! 2})_{\lambda_1} (k_{n \!-\! 1})_{\lambda_2} (k_n)_{\lambda_3} \\
            & \hspace{10pt} \times\!\! \left\{ \! (\Gamma_{\!\rho\alpha\lambda})^{\lambda_2\rho_{n \!-\! 1}\sigma_{n\!-\!1}} \! (\Gamma_{\!\sigma\beta\tau})^{\lambda_3\rho_n\sigma_n} \!\!-\! (\Gamma_{\!\rho\alpha\beta})^{\lambda_2\rho_{n \!-\! 1}\sigma_{n\!-\!1}} \! (\Gamma_{\!\sigma\lambda\tau})^{\lambda_3\rho_n\sigma_n} \! \right\} \!\! (\Gamma_{\!\epsilon\mu\nu})^{\lambda_1\rho_{n\!-\!2}\sigma_{n\!-\!2}} (p_1)_{\alpha_1} \!\!\cdots\! (p_{2b})_{\beta_b} (p_{2b+1})_\omega .
        \end{split}
    \end{align}
\end{theorem}

First, we describe the structure involving scalar field derivatives to proceed with the second part of the interaction.
\begin{theorem}
    \begin{align}
        \begin{split}
            & (\square\phi)^3 - 3 \square\phi \left( \nabla_\mu \nabla_\nu \phi \right)^2 + 2 \left( \nabla_\mu \nabla_\nu \phi \right)^3 \\
            & = \left[ g^{\mu\nu} g^{\alpha\beta} g^{\rho\sigma} -3 \, g^{\mu\nu} g^{\alpha\rho} g^{\beta\sigma} + 2 g^{\nu\alpha} g^{\beta\rho} g^{\sigma\mu} \right] \pd_\mu\pd_\nu\phi \, \pd_\alpha\pd_\beta\phi\, \pd_\rho\pd_\sigma \phi \\
            & \hspace{10pt}+3 \left[ g^{\mu\nu} g^{\alpha\beta} g^{\rho\sigma} -3 \, g^{\mu\nu} g^{\alpha\rho} g^{\beta\sigma} + 2 g^{\nu\alpha} g^{\beta\rho} g^{\sigma\mu} \right] g^{\omega\tau} \Gamma_{\omega\mu\nu} \pd_\tau\phi \, \pd_\alpha\pd_\beta\phi\, \pd_\rho\pd_\sigma \phi \\
            & \hspace{10pt} +3 \left[ g^{\mu\nu} g^{\alpha\beta} g^{\rho\sigma} -3 \, g^{\mu\nu} g^{\alpha\rho} g^{\beta\sigma} + 2 g^{\nu\alpha} g^{\beta\rho} g^{\sigma\mu} \right] g^{\omega_1\tau_1} g^{\omega_2\tau_2} \Gamma_{\omega_1\mu\nu} \Gamma_{\omega_2\alpha\beta} \pd_{\tau_1} \phi \, \pd_{\tau_2}\phi\, \pd_\rho\pd_\sigma \phi \\
            & \hspace{10pt} + \left[ g^{\mu\nu} g^{\alpha\beta} g^{\rho\sigma} -3 \, g^{\mu\nu} g^{\alpha\rho} g^{\beta\sigma} + 2 g^{\nu\alpha} g^{\beta\rho} g^{\sigma\mu} \right] g^{\omega_1\tau_1} g^{\omega_2\tau_2} g^{\omega_3\tau_3} \Gamma_{\omega_1\mu\nu} \Gamma_{\omega_2\alpha\beta} \Gamma_{\omega_3\rho\sigma} \pd_{\tau_1}\phi \, \pd_{\tau_2}\phi\, \pd_{\tau_3} \phi .
        \end{split}
    \end{align}
\end{theorem}
\noindent The following theorem gives the perturbative structure of the second part of the interaction.
\begin{theorem}
    \begin{align}
        \begin{split}
            & \int d^4 x \sqrt{-g} \left[-\cfrac{1}{6} \, G_{5,X} \,\left\{ (\square\phi)^3 - 3 \square\phi \left( \nabla_\mu \nabla_\nu \phi \right)^2 + 2 \left( \nabla_\mu \nabla_\nu \phi \right)^3 \right\} \right] \\
            &=\sum\limits_{n=0}^\infty\int\prod\limits_{i=1}^n \cfrac{d^4 k_i}{(2\pi)^4} h_{\rho_i\sigma_i}(k_i) \prod\limits_{j=1}^{a+2b+1} \phi(p_j) (2\pi)^4 \delta\left(\sum k_i + \sum p_j\right)\\
            &\hspace{5pt}\times \cfrac16\,\kappa^n  (-1)^b\, b\, \Psi_{(a,b)} \left\{\sqrt{-g} \left[ g^{\mu\nu} g^{\alpha\beta} g^{\rho\sigma} -3 \, g^{\mu\nu} g^{\alpha\rho} g^{\beta\sigma} + 2 g^{\nu\alpha} g^{\beta\rho} g^{\sigma\mu} \right]\right\}^{\rho_1\sigma_1\cdots\rho_n\sigma_n}\\
            &\hspace{10pt}\times (p_1)_{\alpha_1} (p_2)_{\beta_2} \cdots (p_{2b \!-\!1})_{\alpha_b} (p_{2b})_{\beta_b} (p_{2b \!+\! 1})_\mu (p_{2b \!+\! 1})_\nu (p_{2b \!+\! 2})_\alpha (p_{2b \!+\! 2})_\beta (p_{2b \!+\! 3})_\rho (p_{2b \!+\! 3})_\sigma\\
            &+\sum\limits_{n=1}^\infty\int\prod\limits_{i=1}^n \cfrac{d^4 k_i}{(2\pi)^4} h_{\rho_i\sigma_i}(k_i) \prod\limits_{j=1}^{a+2b+1} \phi(p_j) (2\pi)^4 \delta\left(\sum k_i + \sum p_j\right) \\
            &\hspace{5pt}\times \cfrac13\,\kappa^n  (-1)^b\, b\, \Psi_{(a,b)} \left\{\sqrt{-g} \left[ g^{\mu\nu} g^{\alpha\beta} g^{\rho\sigma} -3 \, g^{\mu\nu} g^{\alpha\rho} g^{\beta\sigma} + 2 g^{\nu\alpha} g^{\beta\rho} g^{\sigma\mu} \right] g^{\omega\tau} \right\}^{\rho_1\sigma_1\cdots\rho_{n \!-\! 1} \sigma_{n \!-\! 1}} \\
            &\hspace{10pt}\times (p_1)_{\alpha_1} (p_2)_{\beta_2} \cdots (p_{2b \!-\!1})_{\alpha_b} (p_{2b})_{\beta_b} (k_n)_\lambda \left(\Gamma_{\omega\mu\nu}\right)^{\lambda\rho_n\sigma_n} (p_{2b \!+\! 1})_\tau (p_{2b \!+\! 2})_\alpha (p_{2b \!+\! 2})_\beta (p_{2b \!+\! 3})_\rho (p_{2b \!+\! 3})_\sigma
        \end{split}
    \end{align}
    \begin{align*}
        &+\sum\limits_{n=2}^\infty\int\prod\limits_{i=1}^n \cfrac{d^4 k_i}{(2\pi)^4} h_{\rho_i\sigma_i}(k_i) \prod\limits_{j=1}^{a+2b+1} \phi(p_j) (2\pi)^4 \delta\left(\sum k_i + \sum p_j\right) \\
        &\hspace{5pt}\times \cfrac13\,\kappa^n  (-1)^b\, b\, \Psi_{(a,b)} \left\{\sqrt{-g} \left[ g^{\mu\nu} g^{\alpha\beta} g^{\rho\sigma} -3 \, g^{\mu\nu} g^{\alpha\rho} g^{\beta\sigma} + 2 g^{\nu\alpha} g^{\beta\rho} g^{\sigma\mu} \right] g^{\omega_1\tau_1} g^{\omega_2\tau_2} \right\}^{\rho_1\sigma_1\cdots\rho_{n \!-\! 2} \sigma_{n \!-\! 2}} \\
        &\hspace{10pt}\times (p_1)_{\alpha_1} (p_2)_{\beta_2} \cdots (p_{2b \!-\!1})_{\alpha_b} (p_{2b})_{\beta_b} (k_{n \!-\! 1})_{\lambda_1} (k_n)_{\lambda_2} \left(\Gamma_{\omega_1\mu\nu}\right)^{\lambda_1\rho_{n \!-\! 1}\sigma_{n \!-\! 1}} \left(\Gamma_{\omega_2\alpha\beta}\right)^{\lambda_2\rho_n\sigma_n} \\
        & \hspace{15pt} \times (p_{2b \!+\! 1})_{\tau_1} (p_{2b \!+\! 2})_{\tau_2} (p_{2b \!+\! 3})_\rho (p_{2b \!+\! 3})_\sigma \\
        &+\sum\limits_{n=3}^\infty\int\prod\limits_{i=1}^n \cfrac{d^4 k_i}{(2\pi)^4} h_{\rho_i\sigma_i}(k_i) \prod\limits_{j=1}^{a+2b+1} \phi(p_j) (2\pi)^4 \delta\left(\sum k_i + \sum p_j\right) \\
        &\hspace{5pt}\times \cfrac16\,\kappa^n  (-1)^b\, b\, \Psi_{(a,b)} \left\{\sqrt{-g} \left[ g^{\mu\nu} g^{\alpha\beta} g^{\rho\sigma} -3 \, g^{\mu\nu} g^{\alpha\rho} g^{\beta\sigma} + 2 g^{\nu\alpha} g^{\beta\rho} g^{\sigma\mu} \right] g^{\omega_1\tau_1} g^{\omega_2\tau_2} g^{\omega_3\tau_3} \right\}^{\rho_1\sigma_1\cdots\rho_{n \!-\! 3} \sigma_{n \!-\! 3}} \\
        &\hspace{10pt}\times (p_1)_{\alpha_1} (p_2)_{\beta_2} \cdots (p_{2b \!-\!1})_{\alpha_b} (p_{2b})_{\beta_b} (k_{n \!-\! 2})_{\lambda_1} (k_{n \!-\! 1})_{\lambda_2} (k_n)_{\lambda_3} \\
        & \hspace{15pt} \times \left(\Gamma_{\omega_1\mu\nu}\right)^{\lambda_1\rho_{n \!-\! 2}\sigma_{n \!-\! 2}} \left(\Gamma_{\omega_2\alpha\beta}\right)^{\lambda_2\rho_{n \!-\! 1}\sigma_{n \!-\! 1} } \left(\Gamma_{\omega_3\rho\sigma}\right)^{\lambda_3\rho_n \sigma_n } (p_{2b \!+\! 1})_{\tau_1} (p_{2b \!+\! 2})_{\tau_2} (p_{2b \!+\! 3})_{\tau_3} .
    \end{align*}
\end{theorem}

The theorems explain the entire structure of the $G_5$ interaction class. However, we will not explicitly express the corresponding interaction vertex. Because the expression is exceptionally long, it would not be constructive.

\subsection{Dirac spinors}

The standard way to describe fermions in curved spacetime uses vierbein and $\gamma$ matrices \cite{Shapiro:2016pfm}. Firstly, $\gamma$ matrices are introduced to construct a representation of the Lorentz algebra. Dirac matrices satisfy the following relations:
\begin{align}
  \{ \gamma_m , \gamma_n \} \overset{\text{note}}{=} \gamma_m \gamma_n + \gamma_n \gamma_m &= 2\, \eta_{mn} \, , & \gamma^0 \gamma^m \gamma^0 & = (\gamma^m)^+ \,.
\end{align}
They form the following representation of the Lorentz algebra:
\begin{align}\label{Lorentz_algebra_spinor_representation}
  S_{mn} & \overset{\text{def}}{=}  \cfrac{i}{4} \,[\gamma_m , \gamma_n] \, , &  [S_{mn}, S_{ab}] &= -i \left( \eta_{ma} S_{nb} - \eta_{mb} S_{na} +\eta_{nb} S_{ma} -\eta_{na} S_{mb} \right) .
\end{align}
The standard vector representation of the Lorentz algebra is given by $J_{mn}$ matrices:
\begin{align}
  \left( J_{mn} \right)_{ab} = & \,i\, (\eta_{ma} \eta_{nb} - \eta_{mb} \eta_{na} ) \,, & [J_{mn}, J_{ab}] &= -i \left( \eta_{ma} J_{nb} - \eta_{mb} J_{na} +\eta_{nb} J_{ma} -\eta_{na} J_{mb} \right) .
\end{align}
The following relation connects these representations:
\begin{align}
  [\gamma_a , S_{mn} ] = \left( J_{mn}\right)_{ab} \,\gamma^b \,.
\end{align}
One can define Dirac spinors $\psi$ and $\overline{\psi}=\psi^+\gamma^0$ subject to Lorentz group action
\begin{align}\label{spinor_transformations_flat_spacetime}
        \delta\psi &= \cfrac{i}{2} \, \omega^{ab}\, S_{ab}\, \psi  \, ,& \delta\overline\psi &= -\cfrac{i}{2} \, \omega^{ab}\, \overline\psi S_{ab} \,,
\end{align}
and to construct the well-known Lorentz invariant action in a flat spacetime
\begin{align}
  \mathcal{A}_{s=\frac12,m\not =0} = \int d^4 x \left[\cfrac{i}{2} \left(\, \overline\psi\, \gamma^m \pd_m \psi - \pd_m \overline\psi\, \gamma^m \psi \right) - m \,\overline\psi \, \psi\right].
\end{align}

The action is generalised for curved spacetime via vierbein $\mathfrak{e}_m{}^\mu$. In a curved spacetime, one relates an arbitrary frame with a local inertial frame via vierbein $\mathfrak{e}_m{}^\mu$. Its Latin indices are subjected to the Lorentz transformation; Greek indices are subjected to the general coordinate transformations. In turn, vierbein satisfies the following normalisation condition:
\begin{align}
  \mathfrak{e}_m{}^\mu\, \mathfrak{e}_n{}^\nu\, g_{\mu\nu} =& \eta_{mn} , &\mathfrak{e}_m{}^\mu\, \mathfrak{e}_n{}^\nu\, \eta^{mn} =& g^{\mu\nu} .
\end{align}
They define an anti-symmetric spin-connection $(\Gamma_\mu)_{ab} = -(\Gamma_\mu)_{ba}$ that is related to the Christoffel symbols.
\begin{align}\label{spin-connection}
  \left(\Gamma_\mu \right)_{ab} = \mathfrak{e}_{a\alpha} \mathfrak{e}_b{}^\beta \Gamma^\alpha_{\mu\beta} + \mathfrak{e}_{a\sigma}\pd_\mu\mathfrak{e}_b{}^\sigma .
\end{align}
Since the vierbein $\mathfrak{e}_m{}^\mu$ connects Lorentz transformations and general coordinate transformations, one uses it to manipulate indices. This construction forms a set of $\gamma$ matrices that are subject to the general coordinate transformation group:
\begin{align}
  \gamma^\mu = \mathfrak{e}_m{}^\mu \,\gamma^m\,.
\end{align}
Spinor transformations \eqref{spinor_transformations_flat_spacetime} are generalized as follows
\begin{align}
    \delta\psi &= \cfrac{i}{2} \, \omega^{\mu\nu}\, S_{\mu\nu}\, \psi  = \cfrac{i}{2} \, \omega^{ab}(x)\, S_{ab}\, \psi \, ,& \delta\overline\psi &= -\cfrac{i}{2} \, \omega^{\mu\nu}\, \overline\psi\, S_{\mu\nu} = -\cfrac{i}{2} \, \omega^{ab}(x)\, \overline\psi\, S_{ab}\,.
\end{align}
This generalisation promotes transformation \eqref{spinor_transformations_flat_spacetime} to gauge transformations.

Regular derivatives are replaced with covariant derivatives defined as follows:
\begin{align}\label{spinor_transformations_curved_spacetime}
    \nabla_\mu \psi = \pd_\mu \psi -\cfrac{i}{2}\, \left(\Gamma_\mu\right)^{ab}\, S_{ab}\, \psi \, , & \nabla_\mu \overline\psi = \pd_\mu \overline\psi +\cfrac{i}{2}\, \left(\Gamma_\mu\right)^{ab}\,\overline\psi\, S_{ab} \, . 
\end{align}
Covariant derivatives defined in such a way satisfy the following relations:
\begin{align}
  \nabla_\mu (\,\overline\psi\, \psi ) &= \pd_\mu (\,\overline\psi\, \psi)  \, & \nabla_\mu \left(\,\overline\psi\, \gamma^\mu \psi \right) &= \pd_\mu \left(\, \overline\psi\, \gamma^\mu \psi\right) + \Gamma^\alpha_{\mu\beta} \, \overline\psi\, \gamma^\beta \,\psi \,.
\end{align}
This construction produces the following generalisation of the Dirac action
\begin{align}
    \mathcal{A}_{s=1/2} = \int d^4 x \left[ \sqrt{-g} \,\mathfrak{e}_m{}^\mu \, \frac12\, \left( i\,\overline{\psi}\, \gamma^m \,\nabla_\mu \psi - i\, \nabla_\mu\overline{\psi} \,\gamma^m \,\psi\right) - m\, \sqrt{-g}\, \overline{\psi} \,\psi \right] .
\end{align}
Here $m$ is the fermion mass,  $\mathfrak{e}_m{}^\mu$ is the vierbein, and $\nabla$ is the fermionic covariant derivative. We shall note that we omit a few steps in deriving this action. Namely, the action does not contain either Christoffel or spin connection. The reason for this is discussed in detail in \cite{Latosh:2022ydd}.

The following theorem describes the perturbative structure of the Dirac action.
\begin{theorem}
    \begin{align}
        \begin{split}
            \mathcal{A}_{s=1/2}  =& \int d^4 x \left[ \sqrt{-g} \, \mathfrak{e}_m{}^\mu \, \frac12\, \left( i\,\overline{\psi}\, \gamma^m \,\pd_\mu \psi - i\, \pd_\mu\overline{\psi} \,\gamma^m \,\psi\right) - m\, \sqrt{-g}\, \overline{\psi} \,\psi \right]\\
            =& \sum\limits_{n=0}^\infty \int \cfrac{d^4 p_1}{(2\pi)^4} \cfrac{d^4 p_2}{(2\pi)^4} \prod\limits_{i=0}^n \cfrac{d^4 k_i}{(2\pi)^4} \,(2\pi)^4 \delta\left(p_1+p_2+\sum k_i\right) h_{\rho_1\sigma_1}(k_1) \cdots h_{\rho_n\sigma_n} (k_n)\\
            &\times \kappa^n~ \overline{\psi}(p_2) \left[ \left(\sqrt{-g}\, \mathfrak{e}_m{}^\mu\right)^{\rho_1\sigma_1\cdots\rho_n\sigma_n}\,\frac12 \, (p_1-p_2)_\mu \gamma^m - \left(\sqrt{-g}\right)^{\rho_1\sigma_1\cdots\rho_n\sigma_n} m \right] \psi(p_1).
        \end{split}
    \end{align}
\end{theorem}

The background contribution describes the standard fermion propagator:
\begin{align}
  \begin{gathered}
    \begin{fmffile}{Diag02}
      \begin{fmfgraph}(30,30)
        \fmfleft{L}
        \fmfright{R}
        \fmf{fermion}{L,R}
      \end{fmfgraph}
    \end{fmffile}
  \end{gathered}
  \hspace{10pt}= i ~ \cfrac{p_m \,\gamma^m + m}{p^2 - m^2} \,. 
\end{align}
The other terms describe the following interaction rules:
\begin{align}\label{rule_F_1}
    \nonumber \\
    \begin{gathered}
        \begin{fmffile}{FR_F_1}
            \begin{fmfgraph*}(30,30)
                \fmfleft{L1,L2}
                \fmfright{R1,R2}
                \fmf{dbl_wiggly}{L1,V}
                \fmf{dbl_wiggly}{L2,V}
                \fmfdot{V}
                \fmf{fermion}{R1,V}
                \fmf{fermion}{V,R2}
                \fmffreeze
                \fmf{dots}{L1,L2}
                \fmflabel{$p_1$}{R1}
                \fmflabel{$p_2$}{R2}
                \fmflabel{$\rho_1\sigma_1$}{L1}
                \fmflabel{$\rho_n\sigma_n$}{L2}
            \end{fmfgraph*}
        \end{fmffile}
    \end{gathered}
    = i\, \kappa^n \,\left[  \cfrac12\, \left(\sqrt{-g}\, \mathfrak{e}_m{}^\mu\right)^{\rho_1\sigma_1\cdots\rho_n\sigma_n} \, (p_1-p_2)_\mu \gamma^m - \left(\sqrt{-g}\right)^{\rho_1\sigma_1\cdots\rho_n\sigma_n} m \right]. \\ \nonumber
\end{align}
This expression also holds in the $SU(N)$ Yang-Mills model considered below.

\subsection{Proca field}

For quantum field theory, the existence of a vector field mass is crucial. A vector field with zero mass admits a gauge symmetry, which means gauge fixing must be done. On the other hand, a vector field with a non-vanishing mass, the Proca field, has no gauge symmetry, and the gauge fixin is not an issue. We begin the discussion with the Proca case for the sake of simplicity.

The action describing a Proca field reads:
\begin{align}
    \mathcal{A}_{s=1,m} =& \int d^4 x \sqrt{-g} \left[ -\cfrac14\, F_{\mu\nu}F^{\mu\nu} + \cfrac{m^2}{2} \, A_\mu \,A^\mu \right]. 
\end{align}
Here $F_{\mu\nu} = \pd_\mu A_\nu - \pd_\nu A_\mu$ is the field tensor, $m$ is the vector field mass. The action admits the following factorisation:
\begin{align}
    \mathcal{A}_{s=1,m} =& \int d^4 x \left[ -\cfrac14\, \sqrt{-g}\, g^{\mu\alpha} g^{\nu\beta} F_{\mu\nu}F^{\alpha\beta} +  \sqrt{-g}\, g^{\mu\nu} \cfrac{m^2}{2} \, A_\mu \,A_\nu \right].
\end{align}
The following theorem describes the Fourier structure of the action.

\begin{theorem}
    \begin{align}
        \begin{split}
            \mathcal{A}_{s=1,m}=& \sum\limits_{n=0}^\infty \int \cfrac{d^4 p_1}{(2\pi)^4} \cfrac{d^4 p_2}{(2\pi)^4} \prod\limits_{i=0}^n \cfrac{d^4 k_i}{(2\pi)^4} \,(2\pi)^4 \delta\left(p_1+p_2+\sum k_i\right) h_{\rho_1\sigma_1}(k_1) \cdots h_{\rho_n\sigma_n} (k_n)  \\
            &\times \kappa^n \,\Bigg[\cfrac14\, \left(\sqrt{-g}\, g^{\mu\alpha}g^{\nu\beta}\right)^{\rho_1\sigma_1\cdots\rho_n\sigma_n}~(p_1)_{\mu_1} (p_2)_{\mu_2}  ~\big(F_{\mu\nu}\big)^{\mu_1\lambda_1}\big(F_{\alpha\beta}\big)^{\mu_2\lambda_2}  \\
            & \hspace{30pt}+ \cfrac{m^2}{2} \left( \sqrt{-g}\, g^{\lambda_1\lambda_2} \right)^{\rho_1\sigma_1\cdots\rho_n\sigma_n} \Bigg] \,A_{\lambda_1}(p_1) \, A_{\lambda_2}(p_2) .
        \end{split}
    \end{align}
    Here, we introduced the following notations:
    \begin{align}
        F_{\mu\nu} &= -i\,p_\sigma \, \big(F_{\mu\nu}\big)^{\sigma\lambda} \,A_{\lambda}(p) , & \big(F_{\mu\nu}\big)^{\sigma\lambda} &\overset{\text{def}}{=} \delta^\sigma_\mu \, \delta^\lambda_\nu - \delta^\sigma_\nu \, \delta^\lambda_\mu .
    \end{align}
\end{theorem}

This expression generates the Proca propagator:
\begin{align}\label{Proca_propagator}
  \begin{gathered}
    \begin{fmffile}{Diag03}
      \begin{fmfgraph*}(30,30)
        \fmfleft{L}
        \fmfright{R}
        \fmf{photon}{L,R}
        \fmflabel{$\mu$}{L}
        \fmflabel{$\nu$}{R}
      \end{fmfgraph*}
    \end{fmffile}
  \end{gathered}
  \hspace{13pt}=(-i)\,\cfrac{ ~ \eta_{\mu\nu} - \cfrac{p_\mu\,p_\nu}{m^2} ~ }{p^2 - m^2}\,.
\end{align}
The expression describing the interaction rules between gravitons and the Proca field kinetic energy is as follows:
\begin{align}\label{rule_V_1}
    \nonumber \\
    \begin{split}
        \begin{gathered}
            \begin{fmffile}{FR_V_1}
                \begin{fmfgraph*}(30,30)
                    \fmfleft{L1,L2}
                    \fmfright{R1,R2}
                    \fmf{dbl_wiggly}{L1,V}
                    \fmf{dbl_wiggly}{L2,V}
                    \fmfdot{V}
                    \fmf{photon}{R1,V}
                    \fmf{photon}{V,R2}
                    \fmffreeze
                    \fmf{dots}{L1,L2}
                    \fmflabel{$p_1,\lambda_1,m$}{R1}
                    \fmflabel{$p_2,\lambda_2,m$}{R2}
                    \fmflabel{$\rho_1\sigma_1$}{L1}
                    \fmflabel{$\rho_n\sigma_n$}{L2}
                \end{fmfgraph*}
            \end{fmffile}
        \end{gathered}
    \hspace{50pt}=  i\, \kappa^n \Bigg[ & \cfrac12\, \left(\sqrt{-g}\, g^{\mu\alpha}g^{\nu\beta}\right)^{\rho_1\sigma_1\cdots\rho_n\sigma_n} (p_1)_{\mu_1} (p_2)_{\mu_2} \big(F_{\mu\nu}\big)^{\mu_1\lambda_1} \big(F_{\alpha\beta}\big)^{\mu_2\lambda_2} \\
      & + m^2 \left( \sqrt{-g}\, g_{\lambda_1\lambda_2} \right)^{\rho_1\sigma_1\cdots\rho_n\sigma_n} \Bigg].
  \end{split}
\end{align}

\subsection{Vector field}

Before discussing the massless case, we shall recall the Faddeev-Popov prescription \cite{Faddeev:1967fc}. The following generating functional describes a quantum massless vector field:
\begin{align}
  \mathcal{Z} = \int \mathcal{D}[A] \exp\Big[i\,\mathcal{A}[A] \Big].
\end{align}
The integration space includes all conceivable configurations of the vector field. Firstly, we shall add a new term to the microscopic action:
\begin{align}
  \mathcal{Z} = \int \mathcal{D}[A]  \exp\Big[i\,\mathcal{A}[A] \Big] \int\mathcal{D}[\omega] \exp\left[\cfrac{i}{2}\,\epsilon\,\omega^2 \right] = \int\mathcal{D}[A]\mathcal{D}[\omega] \exp\left[ i\, \mathcal{A} +\cfrac{i}{2} \, \epsilon\,\omega^2\right] .
\end{align}
Here $\omega$ is an arbitrary scalar, and $\epsilon$ is a free gauge fixing parameter. The new contribution is a Gauss-like integral, so its introduction merely changes the omitted normalisation factor.

Secondly, we split the integration volume:
\begin{align}
  \int \mathcal{D}[A] =\int \mathcal{D}[\zeta] \int \mathcal{D}[\mathbb{A}] \delta\left( \mathcal{G} - \omega \right) \det\Delta\,.
\end{align}
In this expression, $\mathcal{G}$ is the gauge fixing condition, $\zeta$ is the gauge transformation parameter, and the new field variable $\mathbb{A}$ and the field variable $A$ are related as follows:
\begin{align}
  A_\mu = \mathbb{A}_\mu + \pd_\mu \zeta .
\end{align}
The integration over $\mathbb{A}$ is performed over all conceivable fields. Because of the $\delta$ function, only a single representative from each class of physically equivalent potentials contributes to the integral. Lastly, the Faddeev-Popov determinant $\det\Delta$ preserves the invariance of the integration measure. The corresponding differential operator $\Delta$ is defined as follows:
\begin{align}
  \Delta \overset{\text{def}}{=} \cfrac{\delta\mathcal{G}}{\delta \zeta}\,.
\end{align}

Finally, we perform integrations and obtain the following expression for the generating functional:
\begin{align}
  \begin{split}
    \mathcal{Z} &= \int \mathcal{D}[\mathbb{A}]\mathcal{D}[\omega]\mathcal{D}[\zeta] \left(\det\Delta\right) ~ \delta\left(\mathcal{G} - \omega \right) \exp\left[ i\, \mathcal{A} +\cfrac{i}{2} \, \epsilon\,\omega^2\right] \\
    &= \int \mathcal{D}[\mathbb{A}] \left(\det\Delta\right) \exp\left[ i\, \mathcal{A} +\cfrac{i}{2} \, \epsilon\,\mathcal{G}^2 \right] \\
    &=\int\mathcal{D}[c]\mathcal{D}[\overline{c}]\mathcal{D}[\mathbb{A}] \exp\left[ i \,\overline{c} \, \Delta \, c + i \, \mathcal{A} + \cfrac{i}{2}\,\epsilon \, \mathcal{G}^2 \right] .
  \end{split}
\end{align}
Here $\overline{c}$, $c$ are scalar anticommuting Faddeev-Popov ghosts introduced to account for the Faddeev-Popov determinant. We omit the integration over the gauge parameter $\zeta$ as it is irrelevant due to the normalisation.

This prescription generates a functional suitable for treating gauge models. We chose the standard Lorentz gauge fixing condition for simplicity. In a curved spacetime, it reads:
\begin{align}
  g^{\mu\nu}\,\nabla_\mu A_\nu =0 \leftrightarrow g^{\mu\nu}\,\pd_\mu A_\nu - g^{\mu\nu}\, \Gamma^\sigma_{\mu\nu} A_\sigma = 0 .
\end{align}

The following theorem describes the perturbative structure of the gauge invariant part of the action.
\begin{theorem}
\begin{align}
  \begin{split}
    \mathcal{A}_{s=1,m_\text{v}=0} =& \int d^4 x \sqrt{-g}\left[ -\cfrac14\,g^{\mu\alpha} g^{\nu\beta}\,F_{\mu\nu}F_{\alpha\beta}\right] \\
    =& \sum\limits_{n=0}^\infty\int\cfrac{d^4p_1}{(2\pi)^4}\cfrac{d^4p_2}{(2\pi)^4}\prod\limits_{i=1}^n \cfrac{d^4k_i}{(2\pi)^4} \,(2\pi)^4 \,\delta\big(p_1+p_2+\sum k_i \big) h_{\rho_1\sigma_1}(k_1)\cdots h_{\rho_n\sigma_n}(k_n)\\
    & \times \,\kappa^n \, \left[ \cfrac14\,\left(\sqrt{-g} \,g^{\mu\alpha}g^{\nu\beta}\right)^{\rho_1\sigma_1\cdots\rho_n\sigma_n}\,(p_1)_{\mu_1}(p_2)_{\mu_2}\left(F_{\mu\nu}\right)^{\mu_1\lambda_1} \left(F_{\alpha\beta}\right)^{\mu_2\lambda_2}\right] A_{\lambda_1}(p_1) A_{\lambda_2}(p_2).
  \end{split}
\end{align}
\end{theorem}
\noindent This expression matches the expression for the Proca field with $m=0$. 

The structure of the gauge fixing term is more sophisticated.
\begin{theorem}
\begin{align}\label{gauge_fixing_vector}
  \begin{split}
    \mathcal{A}_\text{gf} =& \int d^4 x \sqrt{-g}\left[ \cfrac{\epsilon}{2}\,\nabla_{\lambda_1} A^{\lambda_1} \, \nabla_{\lambda_2} A^{\lambda_2} \right]\\
    =& \cfrac{\epsilon}{2} \int d^4 x \left(\sqrt{-g}\,g^{\sigma_1\lambda_1}g^{\sigma_2\lambda_2}\right)\,\pd_{\sigma_1} A_{\lambda_1} \, \pd_{\sigma_2} A_{\lambda_2} -\epsilon \int d^4 x \left(\sqrt{-g}\,g^{\mu\nu}g^{\sigma_1\lambda_1}g^{\sigma_2\lambda_2}\right) \,\Gamma_{\sigma_1\mu\nu} \, A_{\lambda_1} \pd_{\sigma_2} A_{\lambda_2}\\
    &+\cfrac{\epsilon}{2}\int d^4 x \left( \sqrt{-g}\, g^{\mu\nu} g^{\alpha\beta} g^{\sigma_1\lambda_1} g^{\sigma_2\lambda_2} \right)\,\Gamma_{\sigma_1\mu\nu} \Gamma_{\sigma_2\alpha\beta} \,A_{\lambda_1} A_{\lambda_2}.
  \end{split}
\end{align}
\end{theorem}

The following theorem gives the Fourier structure of the gauge fixing term.
\begin{theorem}
\begin{align}
  \begin{split}
    \mathcal{A}_\text{gf} =&  \sum\limits_{n=0}^\infty\int\cfrac{d^4 p_1}{(2\pi)^4}\cfrac{d^4p_2}{(2\pi)^4}\prod\limits_{i=1}^n \cfrac{d^4k_i}{(2\pi)^4}\,(2\pi)^4\delta \big(p_1+p_2+\sum k_i \big) \,h_{\rho_1\sigma_1}(k_1) \cdots h_{\rho_n\sigma_n}(k_n) A_{\lambda_1}(p_1) A_{\lambda_2}(p_2)\\
    &\times \,\kappa^n\,\left(\sqrt{-g}\, g^{\mu_1\lambda_1} g^{\mu_2\lambda_2}\right)^{\rho_1\sigma_1\cdots\rho_n\sigma_n} \left[ -\cfrac{\epsilon}{2}\, (p_1)_{\mu_1} (p_2)_{\mu_2}\right] \\
    +&\sum\limits_{n=1}^\infty\int\cfrac{d^4 p_1}{(2\pi)^4}\cfrac{d^4p_2}{(2\pi)^4}\prod\limits_{i=1}^n \cfrac{d^4k_i}{(2\pi)^4}\,(2\pi)^4\delta \big(p_1+p_2+\sum k_i \big) \,h_{\rho_1\sigma_1}(k_1) \cdots h_{\rho_n\sigma_n}(k_n) A_{\lambda_1}(p_1) A_{\lambda_2}(p_2)\\
    & \times\,\kappa^n\,\left( \sqrt{-g}\,g^{\mu\nu} g^{\mu_1\lambda_1} g^{\mu_2\lambda_2} \right)^{\rho_2\sigma_2\cdots\rho_n\sigma_n}\Big[ \epsilon \, \left(\Gamma_{\mu_1\mu\nu}\right)^{\sigma\rho_1\sigma_1}\,(k_1)_\sigma \, (p_2)_{\mu_2}\Big] \\
    +& \sum\limits_{n=2}^\infty\int\cfrac{d^4 p_1}{(2\pi)^4}\cfrac{d^4p_2}{(2\pi)^4}\prod\limits_{i=1}^n \cfrac{d^4k_i}{(2\pi)^4}\,(2\pi)^4\delta \big(p_1+p_2+\sum k_i \big) \,h_{\rho_1\sigma_1}(k_1) \cdots h_{\rho_n\sigma_n}(k_n) \,A_{\lambda_1}(p_1) A_{\lambda_2}(p_2)\\
    &\times\,\kappa^n\,\left(\sqrt{-g}\,g^{\mu\nu} g^{\alpha\beta} g^{\mu_1\lambda_1} g^{\mu_2\lambda_2} \right)^{\rho_3\sigma_3\cdots\rho_n\sigma_n} \left[ - \cfrac{\epsilon}{2}\, (k_1)_{\tau_1}(k_2)_{\tau_2} \,\big(\Gamma_{\mu_1\mu\nu} \big)^{\tau_1\rho_1\sigma_1} \big( \Gamma_{\mu_2\alpha\beta}\big)^{\tau_2\rho_2\sigma_2}\right].
  \end{split}
\end{align}
The following notations for the Christoffel symbols were used.
\begin{align}
  \begin{split}
    \Gamma_{\mu\alpha\beta} =& \cfrac{\kappa}{2} \left[ \pd_\alpha h_{\beta\mu} + \pd_\beta h_{\alpha\mu} - \pd_\mu h_{\alpha\beta}\right] \Leftrightarrow \kappa \, (-i)\,p_\lambda \left(\Gamma_{\mu\alpha\beta}\right)^{\lambda\rho\sigma} h_{\rho\sigma}(p) \,, \\
    \left(\Gamma_{\mu\alpha\beta}\right)^{\lambda\rho\sigma} =& \cfrac12\left[ \delta^\lambda_\alpha I_{\beta\mu}{}^{\rho\sigma} + \delta^\lambda_\beta I_{\alpha\mu}{}^{\rho\sigma} - \delta^\lambda_\mu I_{\alpha\beta}{}^{\rho\sigma} \right].
  \end{split}
\end{align}
\end{theorem}

The background part of this expression corresponds to the standard propagator:
\begin{align}\label{Maxwell_propagator}
  \begin{gathered}
    \begin{fmffile}{Diag04}
      \begin{fmfgraph*}(30,30)
        \fmfleft{L}
        \fmfright{R}
        \fmf{photon}{L,R}
        \fmflabel{$\mu$}{L}
        \fmflabel{$\nu$}{R}
      \end{fmfgraph*}
    \end{fmffile}
  \end{gathered}
  \hspace{20pt} = i ~ \cfrac{ -\eta_{\mu\nu} + \left(1+ \cfrac{1}{\epsilon}\right) \cfrac{p_\mu \, p_\nu}{p^2} }{p^2} ~.
\end{align}
The corresponding interaction rule reads:
\begin{align}\label{rule_V_2}
     \\
    \begin{split}
    & \hspace{25pt}
        \begin{gathered}
            \begin{fmffile}{FR_V_2}
                \begin{fmfgraph*}(30,30)
                    \fmfleft{L1,L2}
                    \fmfright{R1,R2}
                    \fmf{dbl_wiggly}{L1,V}
                    \fmf{dbl_wiggly}{L2,V}
                    \fmfdot{V}
                    \fmf{photon}{R1,V}
                    \fmf{photon}{V,R2}
                    \fmffreeze
                    \fmf{dots}{L1,L2}
                    \fmflabel{$p_1,\lambda_1$}{R1}
                    \fmflabel{$p_2,\lambda_2$}{R2}
                    \fmflabel{$\rho_1\sigma_1,k_1$}{L1}
                    \fmflabel{$\rho_n\sigma_n,k_n$}{L2}
                \end{fmfgraph*}
            \end{fmffile}
        \end{gathered}
        \hspace{20pt}= i\, \kappa^n \Bigg[ \cfrac12\,\left(\sqrt{-g}\,g^{\mu\alpha}g^{\nu\beta}\right)^{\rho_1\sigma_1\cdots\rho_n\sigma_n}\,(p_1)_{\sigma_1}(p_2)_{\sigma_2} \big(F_{\mu\nu}\big)^{\sigma_1\lambda_1} \big(F_{\alpha\beta}\big)^{\sigma_2\lambda_2} \\ \\
        & -\epsilon \, \left(\sqrt{-g} \, g^{\mu_1\lambda_1} g^{\mu_2\lambda_2}\right)^{\rho_1\sigma_1\cdots\rho_n\sigma_n} (p_1)_{\mu_1} (p_2)_{\mu_2}  \\
        & +\epsilon\left\{ \left(\sqrt{-g}\,g^{\mu\nu}g^{\mu_1\lambda_1}g^{\mu_2\lambda_2}\right)^{\rho_2\sigma_2\cdots\rho_n\sigma_n}\,(k_1)_\sigma \left[ (p_2)_{\mu_2}\big(\Gamma_{\mu_1\mu\nu}\big)^{\sigma\rho_1\sigma_1} + (p_1)_{\mu_1}\big(\Gamma_{\mu_2\mu\nu}\big)^{\sigma\rho_1\sigma_1} \right] + \cdots \right\}\\
        & -\cfrac{\epsilon}{2} \Bigg\{ \left(\sqrt{-g}\,g^{\mu\nu}g^{\alpha\beta} g^{\mu_1\lambda_1} g^{\mu_2\lambda_2}\right)^{\rho_3\sigma_3\cdots\rho_n\sigma_n} \left[ (k_1)_{\tau_1}\,(k_2)_{\tau_2} \big(\Gamma_{\mu_1\mu\nu}\big)^{\tau_1\rho_1\sigma_1} \big(\Gamma_{\mu_2\alpha\beta}\big)^{\tau_2\rho_2\sigma_2} \right.\\
        & \hspace{190pt} \left.+ (k_1)_{\tau_2}\,(k_2)_{\tau_1} \big(\Gamma_{\mu_2\mu\nu}\big)^{\tau_1\rho_2\sigma_2} \big(\Gamma_{\mu_1\alpha\beta}\big)^{\tau_2\rho_1\sigma_1}   \right] + \cdots \Bigg\}\Bigg].
    \end{split}
\end{align}
The dots in this expression represent terms that create symmetry with respect to graviton momenta. The final term only affects vertices with two or more gravitons.

We treat the ghost sector of the theory as follows. The Faddeev-Popov differential operator $\Delta$ reduces to the D'Alamber operator in curved spacetime:
\begin{align}
  \Delta = \cfrac{\delta}{\delta \zeta} ~ \nabla_\mu\left(A^\mu + \nabla^\mu \zeta\right) = g^{\mu\nu}\nabla_\mu\nabla_\nu\,.
\end{align}
The ghost part of the functional describes a single massless scalar ghost coupled to gravity:
\begin{align}
    \mathcal{Z}_\text{ghost} =& \int\mathcal{D}[c]\mathcal{D}[\overline{c}] \exp\left[i \,\int d^4 x \sqrt{-g} \, \left( \overline{c}\, \square\, c  \right)\right] =  \int\mathcal{D}[c]\mathcal{D}[\overline{c}] \, \exp\left[ - i\, \int d^4 x \, \sqrt{-g}\, g^{\mu\nu} \,\nabla_\mu \overline{c} \, \nabla_\nu c\right].
\end{align}
The following theorem gives the corresponding perturbative expansion.
\begin{theorem}
    \begin{align}
        \begin{split}
            \mathcal{A}_\text{ghost} =& -\int d^4 x \, \sqrt{-g} \, g^{\mu\nu} \pd_\mu \overline{c} \,\pd_\nu c \\
            =&\sum\limits_{n=0}^\infty\int \cfrac{d^4 p_1}{(2\pi)^4}\cfrac{d^4 p_2}{(2\pi)^4} \prod\limits_{i=1}^n \cfrac{d^4 k_i}{(2\pi)^4} \, (2\pi)^4 \delta\left( p_1 + p_2 + \sum k_i \right) \, h_{\rho_1\sigma_1}(k_1) \cdots h_{\rho_n\sigma_n}(k_n)\\
            &\times\,\kappa^n\, \overline{c}(p_1) \left[ \left(\sqrt{-g} \, g^{\mu\nu} \right)^{\rho_1\sigma_1\cdots\rho_n\sigma_n} (p_1)_\mu (p_2)_\nu \right] c(p_2) .
        \end{split}
    \end{align}
\end{theorem}
We recover the standard ghost propagator:
\begin{align}\label{Maxwell_ghost_propagator}
  \begin{gathered}
    \begin{fmffile}{Diag05}
      \begin{fmfgraph*}(30,30)
        \fmfleft{L}
        \fmfright{R}
        \fmf{dots}{L,R}
      \end{fmfgraph*}
    \end{fmffile}
  \end{gathered}
  \hspace{20pt} =  i\, \cfrac{-1}{p^2} \,,
\end{align}
and in the following interaction rule:
\begin{align}\label{rule_Gh_1}
    \nonumber \\
    \begin{gathered}
        \begin{fmffile}{FR_Gh_1}
            \begin{fmfgraph*}(30,30)
                \fmfleft{L1,L2}
                \fmfright{R1,R2}
                \fmf{dbl_wiggly}{L1,V}
                \fmf{dbl_wiggly}{L2,V}
                \fmf{dots}{R1,V}
                \fmf{dots}{R2,V}
                \fmffreeze
                \fmfdot{V}
                \fmf{dots}{L1,L2}
                \fmflabel{$p_1$}{R1}
                \fmflabel{$p_2$}{R2}
                \fmflabel{$\rho_1\sigma_1$}{L1}
                \fmflabel{$\rho_n\sigma_n$}{L2}
            \end{fmfgraph*}
        \end{fmffile}
    \end{gathered} = i \, \kappa^n \,\left(\sqrt{-g} \, g^{\mu\nu} \right)^{\rho_1\sigma_1\cdots\rho_n\sigma_n} \,I_{\mu\nu}{}^{\alpha\beta} (p_1)_\alpha (p_2)_\beta . \\  \nonumber
\end{align}

Let us note again that the discussed ghosts are the standard Faddeev-Popov. In the context of gravity, there is one additional reason to account for them. The vertex in a diagram represents the interaction between gravitons and vectors, including physical and non-physical vector field polarisations. The coupling of Faddeev-Popov ghosts to gravity cancels out the energy contribution from non-physical polarisations.

\subsection{ SU(N) Yang-Mills}

Let us turn to the gravitational interaction of the $SU(N)$ Yang-Mills model. In the flat spacetime, the $SU(N)$ Yang-Mills model is given by the following action:
\begin{align}\label{SUNYM_flat}
    \begin{split}
        &\mathcal{A} = \int d^4 x \left[ \overline{\psi} \left(i\, \widehat{\mathcal{D}} - m\right) \psi - \cfrac14\, F^a_{\mu\nu} \, F^{a\mu\nu}\right] \\
        &= \int d^4 x \left[ \overline{\psi} (i\,\widehat{\pd} -m )\psi - \cfrac14\,\left(f^a_{\mu\nu}\right)^2 + \gs \,\overline{\psi} \widehat{A}\psi -\gs\,f^{abc}\pd_\mu A^a_\nu \,A^{b\mu} A^{c\nu} - \cfrac14\,\gs^2 \,f^{amn}\,f^{aij} \,(A^m\!\!\cdot\!\! A^i) (A^n\!\!\cdot\!\! A^j) \right] .
    \end{split}
\end{align}
The fermion covariant derivative is defined in the standard way:
\begin{align}
  \mathcal{D}_\mu \psi = \pd_\mu \psi - i\,\gs\,A_\mu\,\psi .
\end{align}
Field tensor $F_{\mu\nu}$ reads
\begin{align}
  F_{\mu\nu} = \pd_\mu A_\nu - \pd_\nu A_\mu -i\,\gs [A_\mu , A_\nu] .
\end{align}
The gauge field $A_\mu$ takes value in $SU(N)$ algebra:
\begin{align}
  A_\mu = A^a_\mu \, T^a ,
\end{align}
where $T^a$ are generators. In turn, field tensor components are given by the following:
\begin{align}
  F^a_{\mu\nu} = \pd_\mu A^a_\nu - \pd_\nu A^a_\mu + \gs\, f^{abc} \, A^b_\mu A^c_\nu \,.
\end{align}
Here $f^{abc}$ are the structure constants of the algebra:
\begin{align}
  [T^a, T^b] = i\, f^{abc} \, T^c\,.
\end{align}

Action \eqref{SUNYM_flat} is generalised for the curved spacetime case as follows. One replaces all derivatives with covariant derivatives and makes explicit the invariant volume factor. Such a generalisation results in the following action:
\begin{align}
  \mathcal{A} = \int d^4 x \, \sqrt{-g} \left[ \,\overline{\psi} \left( i\, \mathfrak{e}_m{}^\mu\,\gamma^m\,  \mathcal{D}_\mu  - m  \right) \psi -\cfrac14\,F^a_{\mu\nu} \,F^{a\mu\nu} \right]. 
\end{align}
Here $\mathfrak{e}_m{}^\mu$ is a vierbein. The covariant derivative for fermions now reads
\begin{align}
  \mathcal{D}_\mu \psi = \nabla_\mu \psi - i\,\gs\,A_\mu \,\psi ,
\end{align}
with $\nabla_\mu$ begin the part accounting for the spacetime curvature via the spin connection. The field tensor $F_{\mu\nu}$ is defined via covariant derivatives. However, it does not involve Christoffel symbols due to its structure:
\begin{align}
    F_{\mu\nu} = \nabla_\mu A_\nu - \nabla_\nu A_\mu - i\,\gs\,[A_\mu,A_\nu] =\pd_\mu A_\nu - \pd_\nu A_\mu -i\,\gs [A_\mu , A_\nu]\,.
\end{align}
This results in the following $SU(N)$ Yang-Mills action in curved spacetime:
\begin{align}\label{SUNYM_curve}
  \begin{split}
    \mathcal{A} = &\int d^4 x \sqrt{-g} \Bigg[ \overline{\psi} \left( i\, \mathfrak{e}_m{}^\mu \, \gamma^m \, \nabla_\mu- m \right) \psi - \cfrac14\, \left(f^a_{\mu\nu}\right)^2\\
      &+ \gs \,\overline{\psi} \left( \mathfrak{e}_m{}^\mu \gamma^m \right) \psi\,A_\mu -g^{\mu\nu} g^{\alpha\beta} \, \gs \,f^{abc} \pd_\mu A^a_\alpha A^b_\nu A^c_\beta -\cfrac14\,\gs^2\,f^{amn}\,f^{aij}\,g^{\mu\nu} g^{\alpha\beta} \,A^m_\mu\,A^i_\nu\,A^n_\alpha\,A^j_\beta \Bigg]\,.
  \end{split}
\end{align}

The perturbative quantisation of kinetic parts of the action is discussed above. The following theorem gives the perturbative expansion for the coupling of fermions to a gauge vector.
\begin{theorem}
    \begin{align}\label{ffv_vertex}
        \begin{split}
            &\int d^4 x \sqrt{-g} \,\gs\, \overline{\psi} \left(\mathfrak{e}_m{}^\mu \gamma^m \right) \psi \, A_\mu \\
            &=\sum\limits_{n=0}^\infty \int \cfrac{d^4 p_1}{(2\pi)^4}\,\cfrac{d^4 p_2}{(2\pi)^4}\,\cfrac{d^4 k}{(2\pi)^4} \,\prod\limits_{i=0}^n \cfrac{d^4 k_i}{(2\pi)^4}\, (2\pi)^4\,\delta\left(p_1+p_2+k+\sum k_i \right) \,h_{\rho_1\sigma_1}(k_1)\cdots h_{\rho_n\sigma_n}(k_n) \\
            &\hspace{10pt}\times \, \kappa^n \, \overline{\psi}(p_2) \left[\gs \,\gamma^m\,T^a \, \left(\sqrt{-g}\,\mathfrak{e}_m{}^\mu\right)^{\rho_1\sigma_1\cdots\rho_n\sigma_n} \right] \psi(p_1) \, A^a_\mu(k)\,.
        \end{split}
    \end{align}
\end{theorem}
This expression produces the following Feynman rule:
\begin{align}\label{rule_QQG}
   \nonumber \\
  \begin{gathered}
    \begin{fmffile}{FR_QQG}
      \begin{fmfgraph*}(30,30)
        \fmfleft{L1,L2,L3}
        \fmftop{T}
        \fmfbottom{B}
        \fmfright{R}
        \fmf{fermion}{B,V,T}
        \fmf{gluon}{V,R}
        \fmf{phantom}{V,L2}
        \fmfdot{V}
        \fmffreeze
        \fmf{dbl_wiggly}{L1,V}
        \fmf{dbl_wiggly}{L3,V}
        \fmf{dots}{L1,L3}
        \fmflabel{$\rho_1\sigma_1$}{L1}
        \fmflabel{$\rho_n\sigma_n$}{L3}
        \fmflabel{$\mu,a$}{R}
      \end{fmfgraph*}
    \end{fmffile}
  \end{gathered} \hspace{25pt} = i\, \kappa^n\, \gs\,\gamma^m\,T^a\,\left(\sqrt{-g} \,\mathfrak{e}_m{}^\mu \right)^{\rho_1\sigma_1\cdots\rho_n\sigma_n} . \\ \nonumber
\end{align}
The perturbative expansion for the cubic term in gauge vectors takes a similar form.
\begin{theorem}
    \begin{align}\label{vvv_vertex}
        \begin{split}
            &\int d^4 x \sqrt{-g} \, (-\gs)\,f^{abc}\, g^{\mu\nu} g^{\alpha\beta}\, \pd_\mu A^a_\alpha A^b_\nu A^c_\beta \\
            & = \sum\limits_{n=0}^\infty \int \cfrac{d^4 p_1}{(2\pi)^4} \, \cfrac{d^4 p_2}{(2\pi)^4} \, \cfrac{d^4 p_3}{(2\pi)^4} \,\prod\limits_{i=1}^n\cfrac{d^4 k_i}{(2\pi)^4}\, (2\pi)^4\,\delta\left(p_1 + p_2 + p_3 + \sum k_i\right) \,h_{\rho_1\sigma_1}(k_1) \cdots h_{\rho_n\sigma_n}(k_n) \\
            &\hspace{10pt}\times\,\kappa^n\,\left[ (-i\, \gs) \,f^{abc} \, (p_1)_{\sigma}\, \left(\sqrt{-g} \,g^{\mu_1\mu_3} g^{\mu_2\sigma}\right)^{\rho_1\sigma_1\cdots\rho_n\sigma_n} \right] \, A^a_{\mu_1}(p_1)\,A^b_{\mu_2}(p_2)\,A^c_{\mu_3}(p_3) \,.
        \end{split}
    \end{align}
\end{theorem}
This expression produces the following rule:
\begin{align}\label{rule_GGG}
    \nonumber \\
    \begin{split}
    &
        \begin{gathered}
            \begin{fmffile}{FR_GGG}
                \begin{fmfgraph*}(30,30)
                    \fmfleft{L1,L2,L3}
                    \fmfright{R1,R2,R3}
                    \fmf{dbl_wiggly,tension=2}{L1,V}
                    \fmf{dbl_wiggly,tension=2}{L3,V}
                    \fmf{gluon,tension=0.5}{V,R1}
                    \fmf{gluon,tension=0.5}{V,R2}
                    \fmf{gluon,tension=0.5}{V,R3}
                    \fmffreeze
                    \fmf{dots}{L1,L3}
                    \fmfdot{V}
                    \fmflabel{$\rho_1\sigma_1$}{L1}
                    \fmflabel{$\rho_n\sigma_n$}{L3}
                    \fmflabel{$\mu_1,a,p_1$}{R1}
                    \fmflabel{$\mu_2,b,p_2$}{R2}
                    \fmflabel{$\mu_3,c,p_3$}{R3}
                \end{fmfgraph*}
            \end{fmffile}
        \end{gathered}
        \hspace{40pt} = \kappa^n \,\gs\,f^{abc} \Big[ (p_1-p_2)_{\sigma}\left(\sqrt{-g} \,g^{\mu_1\mu_2} g^{\mu_3\sigma}\right)^{\rho_1\sigma_1\cdots\rho_n\sigma_n}  \\ \\
        & +(p_3-p_1)_{\sigma} \left(\sqrt{-g} \,g^{\mu_1\mu_3} g^{\mu_2\sigma}\right)^{\rho_1\sigma_1\cdots\rho_n\sigma_n} +(p_2-p_3)_{\sigma} \left(\sqrt{-g} \,g^{\mu_2\mu_3} g^{\mu_1\sigma}\right)^{\rho_1\sigma_1\cdots\rho_n\sigma_n} \Big] .
    \end{split}
\end{align}
Lastly, the four-vector coupling term has the following perturbative structure.
\begin{theorem}
    \begin{align}\label{vvvv_vertex}
        \begin{split}
            &\int d^4 x \sqrt{-g} \left(-\cfrac14\,\gs^2\right)\, f^{amn} f^{aij} \,g^{\mu\nu} g^{\alpha\beta} \, A^m_\mu \, A^i_\nu\, A^n_\alpha\, A^j_\beta\\
            &=\sum\limits_{n=0}^\infty\int \cfrac{d^4 p_1}{(2\pi)^4}\,\cfrac{d^4 p_2}{(2\pi)^4}\,\cfrac{d^4 p_3}{(2\pi)^4}\,\cfrac{d^4 p_4}{(2\pi)^4}\,\prod\limits_{i=0}^n \cfrac{d^4 k_i}{(2\pi)^4}\,(2\pi)^4 \,\delta\left(p_1+p_2+p_3+p_4+\sum k_i \right) h_{\rho_1\sigma_1}(k_1)\cdots h_{\rho_n\sigma_n}(k_n) \\
            &\hspace{10pt}\times\left( -\cfrac14 \right) \gs^2 \kappa^n f^{amn} f^{aij} \left(\sqrt{-g}\, g^{\mu_1\mu_3} g^{\mu_2\mu_4}\right)^{\rho_1\sigma_1\cdots\rho_n\sigma_n}  A^m_{\mu_1}(p_1)A^n_{\mu_2}(p_2) A^i_{\mu_3}(p_3)A^j_{\mu_4}(p_4) .
        \end{split}
    \end{align}
\end{theorem}
This results in the following interaction rule:
\begin{align}\label{rule_GGGG}
    \nonumber \\
    \begin{split}
        &
        \begin{gathered}
            \begin{fmffile}{FR_GGGG}
                \begin{fmfgraph*}(30,30)
                    \fmfleft{L1,L2}
                    \fmfright{R1,R2,R3,R4}
                    \fmf{gluon,tension=.5}{R1,V}
                    \fmf{gluon,tension=.5}{R2,V}
                    \fmf{gluon,tension=.5}{R3,V}
                    \fmf{gluon,tension=.5}{R4,V}
                    \fmf{dbl_wiggly,tension=2}{L1,V}
                    \fmf{dbl_wiggly,tension=2}{L2,V}
                    \fmfdot{V}
                    \fmffreeze
                    \fmf{dots}{L1,L2}
                    \fmflabel{$\rho_1\sigma_1$}{L1}
                    \fmflabel{$\rho_n\sigma_n$}{L2}
                    \fmflabel{$\mu_1,a_1$}{R1}
                    \fmflabel{$\mu_2,a_2$}{R2}
                    \fmflabel{$\mu_3,a_3$}{R3}
                    \fmflabel{$\mu_4,a_4$}{R4}
                \end{fmfgraph*}
            \end{fmffile}
        \end{gathered}
        \hspace{30pt}=  -i\,\gs^2 \kappa^n \Bigg[ f^{a_1 a_4 s} f^{a_2 a_3 s} \left( \left(\sqrt{-g}\, g^{\mu_1\mu_2}g^{\mu_3\mu_4}\right)^{\rho_1\cdots\sigma_n}-\left(\sqrt{-g}\, g^{\mu_1\mu_3}g^{\mu_2\mu_4}\right)\right)^{\rho_1\cdots\sigma_n} \\
        &\hspace{60pt}+f^{a_1 a_3 s} f^{a_2 a_4 s} \left( \left(\sqrt{-g}\, g^{\mu_1\mu_2}g^{\mu_3\mu_4}\right)^{\rho_1\cdots\sigma_n}-\left(\sqrt{-g}\, g^{\mu_1\mu_4}g^{\mu_2\mu_3}\right)\right)^{\rho_1\cdots\sigma_n} \\
        &\hspace{60pt} +f^{a_1 a_2 s} f^{a_3 a_4 s} \left( \left(\sqrt{-g}\, g^{\mu_1\mu_3}g^{\mu_2\mu_4}\right)^{\rho_1\cdots\sigma_n}-\left(\sqrt{-g}\, g^{\mu_1\mu_4}g^{\mu_2\mu_3}\right)\right)^{\rho_1\cdots\sigma_n} \Bigg].
    \end{split}
\end{align}

Finally, we are going to discuss gauge fixing. The Yang-Mills action respects the following gauge transformations:
\begin{align}
  \begin{split}
    \delta \psi =& i\,\theta^a \, T^a \psi ,\\
    \delta A_\mu =& i\,\theta^a \,[T^a, A_\mu] + \cfrac{1}{\gs} \, \pd_\mu \theta^a\,T^a ,\\
    \delta A^a_\mu =& \cfrac{1}{\gs} \left[\pd_\mu \theta^a - g \, f^{abc} \, \theta^b \, A^c_\mu \right] .
  \end{split}
\end{align}
Here $\theta^a$ are the gauge parameters. The standard Lorentz gauge fixing conditions simplify calculations in a flat spacetime.
\begin{align}
  \pd^\mu A^a_\mu =0 .
\end{align}
For the case of curved spacetime, the covariant derivative replaces the standard derivative, which produces the new form of the Lorentz gauge fixing conditions:
\begin{align}
  g^{\mu\nu} \nabla_\mu A^a_\nu =0 .
\end{align}
Introducing this gauge fixing term will bring the kinetic part of the vector field to the same form discussed in the previous section. 

The ghost action is defined by the Faddeev-Popov determinant obtained from the gauge fixing condition:
\begin{align}
  \det\left[ \cfrac{\delta}{\delta \theta^b}\,\left\{ g^{\mu\nu} \nabla_\mu A^a_\nu \right\}  \right]  = \det\left[ \cfrac{1}{\gs}\,g^{\mu\nu} \, \nabla_\mu \left(\delta^{ab}\nabla_\nu  - \gs \,f^{abc}  \, A^c_\nu \right)\right].
\end{align}
It results in the following action:
\begin{align}
  \mathcal{A}_\text{FP} = \int d^4 x \left[  - g^{\mu\nu}\, \nabla_\mu\overline{c}^a \nabla_\nu c^a + \gs \,g^{\mu\nu} \nabla_\mu \overline{c}^a f^{abc} c^b A_\nu^c  \right].
\end{align}
The action's kinetic part is similar to a massless vector field. The section describing the interaction between ghosts, vectors, and gravitons allows for a perturbative expansion.
\begin{theorem}
    \begin{align}
        \begin{split}
            &\int d^4 x \sqrt{-g} \left[ \gs\,\pd_\mu \overline{c}^a \, f^{abc} \, c^b \, A^{c\,\mu} \right]\\
            &=\sum\limits_{n=0}^\infty \int \cfrac{d^4 p_1}{(2\pi)^4} \cfrac{d^4 p_2}{(2\pi)^4} \cfrac{d^4 k}{(2\pi)^4}\prod\limits_{i=0}^n \cfrac{d^4 k_i}{(2\pi)^4}\, (2\pi)^4 \delta\left(p_1+p_2+k+\sum k_i\right)\, h_{\rho_1\sigma_1}(k_1)\cdots h_{\rho_n\sigma_n}(k_n)\\
            &\hspace{15pt}\times\, i\,\kappa^n\,\gs\,(p_1)_\nu \, f^{abc} \, \overline{c}^a(p_1) \,c^b(p_2) \, \left(\sqrt{-g}\,g^{\mu\nu}\right)^{\rho_1\sigma_1\cdots\rho_n\sigma_n}\,\,A^c_\mu(k).
        \end{split}
    \end{align}
\end{theorem}
This expression produced the following rule:
\begin{align}\label{rule_GhGhG}
   \nonumber \\
  \begin{gathered}
    \begin{fmffile}{FR_GhGhG}
      \begin{fmfgraph*}(30,30)
        \fmfleft{L1,L2}
        \fmfright{R1,R2,R3}
        \fmf{dbl_wiggly,tension=2}{L1,V}
        \fmf{dbl_wiggly,tension=2}{L2,V}
        \fmf{gluon,tension=.5}{V,R2}
        \fmf{dots_arrow,tension=.5}{R1,V}
        \fmf{dots_arrow,tension=.5}{V,R3}
        \fmfdot{V}
        \fmffreeze
        \fmf{dots}{L1,L2}
        \fmflabel{$\rho_1\sigma_1$}{L1}
        \fmflabel{$\rho_n\sigma_n$}{L2}
        \fmflabel{$\mu,c$}{R2}
        \fmflabel{$b$}{R1}
        \fmflabel{$p_1,a$}{R3}
      \end{fmfgraph*}
    \end{fmffile}
  \end{gathered}
  \hspace{20pt}=& - \kappa^n\,g_s\,f^{abc} \, (p_1)_\nu \,\left(\sqrt{-g}\,g^{\mu\nu}\right)^{\rho_1\sigma_1\cdots\rho_n\sigma_n}. \\ \nonumber
\end{align}

\subsection{General relativity}

General relativity is invariant with respect to local transformations spawned by coordinate transformations, so a gauge fixing procedure shall be performed. The perturbative approach describes gravity as small metric perturbations over a flat background. Therefore, it may seem that the theory is reduced to a gauge theory of rank-$2$ symmetric tensor, but this is not the case.

One shall distinguish a geometric theory from a theory of a rank-$2$ symmetric tensor. The difference dictates how to treat the gauge fixing condition. First, let us consider the case of a symmetric $h_{\mu\nu}$ tensor theory that admits the following gauge symmetry:
\begin{align}
  \delta h_{\mu\nu} = \pd_\mu \zeta_\nu + \pd_\nu \zeta_\mu .
\end{align}
The $h_{\mu\nu}$ tensor is fundamental to this theory, and a gauge fixing condition can be expressed solely in terms of it. For example, one can use the following gauge fixing condition:
\begin{align}\label{naive_gauge_fixing}
  \pd_\mu h^{\mu\nu}  - \cfrac12\,\pd^\nu h = 0\,.
\end{align}
This condition, along with others, determines the composition of the Faddeev-Popov ghosts. Because $h_{\mu\nu}$ is the fundamental object of the theory, the structure of divergences can only be expressed in terms of $h_{\mu\nu}$ alone. All geometric quantities, such as the Riemann tensor, Ricci tensor, scalar curvature, and others, are expressed via small metric perturbations, but the opposite is not true. Operators given in terms of $h_{\mu\nu}$ alone may not represent any geometric quantities. Therefore, in a gauge theory of a rank-$2$ symmetric tensor, one can expect to find divergencies that cannot be described by geometric quantities, which makes it a non-geometric theory.

On the contrary, gauge transformations are generated by coordinate frame transformations within the geometric approach. This has two immediate implications. Firstly, within a geometrical theory, gauge transformations are given by the so-called Lie derivatives:
\begin{align}
  \delta g_{\mu\nu} \overset{\text{def}}{=} \mathcal{L}_\zeta g_{\mu\nu} = \nabla_\mu \zeta_\nu + \nabla_\nu \zeta_\mu .
\end{align}
Here, $\mathcal{L}_\zeta$ is the Lie derivative with respect to an arbitrary vector field $\zeta$, which contains the gauge transformation parameters. Secondly, the gauge fixing condition must be expressed in geometrical quantities. Thus, gauge fixing conditions \eqref{naive_gauge_fixing} are inconsistent with the geometrical approach and cannot be implemented. Instead, we use the following gauge fixing conditions:
\begin{align}\label{the_gravity_gauge_fixing}
  \mathcal{G}^\mu \overset{\text{note}}{=} g^{\alpha\beta} \Gamma^\mu_{\alpha\beta} =0 .
\end{align}
Combined with the perturbative expansion given by equation \eqref{the_perturbative_expansion}, the gauge fixing conditions represented by equation \eqref{the_gravity_gauge_fixing} produce an infinite series.
\begin{align}
  \mathcal{G}^\nu = \cfrac{\kappa}{2}\,g^{\mu\nu} g^{\alpha\beta} \left[ \pd_\alpha h_{\beta\mu} + \pd_\beta h_{\alpha\mu} - \pd_\mu h_{\alpha\beta}\right] =  \kappa \left[\pd_\mu h^{\mu\nu} - \cfrac12\,\pd^\nu h \right] + \mathcal{O}(\kappa^2) ,
\end{align}
The infinite expansion defines the ghost sector, and the leading term reproduces naive gauge fixing \eqref{naive_gauge_fixing}.

The difference between geometrical theories of gravity and a gauge theory of $h_{\mu\nu}$ tensor is marked by the need to use gauge fixing condition \eqref{the_gravity_gauge_fixing}, and not conditions \eqref{naive_gauge_fixing}. Proceeding with The Faddeev-Popov prescription, firstly, we shall note that the gauge fixing condition $\mathcal{G}^\mu$ defined by \eqref{the_gravity_gauge_fixing} is a vector with mass dimension $+1$. Thus, one shall introduce an additional dimensional parameter in the gauge fixing term:
\begin{align}\label{Hilbert_gauge_fixed}
  \mathcal{A}_{\text{H}+\text{gf}} = \int d^4 x \sqrt{-g} \left[ -\cfrac{2}{\kappa^2}\,R + \cfrac{\epsilon}{2\,\kappa^2} \, g_{\mu\nu} \,\mathcal{G}^\mu \mathcal{G}^\nu \right].
\end{align}
Secondly, the Faddeev-Popov ghosts are also vectors. The structure of their action is defined by the variation of the gauge fixing term given by equation \eqref{the_gravity_gauge_fixing}:
\begin{align}
  \delta G^\mu = \mathcal{L}_\zeta \left[ g^{\alpha\beta} \, \Gamma^\mu_{\alpha\beta} \right] = \square \zeta^\mu - 2\, \Gamma^\mu_{\alpha\beta} \, \nabla^\alpha\zeta^\beta + R^\mu{}_\nu \zeta^\nu 
\end{align}
with $R_{\mu\nu}$ begin the Ricci tensor. Consequently, the ghost action reads:
\begin{align}
  \mathcal{A}_\text{ghost} = \int d^4 x \sqrt{-g} \left[ -g^{\alpha\beta} g^{\mu\nu} \nabla_\alpha \overline{c}_\mu \nabla_\beta c_\nu - 2\,\Gamma^\mu_{\alpha\beta} \,\overline{c}_\mu \,\nabla^\alpha c^\beta + R_{\mu\nu} \, \overline{c}^\mu \,c^\nu \right].
\end{align}
In all other respects, the handling of the Faddeev-Popov ghosts is unchanged.

The standard perturbative expansion generates the interaction rules. The structure of graviton interactions is given by action \eqref{Hilbert_gauge_fixed}:
\begin{align}
  \begin{split}
    \mathcal{A}_{\text{H}+\text{gf}} &= \int d^4 x \sqrt{-g} \left[ -\cfrac{2}{\kappa^2}\,R + \cfrac{\epsilon}{2\,\kappa^2}\, g_{\mu\nu}\, g^{\alpha\beta}\, g^{\rho\sigma} \,\Gamma^\mu_{\alpha\beta}\,\Gamma^\nu_{\rho\sigma}\right] = \\
    & = \int d^4 x \sqrt{-g} \, g^{\mu\nu} g^{\alpha\beta} g^{\rho\sigma} \left(-\cfrac{2}{\kappa^2}\right) \left[ \Gamma_{\alpha\mu\rho}\Gamma_{\sigma\nu\beta} - \Gamma_{\alpha\mu\nu} \Gamma_{\rho\beta\sigma} - \cfrac{\epsilon}{4} \,\Gamma_{\mu\alpha\beta} \Gamma_{\nu\rho\sigma} \right] .
  \end{split}
\end{align}
It admits the following perturbative expansion:
\begin{theorem}
    \begin{align}
        \begin{split}
            \mathcal{A}_{\text{H}+\text{gf}} =& \sum\limits_{n=0}^\infty\int\cfrac{d^4 p_1}{(2\pi)^4} \cfrac{d^4 p_2}{(2\pi)^4} \prod\limits_{i=1}^n \cfrac{d^4 k_i}{(2\pi)^4}\,(2\pi)^4\,\delta\Big(p_1+p_2+\sum k_i \Big) \, h_{\rho_1\sigma_1}(k_1) \cdots h_{\rho_n\sigma_n}(k_n)\\
            &\times \left(2\,\kappa^n \right)\left( \sqrt{-g} \, g^{\mu\nu} g^{\alpha\beta} g^{\rho\sigma} \right)^{\rho_1\sigma_1\cdots\rho_n\sigma_n} (p_1)_{\lambda_1} (p_2)_{\lambda_2} \, h_{\mu_1\nu_1}(p_1) h_{\mu_2\nu_2}(p_2) \\
            & \times \Bigg[ \left(\Gamma_{\alpha\mu\rho}\right)^{\lambda_1\mu_1\nu_1} \left(\Gamma_{\sigma\nu\beta}\right)^{\lambda_2\mu_2\nu_2} - \left(\Gamma_{\alpha\mu\nu} \right)^{\lambda_1\mu_1\nu_1} \left( \Gamma_{\rho\beta\sigma}\right)^{\lambda_2\mu_2\nu_2} - \cfrac{\epsilon}{4} \left( \Gamma_{\mu\alpha\beta} \right)^{\lambda_1\mu_1\nu_1} \left(\Gamma_{\nu\rho\sigma} \right)^{\lambda_2\mu_2\nu_2}\Bigg].
        \end{split} 
    \end{align}
\end{theorem}
The following formula gives the complete expression for the graviton vertex:
\begin{align}\label{Rule_Gravitons}
    \nonumber \\
    \begin{split}
        & \hspace{10pt}
        \begin{gathered}
            \begin{fmffile}{FR_h}
                \begin{fmfgraph*}(30,30)
                    \fmfleft{L1,L2}
                    \fmfright{R1,R2}
                    \fmf{dbl_wiggly}{L1,V}
                    \fmf{dbl_wiggly}{L2,V}
                    \fmf{dbl_wiggly}{R1,V}
                    \fmf{dbl_wiggly}{R2,V}
                    \fmffreeze
                    \fmfdot{V}
                    \fmf{dots}{L1,L2}
                    \fmflabel{$\mu_1\nu_1,p_1$}{R1}
                    \fmflabel{$\mu_2\nu_2,p_2$}{R2}
                    \fmflabel{$\mu_3\nu_3,p_3$}{L1}
                    \fmflabel{$\mu_n\nu_n,p_n$}{L2}
                \end{fmfgraph*}
            \end{fmffile}
        \end{gathered}
        \hspace{35pt}= \,i\,2\,\kappa^{n-2}\,\left(\sqrt{-g}\,g^{\mu\nu}g^{\alpha\beta}g^{\rho\sigma}\right)^{\mu_3\nu_3\cdots\mu_n\nu_n} \, (p_1)_{\lambda_1} (p_2)_{\lambda_2} \\ \\
        & \times \Bigg[ \left(\Gamma_{\alpha\mu\rho}\right)^{\lambda_1\mu_1\nu_1} \left(\Gamma_{\sigma\nu\beta}\right)^{\lambda_2\mu_2\nu_2} - \left(\Gamma_{\alpha\mu\nu} \right)^{\lambda_1\mu_1\nu_1} \left( \Gamma_{\rho\beta\sigma}\right)^{\lambda_2\mu_2\nu_2} - \cfrac{\epsilon}{4} \left( \Gamma_{\mu\alpha\beta} \right)^{\lambda_1\mu_1\nu_1} \left(\Gamma_{\nu\rho\sigma} \right)^{\lambda_2\mu_2\nu_2}\Bigg] \\
        & + \text{permutations}.
    \end{split}
\end{align}
The sum goes over all possible permutations of graviton parameters $\{\mu_i\,\nu_i\,p_i\}$.

The ghost action is treated similarly.
\begin{align}
  \begin{split}
    \mathcal{A}_\text{ghost} =& \int d^4 x \sqrt{-g} \left[ -g^{\alpha\beta} g^{\mu\nu} \nabla_\alpha \overline{c}_\mu \nabla_\beta c_\nu - 2\,\Gamma^\mu_{\alpha\beta} \,\overline{c}_\mu \,\nabla^\alpha c^\beta + R_{\mu\nu} \, \overline{c}^\mu \,c^\nu \right] \\
    =& \int d^4 x \sqrt{-g} \left[ -g^{\mu\nu} g^{\alpha\beta} \, \pd_\alpha\overline{c}_\mu \, \pd_\beta c_\nu \right]\\
    & + \int d^4 x \sqrt{-g}\, g^{\mu\alpha}g^{\nu\beta}g^{\rho\sigma} \left[ \Gamma_{\beta\rho\alpha} \pd_\sigma \overline{c}_\mu c_\nu - \Gamma_{\alpha\rho\beta}\,\overline{c}_\mu \pd_\sigma c_\nu + \pd_\rho\Gamma_{\sigma\alpha\beta} \,\overline{c}_\mu \,c_\nu - \pd_\alpha \Gamma_{\rho\beta\sigma} \overline{c}_\mu c_\nu \right]\\
    & + \int d^4 x \sqrt{-g}\,g^{\mu\alpha} g^{\nu\beta} g^{\rho\sigma} g^{\lambda\tau} \left[ \Gamma_{\rho\alpha\lambda} \Gamma_{\sigma\beta\tau} - \Gamma_{\rho\alpha\beta} \Gamma_{\sigma\lambda\tau} + \Gamma_{\alpha\rho\lambda} \Gamma_{\beta\sigma\tau} \right] \overline{c}_\mu c_\nu .
  \end{split}
\end{align}
It has the following perturbative expansion:
\begin{theorem}
    \begin{align}
        \begin{split}
            \mathcal{A}_\text{ghost} =& \sum\limits_{n=0}^\infty\int \cfrac{d^4 p_1}{(2\pi)^4} \cfrac{d^4 p_2}{(2\pi)^4} \prod\limits_{i=1}^n \cfrac{d^4 k_i}{(2\pi)^4} \, (2\pi)^4 \delta\big(p_1+p_2 + \sum k_i\big) h_{\rho_1\sigma_1}(k_1)\cdots h_{\rho_n\sigma_n}(k_n) \, \overline{c}_\mu (p_1) c_\nu(p_2)\\
            &\times \kappa^n \,\left(\sqrt{-g}\,g^{\mu\nu} g^{\alpha\beta} \right)^{\rho_1\sigma_1\cdots\rho_n\sigma_n} (p_1)_\alpha (p_2)_\beta\\
            +& \sum\limits_{n=1}^\infty\int \cfrac{d^4 p_1}{(2\pi)^4} \cfrac{d^4 p_2}{(2\pi)^4} \prod\limits_{i=1}^n \cfrac{d^4 k_i}{(2\pi)^4} \, (2\pi)^4 \delta\big(p_1+p_2 + \sum k_i\big) h_{\rho_1\sigma_1}(k_1)\cdots         h_{\rho_n\sigma_n}(k_n) \, \overline{c}_\mu (p_1) c_\nu(p_2)\\
            &\times \kappa^n (-1)\left(\sqrt{-g}\,g^{\mu\alpha}g^{\nu\beta} g^{\rho\sigma}\right)^{\rho_2\sigma_2\cdots\rho_n\sigma_n} (k_1)_\lambda\left[ (p_1)_\sigma \left(\Gamma_{\beta\rho\alpha}\right)^{\lambda\rho_1\sigma_1}-(p_2)_\sigma \left(\Gamma_{\alpha\rho\beta}\right)^{\lambda\rho_1\sigma_1} \right.\\
            &\hspace{210pt}\left. + (k_1)_\rho \left(\Gamma_{\sigma\alpha\beta}\right)^{\lambda\rho_1\sigma_1} - (k_1)_\alpha \left(\Gamma_{\rho\beta\sigma}\right)^{\lambda\rho_1\sigma_1}\right]\\
            +& \sum\limits_{n=2}^\infty\int \cfrac{d^4 p_1}{(2\pi)^4} \cfrac{d^4 p_2}{(2\pi)^4} \prod\limits_{i=1}^n \cfrac{d^4 k_i}{(2\pi)^4} \, (2\pi)^4 \delta\big(p_1+p_2 + \sum k_i\big) h_{\rho_1\sigma_1}(k_1)\cdots h_{\rho_n\sigma_n}(k_n) \, \overline{c}_\mu (p_1) c_\nu(p_2)\\
            &\times \kappa^n (-1) \left(\sqrt{-g} \,g^{\mu\alpha}g^{\nu\beta} g^{\rho\sigma} g^{\lambda\tau}\right)^{\rho_3\sigma_3\cdots\rho_n\sigma_n}\,(k_1)_{\lambda_1} (k_2)_{\lambda_2} \left[ \left(\Gamma_{\rho\alpha\lambda}\right)^{\lambda_1\rho_1\sigma_1} \left(\Gamma_{\sigma\beta\tau}\right)^{\lambda_2\rho_2\sigma_2} \right.\\
            & \hspace{140pt}\left. - \left(\Gamma_{\rho\alpha\beta}\right)^{\lambda_1\rho_1\sigma_1} \left(\Gamma_{\sigma\lambda\tau}\right)^{\lambda_2\rho_2\sigma_2} + \left(\Gamma_{\alpha\rho\lambda}\right)^{\lambda_1\rho_1\sigma_1} \left(\Gamma_{\beta\sigma\tau}\right)^{\lambda_2\rho_2\sigma_2} \right].
        \end{split}
    \end{align}
\end{theorem}
The complete expression describing graviton-ghost vertices reads:
\begin{align}\label{Rules_Graviton-Ghosts}
    \nonumber \\
    \begin{split}
        & \hspace{30pt}
        \begin{gathered}
            \begin{fmffile}{FR_Gh_2}
                \begin{fmfgraph*}(30,30)
                    \fmfleft{L1,L2}
                    \fmfright{R1,R2}
                    \fmf{dbl_wiggly}{L1,V}
                    \fmf{dbl_wiggly}{L2,V}
                    \fmf{dots_arrow}{R1,V}
                    \fmf{dots_arrow}{V,R2}
                    \fmfdot{V}
                    \fmffreeze
                    \fmf{dots}{L1,L2}
                    \fmflabel{$\rho_1\sigma_1,k_1$}{L1}
                    \fmflabel{$\rho_n\sigma_n,k_n$}{L2}
                    \fmflabel{$\nu,p_2$}{R1}
                    \fmflabel{$\mu,p_1$}{R2}
                \end{fmfgraph*}
            \end{fmffile}
        \end{gathered}
        \hspace{25pt}= i\,\kappa^n \Bigg[ \left(\sqrt{-g} \,g^{\mu\nu}g^{\alpha\beta}\right)^{\rho_1\sigma_1\cdots\rho_n\sigma_n} (p_1)_\alpha (p_2)_\beta \\ \\
        &+\Bigg\{ - \left(\sqrt{-g}\,g^{\mu\alpha} g^{\nu\beta}g^{\rho\sigma}\right)^{\rho_2\sigma_2\cdots\rho_n\sigma_n}(k_1)_\lambda \Bigg[ (p_1)_\sigma (\Gamma_{\beta\rho\alpha})^{\lambda\rho_1\sigma_1} - (p_2)_\sigma (\Gamma_{\alpha\rho\beta})^{\lambda\rho_1\sigma_1} \\
        & \hspace{150pt}+ (k_1)_\rho (\Gamma_{\sigma\alpha\beta})^{\lambda\rho_1\sigma_1}-(k_1)_\alpha (\Gamma_{\rho\beta\sigma})^{\lambda\rho_1\sigma_1} \Bigg] + \text{permutations} \Bigg\}\\
        & + \Bigg\{ -\left(\sqrt{-g}\,g^{\mu\alpha} g^{\nu\beta}g^{\rho\sigma} g^{\lambda\tau} \right)^{\rho_3\sigma_3\cdots\rho_n\sigma_n} \,(k_1)_{\lambda_1} (k_2)_{\lambda_2} \left[ \left(\Gamma_{\rho\alpha\lambda}\right)^{\lambda_1\rho_1\sigma_1} \left(\Gamma_{\sigma\beta\tau}\right)^{\lambda_2\rho_2\sigma_2} \right.\\
        & \hspace{70pt}\left. - \left(\Gamma_{\rho\alpha\beta}\right)^{\lambda_1\rho_1\sigma_1} \left(\Gamma_{\sigma\lambda\tau}\right)^{\lambda_2\rho_2\sigma_2} + \left(\Gamma_{\alpha\rho\lambda}\right)^{\lambda_1\rho_1\sigma_1} \left(\Gamma_{\beta\sigma\tau}\right)^{\lambda_2\rho_2\sigma_2} \right] + \text{permutations} \Bigg\} \Bigg].
    \end{split}
\end{align}

The standard procedure derives propagators for ghosts and gravitons. The following expression gives the ghost propagator.
\begin{align}
  \begin{gathered}
    \begin{fmffile}{FR_Ghost_Propagator}
      \begin{fmfgraph*}(30,30)
        \fmfleft{L}
        \fmfright{R}
        \fmf{dots}{L,R}
        \fmflabel{$\mu$}{L}
        \fmflabel{$\nu$}{R}
      \end{fmfgraph*}
    \end{fmffile}
  \end{gathered}
  \hspace{15pt} = i \, \cfrac{\eta_{\mu\nu}}{k^2}\,.
\end{align}
The graviton propagator contains the gauge fixing parameter $\epsilon$. The propagator corresponds to the part of the microscopic action quadratic in perturbations:
\begin{align}
  \int d^4 x \sqrt{-g} \left[ -\cfrac{2}{\kappa^2}\, R + \cfrac{\epsilon}{2\,\kappa^2} ~\mathcal{G}_\mu \mathcal{G}^\mu \right] = \int d^4 x \left[  -\cfrac12\,h^{\mu\nu} \mathcal{D}_{\mu\nu\alpha\beta}(\epsilon)\, \square h^{\alpha\beta}  \right] +\okappa{1}.
\end{align}
The Nieuwenhuizen operators providee a good mean to express the $\mathcal{D}$ operator in the momentum representation \cite{VanNieuwenhuizen:1981ae,Accioly:2000nm}
\begin{align}
  \mathcal{D}_{\mu\nu\alpha\beta} (\epsilon) = \cfrac{3 \epsilon -8}{4}\, P^0_{\mu\nu\alpha\beta} + \cfrac{\epsilon}{2}\,P^1_{\mu\nu\alpha\beta} + P^2_{\mu\nu\alpha\beta} - \cfrac{\epsilon}{4} ~ \overline{P}^0_{\mu\nu\alpha\beta} - \cfrac{\epsilon}{4} ~ \overline{\overline{P}}^0_{\mu\nu\alpha\beta} \,.
\end{align}
Operators $P^0$ and $P^2$ are gauge invariant, making the operator non-invertible unless $\epsilon \not =0$. The inverse operator reads:
\begin{align}
  \mathcal{D}^{-1}_{\mu\nu\alpha\beta}(\epsilon) = -\cfrac12\,P^0_{\mu\nu\alpha\beta}+\cfrac{2}{\epsilon}\, P^1_{\mu\nu\alpha\beta} + P^2_{\mu\nu\alpha\beta} - \cfrac{3\,\epsilon -8}{2\,\epsilon}~\overline{P}^0_{\mu\nu\alpha\beta} - \cfrac{1}{2} ~\overline{\overline{P}}^0_{\mu\nu\alpha\beta} .
\end{align}
The given formula expresses the graviton propagator in any arbitrary gauge.
\begin{align}
  \begin{gathered}
    \begin{fmffile}{FR_Graviton_Propagator}
      \begin{fmfgraph*}(30,30)
        \fmfleft{L}
        \fmfright{R}
        \fmf{dbl_wiggly}{L,R}
        \fmflabel{$\mu\nu$}{L}
        \fmflabel{$\alpha\beta$}{R}
      \end{fmfgraph*}
    \end{fmffile}
  \end{gathered}
  \hspace{20pt} = i ~ \cfrac{ \mathcal{D}^{-1}_{\mu\nu\alpha\beta}(\epsilon) }{k^2} \, .
\end{align}
We can discuss the general case, but using $\epsilon = 2$ in practical applications is easier. With this value, the operator $\mathcal{D}^{-1}$ becomes much simpler in form:
\begin{align}
  \mathcal{D}^{-1}_{\mu\nu\alpha\beta}(2) = \cfrac12\left[ \eta_{\mu\alpha} \eta_{\nu\beta} + \eta_{\mu\beta} \eta_{\nu\alpha} - \eta_{\mu\nu} \eta_{\alpha\beta} \right].
\end{align}

\section{FeynGrav}\label{FeynGrav_Section}

FeynGrav is a package for Wolfram Mathematica extending FeynCalc functionality \cite{Shtabovenko:2016sxi,Shtabovenko:2020gxv}. FeynCalc provides tools to study both tree and loop-level amplitudes. At the same time, there are many packages further extending its functionality \cite{Patel:2015tea,Patel:2016fam,Shapiro:2016pfm}, which makes it a platform for different computational tools of high energy physics. Because of these reasons, FeynCalc was chosen to implement the perturbative quantum gravity framework described above.

This subsection is split into two parts to discuss two different aspects of FeynGrav. The first subsection discussed the general features of the package and its published versions. The second subsection discussed the commands implemented in the latest available version.

\subsection{Implementation}

We shall begin with a discussion of the existing versions of FeynGrav. There are a few published versions of the package, each extending its functionality and improving performance. The package is constantly developing, and the latest version is publically available at \cite{FeynGrav}. We shall briefly discuss the features of all these versions for completeness.

The original version 1.0 of FeynGrav was published in \cite{Latosh:2022ydd}. This version only considered matter with massless spin $0$, $1$, $2$, and $1/2$, and did not include the $SU(N)$ Yang-Mills model. The recursive relations discussed in this paper were yet to be discovered during the development of this version. Consequently, the package used a less effective algorithm to generate all perturbative expansions. The package's applicability was also limited since the gauge-fixing algorithm was not considered in detail.

The second version published in \cite{Latosh:2023zsi} presented a significant update. This version added interaction rules for massive matter with spins $0$, $1$, $1/2$, and implemented rules for the $SU(N)$ Yang-Mills. The gauge-fixing procedure was fully addressed, so the implemented interaction rules became applicable at any perturbation theory level. However, the recursive algorithm still needed to be discovered, so the package used the same algorithm for perturbative expansions.

The recursive relations were discovered after the publication \cite{Latosh:2023zsi} and were implemented in the latest FeynGrav version 2.1. The discovery of the recursive algorithm allowed expressions for the interaction rules to be generated more efficiently, so the corresponding libraries were updated. This version also implemented the polarisation tensors for gravity, discussed below. Lastly, minor misprints were corrected, and the sample file was significantly updated and improved.

The interaction rules for the Horndeski theory still need to be implemented in any existing version of FeynGrav. Because of their shared length and complexity, their implementation takes time and is expected to be published soon.

Let us discuss the package structure that remains similar in all the versions. The package addressed a few different challenges, consisting of a few semi-independent modules. The main computational challenge is the generation of the perturbative expansion terms. The package's core addresses this problem but operates separately from the main file. In turn, the main file of the package only imports the interaction rules in the FeynCalc environment, allowing a user to operate with them.

The above recursive relations are implemented in a series of sub-packages in a separate folder. Each package provides a tool to calculate a separate family of tensors and can be used independently within FeynCalc. As discussed above, it is essential to introduce tensors with particular symmetry for the Lorentz indices. With this symmetry imposed, a typical tensor with $2n$ pair of Lorentz indices will have approximately $2^n \, n!$ terms. Because of this large number of terms, a computation of an interaction rule involving many particles can take significant time. To soften this issue, the interaction rules are calculated separately.

The interaction rules discussed above are implemented in a single separate sub-package. Because such rules only require information about perturbative expansions, this sub-package depends on sub-packages describing perturbative expansions. It shall be run independently from FeynGrav to generate libraries containing the final expressions for the interaction rules. 

The main package file is independent of these sub-packages directly since it only imports the existing libraries and places them in the FeynCalc environment. Because of this package structure, the final user does not have to perform computationally heavy calculations of the interaction rules, significantly improving the package performance. When FeynGrav imports the rules, FeynCalc performs index contractions and other operations that constrain FeynGrav's performance.

\subsection{Interaction rules}

The scalar field kinetic and potential energy interaction rules are implemented with commands ``{\bf GravitonScalarVertex}'' and ``{\bf GravitonScalarPotentialVertex}''. The first command takes four arguments: the array of graviton indices, two momenta of scalar fields, and the scalar field mass. The second command takes two arguments: the array of graviton indices and the scalar field coupling coupling. FeynGrav also contains a command realising the scalar field propagator. The command is ``{\bf ScalarPropagator}'' and takes two arguments: the scalar field momentum and mass. Table \ref{Table_Scalars_Proca_Vectors} presents examples of these commands' usage.

Interaction rules for the Proca field are implemented with a single command ``{\bf GravitonMassiveVectorVertex}''. The name is chosen for the sake of naming consistency. The command takes six arguments: the array of graviton indices, the Lorentz index and momentum of the first vector particle, the Lorentz index and momentum of the second vector particle, and the Proca field mass. The package also has the Proca propagator implementation. The command ``{\bf ProcaPropagator}'' implements the propagator and takes four arguments: two Lorentz indices, the momentum and mass of the Proca field. Table \ref{Table_Scalars_Proca_Vectors} presents examples of these commands' usage.

The interaction rules for a single massless vector field are implemented with two commands, ``{\bf GravitonVectorVertex}'' and ``{\bf GravitonVectorGhostVertex}''. The command ``{\bf GravitonVectorVertex}'' takes five arguments: the array of Lorentz indices and momenta of gravitons, the Lorentz index and momentum of the first vector particle, the Lorentz index and momentum of the second particle. The command ``{\bf GravitonVectorGhostVertex}'' takes three arguments: the array of Lorentz indices and momenta of gravitons, and the momenta of the Faddeev-Popov ghost. Table \ref{Table_Scalars_Proca_Vectors} presents examples of these commands. Let us also note that FeynCalc already provides the propagators required to treat a massless scalar field.

\begin{table}[ht]
    \centering
    \begin{tabular}{c|l}
    \hspace{50pt} Diagram \hspace{50pt} & Command \\ \hline 
    $
    \begin{gathered}
        \begin{fmffile}{Table_Illustration_Scalar_0}
            \begin{fmfgraph*}(30,30)
                \fmfleft{L}
                \fmfright{R}
                \fmf{dashes}{L,R}
                \fmflabel{~$p,m$}{R}
            \end{fmfgraph*}
        \end{fmffile}
    \end{gathered}
    $
    & ScalarPropagator$[p,m]$ \\ \hline \\ \\
    $
    \begin{gathered}
        \begin{fmffile}{Table_Illustration_Scalar_1}
            \begin{fmfgraph*}(30,30)
                \fmfleft{L}
                \fmfright{R1,R2}
                \fmf{dbl_wiggly}{L,V}
                \fmf{dashes}{R1,V}
                \fmf{dashes}{R2,V}
                \fmfdot{V}
                \fmflabel{$\rho_1\sigma_1$}{L}
                \fmflabel{$p_1,m$}{R1}
                \fmflabel{$p_2,m$}{R2}
            \end{fmfgraph*}
        \end{fmffile}
    \end{gathered}
    $
    & GravitonScalarVertex$[\{\rho_1 , \sigma_1\},p_1,p_2,m]$ \\ \\ \hline \\ \\
    $
    \begin{gathered}
        \begin{fmffile}{Table_Illustration_Scalar_2}
            \begin{fmfgraph*}(30,30)
                \fmfleft{L1,L2}
                \fmfright{R1,R2}
                \fmf{dbl_wiggly}{L1,V}
                \fmf{dbl_wiggly}{L2,V}
                \fmf{dashes}{R1,V}
                \fmf{dashes}{R2,V}
                \fmfdot{V}
                \fmflabel{$\rho_1\sigma_1$}{L1}
                \fmflabel{$\rho_2\sigma_2$}{L2}
                \fmflabel{$p_1,m$}{R1}
                \fmflabel{$p_2$}{R2}
            \end{fmfgraph*}
        \end{fmffile}
    \end{gathered}
    $
    & GravitonScalarVertex$[\{\rho_1 , \sigma_1,\rho_2,\sigma_2\},p_1,p_2,m]$ \\ \\ \hline \\ 
    $
    \begin{gathered}
        \begin{fmffile}{Table_Illustration_Scalar_3}
            \begin{fmfgraph*}(30,30)
                \fmfleft{L1}
                \fmfright{R1,R2,R3}
                \fmf{dbl_wiggly,tension=3}{L1,V}
                \fmf{dashes}{R1,V}
                \fmf{dashes}{R2,V}
                \fmf{dashes}{R3,V}
                \fmfdot{V}
                \fmflabel{$\rho_1\sigma_1$}{L1}
                \fmflabel{$\lambda_3$}{R2}
            \end{fmfgraph*}
        \end{fmffile}
    \end{gathered}
    $
    & GravitonScalarPotentialVertex$[\{\rho_1 , \sigma_1\},\lambda_3]$ \\  \hline \\ \\
    $
    \begin{gathered}
        \begin{fmffile}{Table_Illustration_Scalar_4}
            \begin{fmfgraph*}(30,30)
                \fmfleft{L1,L2}
                \fmfright{R1,R2,R3,R4}
                \fmf{dbl_wiggly,tension=2}{L1,V}
                \fmf{dbl_wiggly,tension=2}{L2,V}
                \fmf{dashes}{R1,V}
                \fmf{dashes}{R2,V}
                \fmf{dashes}{R3,V}
                \fmf{dashes}{R4,V}
                \fmfdot{V}
                \fmflabel{$\rho_1\sigma_1$}{L1}
                \fmflabel{$\rho_2\sigma_2$}{L2}
                \fmflabel{$\lambda_4$}{R2}
            \end{fmfgraph*}
        \end{fmffile}
    \end{gathered}
    $
    & GravitonScalarPotentialVertex$[\{\rho_1 , \sigma_1, \rho_2, \sigma_2\},\lambda_4]$ \\ \\ \hline
    $
    \begin{gathered}
        \begin{fmffile}{Table_Illustration_Proca_0}
            \begin{fmfgraph*}(30,30)
                \fmfleft{L}
                \fmfright{R}
                \fmf{photon}{L,R}
                \fmflabel{$\mu$}{L}
                \fmflabel{$\nu ~~ p,m$}{R}
            \end{fmfgraph*}
        \end{fmffile}
    \end{gathered}
    $
    & ProcaPropagator$[p,m]$ \\ \hline \\ \\
    $
    \begin{gathered}
        \begin{fmffile}{Table_Illustration_Proca_1}
            \begin{fmfgraph*}(30,30)
                \fmfleft{L}
                \fmfright{R1,R2}
                \fmf{dbl_wiggly}{L,V}
                \fmf{photon}{R1,V}
                \fmf{photon}{R2,V}
                \fmfdot{V}
                \fmflabel{$\rho_1\sigma_1$}{L}
                \fmflabel{$p_1,\lambda_1,m$}{R1}
                \fmflabel{$p_2,\lambda_2,m$}{R2}
            \end{fmfgraph*}
        \end{fmffile}
    \end{gathered}
    $
    & GravitonMassiveVectorVertex$[\{\rho_1 , \sigma_1\},\lambda_1,p_1,\lambda_2,p_2,m]$ \\ \\ \hline \\ \\
    $
    \begin{gathered}
        \begin{fmffile}{Table_Illustration_Vector_1}
            \begin{fmfgraph*}(30,30)
                \fmfleft{L}
                \fmfright{R1,R2}
                \fmf{dbl_wiggly}{L,V}
                \fmf{photon}{R1,V}
                \fmf{photon}{R2,V}
                \fmfdot{V}
                \fmflabel{$\rho_1\sigma_1,k_1$}{L}
                \fmflabel{$\lambda_1, p_1$}{R1}
                \fmflabel{$\lambda_2, p_2$}{R2}
            \end{fmfgraph*}
        \end{fmffile}
    \end{gathered}
    $
    & GravitonVectorVertex$[\{\rho_1 , \sigma_1, k_1\},\lambda_1,p_1,\lambda_2,p_2]$ \\ \\ \hline \\ \\
    $
        \begin{gathered}
        \begin{fmffile}{Table_Illustration_Vector_2}
            \begin{fmfgraph*}(30,30)
                \fmfleft{L}
                \fmfright{R1,R2}
                \fmf{dbl_wiggly}{L,V}
                \fmf{dots}{R1,V}
                \fmf{dots}{R2,V}
                \fmfdot{V}
                \fmflabel{$\rho_1\sigma_1,k_1$}{L}
                \fmflabel{$p_1$}{R1}
                \fmflabel{$p_2$}{R2}
            \end{fmfgraph*}
        \end{fmffile}
    \end{gathered}
    $
    & GravitonVectorGhostVertex$[\{\rho_1 , \sigma_1, k_1\},p_1,p_2]$ \\ \\ \hline 
    \end{tabular}
    \caption{Examples of interaction rules for the scalar, Proca, and massless vector fields.}
    \label{Table_Scalars_Proca_Vectors}
\end{table}

The interaction rules for Dirac fermions are implemented with a single command ``{\bf GravitonFermionVertex}''. The command takes four arguments: the array of graviton Lorentz indices, the momentum of the in-going fermion line, the momentum of the out-going fermion line, and the fermion mass. The fermion propagator has already been implemented in FeynCalc. Table \ref{Table_Fermions} provides an example of this interaction rule.

The $SU(N)$ Yang-Mills model implementation is sophisticated and involves several commands. The command ``{\bf GravitonGluonVertex} implements a coupling of a few gluons to gravity. The command's first argument is an array of Lorentz indices and momenta. The other arguments describe gluons Lorentz indices, momenta, and $SU(N)$ indices. The command can describe a coupling of two, three, and four gluons to gravity. Gravitational coupling of quarks kinetic energy matches the expression for the coupling of a Dirac fermion. The coupling of the quark-gluon interaction energy is described by the ``{\bf GravitonQuarkGluonVertex}'' command. It takes only three arguments: the array of graviton Lorentz indices, the quark-gluon vertex Lorentz index, and the colour index. Lastly, two commands are responsible for gravitational coupling to the Faddeev-Popov ghosts of the Yang-Mills theory. The command ``{\bf GravitonYMGhostVertex}'' describes the coupling to the ghost itself. The command takes five arguments: the array of graviton Lorentz indices, the momentum and colour index of the first ghost, and the momentum and colour index of the second ghost. The command ``{\bf GravitonGluonGhostVertex}'' corresponds to the gravitational coupling of a vertex describing the interaction between two ghosts and one gluon. The command's arguments are the array of graviton Lorentz indices, the Lorentz indices, momenta, and colour indices of other particles. Table \ref{Table_Fermions} lists examples of these commands.

\begin{table}[ht]
    \centering
    \begin{tabular}{c|l}
    \hspace{40pt} Diagram \hspace{40pt} & Command \\ \hline \\ \\
    $
    \begin{gathered}
        \begin{fmffile}{Table_Illustration_Fermion_1}
            \begin{fmfgraph*}(30,30)
                \fmfleft{L}
                \fmfright{R1,R2}
                \fmf{dbl_wiggly}{L,V}
                \fmf{fermion}{R1,V}
                \fmf{fermion}{V,R2}
                \fmfdot{V}
                \fmflabel{$\rho_1\sigma_1$}{L}
                \fmflabel{$p_1,m$}{R1}
                \fmflabel{$p_2$,m,m}{R2}
            \end{fmfgraph*}
        \end{fmffile}
    \end{gathered}
    $
    & GravitonFermionVertex$[\{\rho_1 , \sigma_1\},p_1,p_2,m]$ \\ \\ \hline \\ \\
    $
    \begin{gathered}
        \begin{fmffile}{Table_Illustration_Gluon_1}
            \begin{fmfgraph*}(30,30)
                \fmfleft{L}
                \fmfright{R1,R2}
                \fmf{dbl_wiggly,tension=2}{L,V}
                \fmf{gluon}{R1,V}
                \fmf{gluon}{R2,V}
                \fmfdot{V}
                \fmflabel{$\rho_1\sigma_1,k_1$}{L}
                \fmflabel{$\lambda_1,p_1,a_1$}{R1}
                \fmflabel{$\lambda_2,p_2,a_2$}{R2}
            \end{fmfgraph*}
        \end{fmffile}
    \end{gathered}
    $
    & GravitonGluonVertex$[\{\rho_1 , \sigma_1, k_1\},p_1,\lambda_1,a_1,p_2,\lambda_2,a_2]$ \\ \\ \hline \\ \\
    $
    \begin{gathered}
        \begin{fmffile}{Table_Illustration_Gluon_2}
            \begin{fmfgraph*}(30,30)
                \fmfleft{L}
                \fmfright{R1,R2,R3}
                \fmf{dbl_wiggly,tension=3}{L,V}
                \fmf{gluon}{R1,V}
                \fmf{gluon}{R2,V}
                \fmf{gluon}{R3,V}
                \fmfdot{V}
                \fmflabel{$\rho_1\sigma_1,k_1$}{L}
                \fmflabel{$\lambda_1,p_1,a_1$}{R1}
                \fmflabel{$\lambda_2,p_2,a_2$}{R2}
                \fmflabel{$\lambda_3,p_3,a_3$}{R3}
            \end{fmfgraph*}
        \end{fmffile}
    \end{gathered}
    $
    & GravitonGluonVertex$[\{\rho_1 , \sigma_1, k_1\},p_1,\lambda_1,a_1,p_2,\lambda_2,a_2,\lambda_3,p_3,a_3]$ \\ \\ \hline \\ \\
    $
    \begin{gathered}
        \begin{fmffile}{Table_Illustration_Gluon_3}
            \begin{fmfgraph*}(30,30)
                \fmfleft{L}
                \fmfright{R1,R2,R3,R4}
                \fmf{dbl_wiggly,tension=4}{L,V}
                \fmf{gluon}{R1,V}
                \fmf{gluon}{R2,V}
                \fmf{gluon}{R3,V}
                \fmf{gluon}{R4,V}
                \fmfdot{V}
                \fmflabel{$\rho_1\sigma_1,k_1$}{L}
                \fmflabel{$\lambda_1,p_1,a_1$}{R1}
                \fmflabel{$\lambda_2,p_2,a_2$}{R2}
                \fmflabel{$\lambda_3,p_3,a_3$}{R3}
                \fmflabel{$\lambda_4,p_4,a_4$}{R4}
            \end{fmfgraph*}
        \end{fmffile}
    \end{gathered}
    $
    & GravitonGluonVertex$[\{\rho_1 , \sigma_1, k_1\},p_1,\lambda_1,a_1,p_2,\lambda_2,a_2,\lambda_3,p_3,a_3,\lambda_4,p_4,a_4]$ \\ \\ \hline \\ 
    $
        \begin{gathered}
        \begin{fmffile}{Table_Illustration_Quark-Gluon_1}
            \begin{fmfgraph*}(30,30)
                \fmfleft{L}
                \fmfright{R1,R2,R3}
                \fmf{dbl_wiggly,tension=3}{L,V}
                \fmf{fermion}{R1,V}
                \fmf{gluon}{R2,V}
                \fmf{fermion}{V,R3}
                \fmfdot{V}
                \fmflabel{$\rho_1\sigma_1,k_1$}{L}
                \fmflabel{$\lambda, a$}{R2}
            \end{fmfgraph*}
        \end{fmffile}
    \end{gathered}
    $
    & GravitonQuarkGluonVertex$[\{\rho_1, \sigma_1\}, \lambda, a]$ \\  \hline \\ \\
    $
        \begin{gathered}
        \begin{fmffile}{Table_Illustration_YM_Ghost_1}
            \begin{fmfgraph*}(30,30)
                \fmfleft{L}
                \fmfright{R1,R2}
                \fmf{dbl_wiggly,tension=2}{L,V}
                \fmf{dots}{R1,V}
                \fmf{dots}{R2,V}
                \fmfdot{V}
                \fmflabel{$\rho_1\sigma_1,k_1$}{L}
                \fmflabel{$p_1,a_1$}{R1}
                \fmflabel{$p_2,a_2$}{R2}
            \end{fmfgraph*}
        \end{fmffile}
    \end{gathered}
    $
    & GravitonYMGhostVertex$[\{\rho_1 , \sigma_1\}, p_1, a_1, p_2, a_2]$ \\ \\ \hline \\ \\
    $
    \begin{gathered}
        \begin{fmffile}{Table_Illustration_YM_Ghost_2}
            \begin{fmfgraph*}(30,30)
                \fmfleft{L}
                \fmfright{R1,R2,R3}
                \fmf{dbl_wiggly,tension=3}{L,V}
                \fmf{dots}{R1,V}
                \fmf{gluon}{R2,V}
                \fmf{dots}{R3,V}
                \fmfdot{V}
                \fmflabel{$\rho_1\sigma_1,k_1$}{L}
                \fmflabel{$p_2,a_2$}{R1}
                \fmflabel{$\lambda_1,p_1,a_1$}{R2}
                \fmflabel{$p_3,a_3$}{R3}
            \end{fmfgraph*}
        \end{fmffile}
    \end{gathered}
    $
    & GravitonGluonGhostVertex$[\{\rho_1 , \sigma_1\},p_1, \lambda_1, a_1, p_2, \lambda_2, a_2, p_3, \lambda_3, a_3]$ \\ \\ \hline 
    \end{tabular}
    \caption{Examples of interaction rules for the scalar, Proca, and massless vector fields.}
    \label{Table_Fermions}
\end{table}

Lastly, the following commands describe the gravitational sector. The command ``{\bf GravitonPropagator}'' implements the graviton propagator with ``FeynAmpDenominator'' functions from FeynCalc. This command shall be used for loop calculations. The command ``GravitonPropagatorAlternative'' implements the graviton propagator with the simple scalar products in the denominators. This command shall only be used in three-level calculations. Two more commands, ``{\bf GravitonPropagatorTop}'' and ``{\bf GravitonPropagatorTopFAD}'', generate the numerator of the graviton propagator alone. The first one uses simple scalar products, while the second one uses ``FeynAmpDenominator''. All these commands take five arguments: four Lorentz indices and the momenta of a graviton. The command ``{\bf GravitonVertex}'' corresponds to the $n$-graviton vertex and takes $3 n$ arguments: two Lorentz indices and the momentum of each graviton. The Faddeev-Popov ghosts are vectors, so their propagator receives an additional Minkowski metric in the numerator. One can use the command for the Yang-Mills ghost propagator and manually multiply it on the flat metric. Because of this, no new command corresponding to this propagator is present in FeynGrav. Lastly, the rules for graviton coupling to the corresponding Faddeev-Popov ghosts are implemented with ``{\bf GravitonGhostVertex}''. It takes five arguments: the array of Lorentz indices and momenta of the gravitons, and Lorenz indices and momenta of ghosts. Table \ref{Table_gravitons} presents these rules.

\begin{table}[ht]
    \centering
    \begin{tabular}{c|l}
    \hspace{40pt} Diagram \hspace{40pt} & Command \\ \hline
    $
    \begin{gathered}
        \begin{fmffile}{Table_Illustration_Graviton_0}
            \begin{fmfgraph*}(30,30)
                \fmfleft{L}
                \fmfright{R}
                \fmf{dbl_wiggly}{L,R}
                \fmflabel{$\mu\nu$}{L}
                \fmflabel{$\alpha\beta ~~ p$}{R}
            \end{fmfgraph*}
        \end{fmffile}
    \end{gathered}
    $
    & GravitonPropagator$[\mu, \nu, \alpha, \beta, p]$ \\ \hline \\ \\
    $
    \begin{gathered}
        \begin{fmffile}{Table_Illustration_Graviton_1}
            \begin{fmfgraph*}(30,30)
                \fmfleft{L}
                \fmfright{R1,R2}
                \fmf{dbl_wiggly}{L,V}
                \fmf{dbl_wiggly}{V,R1}
                \fmf{dbl_wiggly}{V,R2}
                \fmflabel{$\rho_1,\sigma_1,p_1$}{L}
                \fmflabel{$\rho_2,\sigma_2,p_2$}{R1}
                \fmflabel{$\rho_3,\sigma_3,p_3$}{R2}
                \fmfdot{V}
            \end{fmfgraph*}
        \end{fmffile}
    \end{gathered}
    $
    & GravitonVertex$[\rho_1, \sigma_1, p_1, \rho_2, \sigma_2, p_2, \rho_3, \sigma_3, p_3]$ \\ \\ \hline \\ \\
    $
    \begin{gathered}
        \begin{fmffile}{Table_Illustration_Graviton_2}
            \begin{fmfgraph*}(30,30)
                \fmfleft{L1,L2}
                \fmfright{R1,R2}
                \fmf{dbl_wiggly}{L1,V}
                \fmf{dbl_wiggly}{L2,V}
                \fmf{dbl_wiggly}{V,R1}
                \fmf{dbl_wiggly}{V,R2}
                \fmflabel{$\rho_1,\sigma_1,p_1$}{L1}
                \fmflabel{$\rho_2,\sigma_2,p_2$}{L2}
                \fmflabel{$\rho_3,\sigma_3,p_3$}{R1}
                \fmflabel{$\rho_4,\sigma_4,p_4$}{R2}
                \fmfdot{V}
            \end{fmfgraph*}
        \end{fmffile}
    \end{gathered}
    $
    & GravitonVertex$[\rho_1, \sigma_1, p_1, \rho_2, \sigma_2, p_2, \rho_3, \sigma_3, p_3, \rho_4, \sigma_4, p_4]$ \\ \\ \hline \\ \\
    $
    \begin{gathered}
        \begin{fmffile}{Table_Illustration_Graviton_Ghost}
            \begin{fmfgraph*}(30,30)
                \fmfleft{L}
                \fmfright{R1,R2}
                \fmf{dbl_wiggly}{L,V}
                \fmf{dots}{V,R1}
                \fmf{dots}{V,R2}
                \fmflabel{$\rho_1,\sigma_1,k_1$}{L}
                \fmflabel{$\lambda_1, p_1$}{R1}
                \fmflabel{$\lambda_2, p_2$}{R2}
                \fmfdot{V}
            \end{fmfgraph*}
        \end{fmffile}
    \end{gathered}
    $
    & GravitonGhostVertex$[\{\rho_1 , \sigma_1, k_1\},\lambda_1, p_1, \lambda_2, p_2]$ \\ \\ \hline
    \end{tabular}
    \caption{Examples of interaction rules for gravitons.}
    \label{Table_gravitons}
\end{table}

FeynGrav also implements a few additional tools. Firstly, the package provides a realisation of the standard gauge projectors:
\begin{align}
    \theta_{\mu\nu}(p) & = \eta_{\mu\nu} - \cfrac{ p_\mu p_\nu}{p^2} \, , & \overline{\theta}_{\mu\nu} & = \cfrac{p_\mu p_\nu}{p^2} \, .
\end{align}
They are realised with commands ``{\bf GaugeProjector}'', ``{\bf GaugeProjectorBar}'', ``{\bf GaugeProjectorFAD}'', and ``{\bf GaugeProjectorBarFAD}''. The first two commands realise these projectors with the simple scalar products, while the former two commands use ``FeynAmpDenominator''.

Secondly, the package provides a realisation of the Nieuwenhuizen operators \cite{VanNieuwenhuizen:1973fi,VanNieuwenhuizen:1981ae,Accioly:2000nm}. These operators $P^0_{\mu\nu\alpha\beta}(p)$, $P^1_{\mu\nu\alpha\beta}(p)$, $P^2_{\mu\nu\alpha\beta}(p)$, $\overline{P^0}_{\mu\nu\alpha\beta}(p)$, $\overline{\overline{P^0}}_{\mu\nu\alpha\beta}(p)$ generalise the standard projectors for the spin-$2$ systems. Their features were discussed in detail in a previous publication \cite{Latosh:2022ydd}. Commands ``{\bf NieuwenhuizenOperator0}'', ``{\bf NieuwenhuizenOperator1}'', ``{\bf NieuwenhuizenOperator2}'', ``{\bf NieuwenhuizenOperator0Bar}'', and ``{\bf NieuwenhuizenOperator0BarBar}'' implement these operators using only the simple scalar product. The set of command ``{\bf NieuwenhuizenOperator0FAD}'', $\cdots$, ``{\bf NieuwenhuizenOperator0BarBarFAD}'' implements them with ``{\bf FeynAmpDenominator}''.

Lastly, the package implements the polarisation tensors for gravitons. Following the standard procedure, the polarisation tensor $\epsilon_{\mu\nu}^{\pm}$ for gravitons is defined through the polarisation vector $\epsilon_\mu^\pm$ for the standard vector field:
\begin{align}
    \epsilon_{\mu\nu}^{\pm} = \epsilon^\pm_\mu \epsilon^\pm_\nu .
\end{align}
The command ``{\bf PolarizationTensor}'' implements this definition. The command takes three arguments: two Lorentz indices and one four-momentum. The command directly multiplies polarisation vectors implemented in FeynCalc. Due to the definitions used in FeynCalc, the implemented polarisation tensor is neither traceless nor transverse -- the command ``{\bf SetPolarizationTensor}'' the tensors traceless and transverse.

\subsection{Examples}

Several publications have used FeynGrav \cite{Latosh:2022ydd,Latosh:2023zsi,Latosh:2022hrf,Latosh:2023ueg,Latosh:2023xej}. These publications discuss the problems in detail, so we will only provide a brief discussion. We also focus on two vivid examples.

The first illustrative example of FeynGrav usage is the calculation of $2\to 2$ graviton scattering cross sections at the leading order of perturbation theory. The on-shell cross section was calculated long ago without using any contemporary computational tools \cite{Sannan:1986tz}, which makes such calculations notoriously complicated. To be exact, in the original article \cite{Sannan:1986tz}, the author operated with expressions symmetries in a certain way to make the calculation more compact. Nonetheless, even the basic description of the calculations takes a few pages in the original publication.

FeynGrav extremely simplifies such calculations. Firstly, one shall fix all the momenta on the mass shell and define all relations for scalar products of polarisation tensors and momenta since FeynCalc does not do this automatically. Secondly, one calculates the matrix element for the $s$-channel amplitude:
\begin{align}
    \nonumber \\
    \begin{split}
        & \hspace{40pt}
        \begin{gathered}
            \begin{fmffile}{Examples_Pic_1}
                \begin{fmfgraph*}(30,30)
                    \fmfleft{L1,L2}
                    \fmfright{R1,R2}
                    \fmf{dbl_wiggly}{L1,V1,R1}
                    \fmf{dbl_wiggly}{L2,V2,R2}
                    \fmf{dbl_wiggly}{V1,V2}
                    \fmfdot{V1,V2}
                    \fmflabel{$\mu_1\nu_1,p_1$}{L1}
                    \fmflabel{$\mu_2\nu_2,p_2$}{R1}
                    \fmflabel{$\mu_3\nu_3,p_3$}{L2}
                    \fmflabel{$\mu_4\nu_4,p_4$}{R2}
                \end{fmfgraph*}
            \end{fmffile}
        \end{gathered} \hspace{35pt} = \text{PolarizationTensor}[\mu_1,\nu_1,p_1] \text{PolarizationTensor}[\mu_2,\nu_2,p_2] \\ \\
        & \times \text{ComplexConjugate}[\text{PolarizationTensor}[\mu_3,\nu_3,p_3]] \text{ComplexConjugate}[\text{PolarizationTensor}[\mu_4,\nu_4,p_4]]\\
        & \times \text{GravitonVertex}[\mu_1,\nu_1,p_1,\mu_2,\nu_2,p_2,\alpha_1,\beta_1,p_1 \!+\! p_2] \text{GravitonVertex}[\alpha_2,\beta_2,p_1 \!+\! p_2,\mu_3,\nu_3,p_3,\mu_4,\nu_4,p_4] \\
        & \times \text{GravitonPropagatorAlternative}[\alpha_1,\beta_1,\alpha_2,\beta_2,p_1 \!+\! p_2].
    \end{split}
\end{align}
An average computer not designed for intensive computational tasks can complete this element in less than five minutes. One obtains the expressions for $t$ and $u$-channels similarly. The following expression gives the contact four-graviton interaction:
\begin{align}
    \nonumber \\
    \begin{split}
        & \hspace{40pt}
        \begin{gathered}
            \begin{fmffile}{Examples_Pic_2}
                \begin{fmfgraph*}(30,30)
                    \fmfleft{L1,L2}
                    \fmfright{R1,R2}
                    \fmf{dbl_wiggly}{L1,V}
                    \fmf{dbl_wiggly}{L2,V}
                    \fmf{dbl_wiggly}{R1,V}
                    \fmf{dbl_wiggly}{R2,V}
                    \fmflabel{$\mu_1\nu_1,p_1$}{L1}
                    \fmflabel{$\mu_2\nu_2,p_2$}{R1}
                    \fmflabel{$\mu_3\nu_3,p_3$}{L2}
                    \fmflabel{$\mu_4\nu_4,p_4$}{R2}
                    \fmfdot{V}
                \end{fmfgraph*}
            \end{fmffile}
        \end{gathered} \hspace{40pt} = \text{PolarizationTensor}[\mu_1,\nu_1,p_1] \text{PolarizationTensor}[\mu_2,\nu_2,p_2] \\ \\
        & \times \text{ComplexConjugate}[\text{PolarizationTensor}[\mu_3,\nu_3,p_3]] \text{ComplexConjugate}[\text{PolarizationTensor}[\mu_4,\nu_4,p_4]]\\
        & \times \text{GravitonVertex}[\mu_1,\nu_1,p_1,\mu_2,\nu_2,p_2,\mu_3,\nu_3,p_3,\mu_4,\nu_4,p_4] .
    \end{split}
\end{align}
This expression takes less than one minute to be computed.

Lastly, one constructs the complete scattering amplitude accounting for all contributions. It recovers the well-known expressions for chiral amplitudes in $d=4$:
\begin{align}
    \mathcal{M}(++++) &= i\, \cfrac{\kappa^2}{4}\,\cfrac{s^4}{s\,t\,u} \, , & \mathcal{M}(+-+-) & = i\, \cfrac{\kappa^2}{4}\,\cfrac{u^4}{s\,t\,u} \, , & \mathcal{M}(+--+) & = i\, \cfrac{\kappa^2}{4}\,\cfrac{t^4}{s\,t\,u} \, ,
\end{align}
\begin{align*}
    \mathcal{M}(++--) &= \mathcal{M} (+++-) = 0 .
\end{align*}

This example vividly shows that such complicated yet essential calculations are significant with FeynGrav. Moreover, the paper \cite{Latosh:2022ydd} presents an explicit expression for the scattering amplitude expressed in terms of graviton chiralities.

The second example is the calculation of the one-loop graviton-scalar vertex operator. Since a scalar field interacts with gravity, the interaction vertex receives correction at the loop level. The one-loop vertex operator describes these corrections:
\begin{align}
    \nonumber \\
    \begin{gathered}
        \begin{fmffile}{Example_Vertex_0}
            \begin{fmfgraph*}(30,30)
                \fmftop{T1,T2}
                \fmfbottom{B}
                \fmf{dashes}{T1,V,T2}
                \fmf{dbl_wiggly}{B,V}
                \fmfv{decor.shape=circle,decor.filled=shaded,decor.size=15}{V}
                \fmflabel{$\mu\nu,k$}{B}
                \fmflabel{$p_1$}{T1}
                \fmflabel{$p_2$}{T2}
            \end{fmfgraph*}
        \end{fmffile}
    \end{gathered} \hspace{20pt} = i\, \Gamma_{\mu\nu}(p_1,p_2,k,m)
    =
    \begin{gathered}
        \begin{fmffile}{Example_Vertex_1}
            \begin{fmfgraph}(30,30)
                \fmftop{T1,T2}
                \fmfbottom{B}
                \fmf{dbl_wiggly,tension=2}{B,VB}
                \fmf{dashes,tension=2}{T1,VT1}
                \fmf{dashes,tension=2}{T2,VT2}
                \fmf{phantom}{VB,VT2,VT1,VB}
                \fmfdot{VB,VT1,VT2}
                \fmffreeze
                \fmf{dashes}{VT1,VT2}
                \fmf{dbl_wiggly}{VT1,VB,VT2}
            \end{fmfgraph}
        \end{fmffile}
    \end{gathered}
    +
    \begin{gathered}
        \begin{fmffile}{Example_Vertex_2}
            \begin{fmfgraph}(30,30)
                \fmftop{T1,T2}
                \fmfbottom{B}
                \fmf{dbl_wiggly,tension=2}{B,VB}
                \fmf{dashes,tension=2}{T1,VT1}
                \fmf{dashes,tension=2}{T2,VT2}
                \fmf{phantom}{VB,VT2,VT1,VB}
                \fmfdot{VB,VT1,VT2}
                \fmffreeze
                \fmf{dbl_wiggly}{VT1,VT2}
                \fmf{dashes}{VT1,VB,VT2}
            \end{fmfgraph}
        \end{fmffile}
    \end{gathered}
    +
    \begin{gathered}
        \begin{fmffile}{Example_Vertex_3}
            \begin{fmfgraph}(30,30)
                \fmftop{T1,T2}
                \fmfbottom{B}
                \fmf{dashes,tension=2}{T1,VT,T2}
                \fmf{dbl_wiggly,left=1}{VT,VB,VT}
                \fmf{dbl_wiggly,tension=2}{B,VB}
                \fmfdot{VT,VB}
            \end{fmfgraph}
        \end{fmffile}
    \end{gathered}
    +
    \begin{gathered}
        \begin{fmffile}{Example_Vertex_4}
            \begin{fmfgraph}(30,30)
                \fmftop{T1,T2}
                \fmfbottom{B}
                \fmf{dashes}{T1,V,T2}
                \fmf{dbl_wiggly,tension=2}{V,B}
                \fmffreeze
                \fmf{phantom}{V,VV,T1}
                \fmffreeze
                \fmf{dbl_wiggly,left=1}{V,VV}
                \fmfdot{V,VV}
            \end{fmfgraph}
        \end{fmffile}
    \end{gathered}
    +
    \begin{gathered}
        \begin{fmffile}{Example_Vertex_5}
            \begin{fmfgraph}(30,30)
                \fmftop{T1,T2}
                \fmfbottom{B}
                \fmf{dashes}{T1,V,T2}
                \fmf{dbl_wiggly,tension=2}{V,B}
                \fmffreeze
                \fmf{phantom}{V,VV,T2}
                \fmffreeze
                \fmf{dbl_wiggly,right=1}{V,VV}
                \fmfdot{V,VV}
            \end{fmfgraph}
        \end{fmffile}
    \end{gathered}
    +
    \begin{gathered}
        \begin{fmffile}{Example_Vertex_6}
            \begin{fmfgraph}(30,30)
                \fmftop{T1,T,T2}
                \fmfbottom{B}
                \fmf{dashes,right=.3}{T1,V,T2}
                \fmf{dbl_wiggly,tension=2}{V,B}
                \fmffreeze
                \fmf{phantom,tension=2}{T,VT}
                \fmf{phantom}{V,VT}
                \fmffreeze
                \fmf{dbl_wiggly,left=.7}{V,VT,V}
                \fmfdot{V}
            \end{fmfgraph}
        \end{fmffile}
    \end{gathered} . \\ \nonumber
\end{align}

In contrast to the previous example, this expression was only calculated with FeynGrav because it involved hundreds of terms. The expression itself is too long to be presented in a printed ford. It is published in an open-access repository \cite{latosh2023}. The obtained expression allows one to study the low energy limit of a gravitational scattering of two scalars \cite{Latosh:2023ueg}:
\begin{align}
    \cfrac{d\sigma}{d\Omega} \Bigg|_\text{low energy} = \cfrac14\, \left(G_\text{N} \,\mu\, m_1\, m_2\right)^2 \cfrac{1}{p^4\, \sin^4\frac{\theta}{2}}\left[ 1 + 32\, \pi^2 \,G_\text{N} \, p (m_1+m_2) \sin\cfrac{\theta}{2}  +\mathcal{O}\left( G_\text{N}^2\right)\right].
\end{align}
In this expression for the differential cross section, $p$ is the centre of mass scattering momentum, $\theta$ is the scattering angle, $m_1$ and $m_2$ are masses of the scattered particles, and $\mu$ is the reduced mass.

\section{Summary and further development}\label{Summary_Section}

The paper reviews a recently developed theoretical framework for efficient computation of Feynman rules for perturbative quantum gravity. The perturbative approach to quantum gravity allows one to study quantum gravitational effects within the standard quantum field theory framework. The resulting theory is effective, so it cannot be indefinitely extended in the ultraviolet region and treated similarly to other renormalisable theories. Due to its effective nature, the theory admits infinite interaction terms, all parametrised with a single gravitational coupling. The discussed theoretical framework provides a universal way to compute such interaction rules at any order of perturbation theory.

Perturbative quantum gravity admits the factorisation theorem applicable to any gravity model within Riemann geometry. The theorem states that the action splits into two parts in a broad class of gravity models. One part does not involve derivatives and is constituted of three independent factors that are infinite series. The other part involves derivatives, but it is a finite expression that can be calculated explicitly.

Certain recursive relations exist for the factors generating infinite series. They provide a way to efficiently calculate their contribution at any order of perturbation theory. However, the number of terms in each contribution grows faster than the factorial number of particles involved in the interaction. This feature makes manual calculations of such factors challenging.

FeynGrav provides a tool to operate with the theory at the practical level. It implements the discussed theoretical framework within FeynCalc, which possesses a broad spectrum of tools to operate with expressions within the quantum field theory.

The latest version implements the interaction rules for massless and massive scalars, for scalar field potential interaction, for massive and massless Dirac fermions, for Proca field, for massless vector field, for $SU(N)$ Yang-Mills model, and general relativity. The package also implements auxiliary tools such as graviton polarisation tensors, gauge projectors, and their generalisation for the spin-$2$ case.

The package is rapidly developing and is expected to receive a new update soon. The update will be devoted to the Horndeski gravity, the most general scalar-tensor gravity admitting second-order field equations and the minimal coupling to the matter degrees of freedom.

Further development of the package aims in two directions: to improve the performance and to account for other relevant gravity models. The discovery of the recursive algorithm discussed in this paper significantly increased performance. At the same time, the same relations directly indicate the limits of this existing approach. The number of terms in an $n$ particle interaction vertex grows faster than $n!$, which poses an additional computation challenge. Although such behaviour is expected for an effective theory considered perturbatively, we will strive to improve the computational algorithm even further to push the perturbation theory to its limits.

Consideration of more sophisticated yet relevant models of gravity is another challenge. The Horndeski gravity, the most studied extension of general relativity, is studied in this paper, and its implementation is to enter the next version of FeynGrav. There are few options for other models of gravity that would be relevant to the contemporary research field. The most straightforward generalisation would be massive gravity, but the contemporary observational data imposes strong constraints on the graviton mass. Another possible way to extend FeynGrav is to implement Beyond Horndeski theories, which introduce non-minimal couplings between the scalar field and matter. Lastly, it may be fruitful to consider some supersymmetric expressions and examine whether they can be treated similarly.

In conclusion, the theoretical framework discussed in the paper and the FeynGrav package provide a valuable and efficient tool to study perturbative quantum gravity, with further development aimed at improving performance and incorporating other relevant gravity models.

\section*{Acknowledgment}

This work was supported by the Institute for Basic Science Grant IBS-R018-Y1.

\bibliographystyle{unsrturl}
\bibliography{FGaRPiCPQG.bib}

\end{document}